\documentclass[10pt,aps,prx,twocolumn,floatfix]{revtex4-2}
\usepackage{amsthm}
\usepackage{graphicx}
\usepackage{dcolumn}
\usepackage{bm}
 \usepackage{amsmath}
\usepackage{physics}
\usepackage{xcolor}

\DeclareMathOperator{\arccoth}{arccoth}

\DeclareMathOperator{\diag}{diag}

\newcommand{\mathbbm}[1]{\text{\usefont{U}{bbm}{m}{n}#1}}
\usepackage[caption=false]{subfig}
\usepackage{ragged2e}

\newtheorem{lemma}{Lemma}
\newtheorem{theorem}{Theorem}
\newtheorem{criterion}{Criterion}

\usepackage{indentfirst}
\usepackage{titlesec}
\usepackage{amssymb}
\titlespacing*{\section}{0pt}{*1.2}{*0.8}       
\titlespacing*{\subsection}{0pt}{*1.0}{*0.6}    

\newtheorem{corollary}{Corollary}
\DeclareMathOperator{\sinc}{sinc}

\begin{document}

\newcommand{\papertitle}{Perturbation theory for quantum information quantities of Gaussian states}
\title{\papertitle{}}
  \author{ Zihao Gong
        and Saikat Guha \\
 Department of Electrical and Computer Engineering, \\ University of Maryland, 8228 Paint Branch Drive, College Park, MD 20742
}

\begin{abstract}
Many quantum-information tasks can benefit from Gaussian quantum resources, while their fundamental performance limits are characterized by quantum information quantities. However, a unified perturbation theory for these quantities directly within the Gaussian-state formalism has been lacking. We develop such a theory for the von Neumann entropy, quantum relative entropy, quantum Chernoff quantity, and quantum fidelity, using perturbations of the moments of a reference Gaussian state. Using general Fréchet-derivative identities, we derive compact second-order expressions for reference-state covariance matrices whose symplectic eigenvalues are all greater than $1/2$. We also show that the perturbative quantum Chernoff quantity is maximized at its symmetric point, yielding the quantum Chernoff exponent (QCE). We further show that symplectic eigenvalues equal to $1/2$ can lead to slower asymptotic vanishing, potentially improving the corresponding fundamental performance limits. For reference-state covariance matrices containing symplectic eigenvalues both equal to and greater than $1/2$, we identify sufficient conditions under which the former determine the leading covariance-matrix contributions, allowing the remaining modes to be traced out without changing the leading-order asymptotic scaling or the coefficient. Applying the framework to quantum illumination and incoherent-light sub-Rayleigh object discrimination, we recover known asymptotic QCEs directly within the Gaussian-state formalism. In the latter problem, the slower asymptotic vanishing explains the known quadratic scaling of the QCE with the source separation.
\end{abstract}
  \maketitle

  \section{Introduction}






Gaussian states play a central role in continuous variable quantum information theory \cite{Braunstein2005,WANG2007,weedbrook2012Gaussian,Adesso2014}. They include experimentally accessible quantum states in quantum optics such as coherent, thermal, squeezed, and two-mode squeezed states \cite{Henning2016Squeezing15db,Warwick2002Entanglement}. Despite living in an infinite-dimensional Hilbert space, Gaussian states admit a compact analytical description through their mean vectors and covariance matrices \cite{ferraro2005gaussianstatescontinuousvariable}. This finite-dimensional phase-space representation supports closed-form expressions for many fundamental quantum information quantities, including the von Neumann entropy (VNE) \cite{weedbrook2012Gaussian}, quantum relative entropy (QRE) \cite{MarkQRE,Kaushik2018Renyi}, quantum Chernoff exponent (QCE) \cite{stefanoQCB}, and quantum fidelity (QF) \cite{LeonardoQF}.

These quantum information quantities characterize the fundamental limits across various quantum information processing tasks. The VNE plays a central role in quantum communications \cite{Holevo1998}, where Gaussian states and entanglement have been used to enhance continuous-variable communication protocols \cite{Braunstein2000Densecoding,Haowei2020Com}. In quantum hypothesis testing, the QRE and QCE determine the optimal asymptotic error exponents for asymmetric and symmetric multi-copy state discrimination, respectively \cite{Ogawa2000Steinlemma,Audenaert2007QCB,Bykeli2016QCB}, while Gaussian entanglement can improve discrimination performance in noisy and lossy settings \cite{SiHui2008QI,Saikat2009discrimination}. QRE also characterizes the fundamental limit of quickest quantum change-point detection \cite{Marco2003change,Bash2025Change}, where Gaussian resources can reduce detection latency \cite{Saikat2025change,Zihao2025change}. Finally, QF is related to the quantum Fisher information (QFI) and thus the quantum Cram\'er--Rao bound \cite{Braunstein1994,LeonardoQF,Helstrom1969QFI}, and Gaussian probes can enhance parameter-estimation precision in various active-sensing problems \cite{Caves1981,Zihao2023loss,gardner2026optimal}.

Despite the availability of closed-form expressions, evaluating quantum information quantities of Gaussian states can be analytically cumbersome because they often involve symplectic eigenvalues and the corresponding symplectic eigen-decompositions. When Gaussian states differ only by small perturbations of their moments from those of a {\em reference state}, perturbation theory provides a natural way to obtain the leading-order terms of these quantities directly from the mean and covariances of the reference state and the perturbations thereof. A perturbative framework based on Fréchet derivatives has recently been developed for the VNE, QRE, QCE, and QF \cite{Grace2022}, but it was formulated for finite-dimensional density operators. Applications to Gaussian-state problems therefore must rely on finite-dimensional Fock-space truncations \cite{lu2018quantum,Michael2022identigy} or approximations based on phase-space quasiprobability representations \cite{AliCox2024}. This thereby motivates our perturbative framework formulated directly within the Gaussian-state formalism that is unfettered from any such approximation.

We develop such a framework by expressing the relevant information quantities of Gaussian states as primary matrix functions of the Gaussian-state moments and applying Fréchet expansions, which provide a matrix generalization of the Taylor expansion. In the finite-dimensional density-operator setting \cite{Grace2022}, the matrix functions are evaluated on Hermitian operators, for which Fréchet-derivative representations are provided by the classical Daleckiĭ–Kreĭn theorem \cite{daletskii1965integration}. In contrast, the corresponding Gaussian matrix functions are generally evaluated on non-Hermitian but diagonalizable matrices. We therefore employ generalized divided-difference representations of Fréchet derivatives for diagonalizable matrices \cite{higham2008functions,Rubensson2024}. Building on these representations, we further derive general trace identities for Fréchet derivatives of primary matrix functions, some of which generalize identities developed in \cite{Grace2022}, extending the mathematical tools available for perturbative analysis.

In this paper, building on these mathematical identities, we derive compact second-order perturbative expressions for the VNE, QRE, quantum Chernoff quantity, and QF, as compact functions of the perturbations of the mean vectors and covariance matrices. These expressions apply to reference covariance matrices whose symplectic eigenvalues are all greater than $1/2$. We also show that the perturbative quantum Chernoff quantity is maximized at $s=1/2$, yielding the perturbative QCE. We then identify slower asymptotic vanishing associated with reference-state symplectic eigenvalues equal to $1/2$ for all four information quantities. Furthermore, we establish sufficient conditions under which modes with symplectic eigenvalues greater than $1/2$ can be traced out without changing the leading-order asymptotic coefficients of the covariance-matrix contributions. Finally, we apply our framework to two well-known contemporary use cases: (1) quantum illumination for remote target detection and classification~\cite{AliCox2024} and (2) incoherent-light sub-Rayleigh optical source discrimination ~\cite{krovi2016attainingquantumlimitpassive,lu2018quantum,Michael2022identigy}, recovering known asymptotic QCEs directly within the Gaussian-state formalism. In the latter problem, our aforesaid result on the slower asymptotic vanishing explains the known quadratic scaling of the QCE with source separation. Our framework presented in this paper can be applied to various other problems and use cases, such as multimode squeezing based quantum probes for beam deflectometry~\cite{He2026} and opto-mechanical imaging~\cite{Wilson2025}, Gaussian-entanglement-enhanced quickest change detection of channel parameters~\cite{Saikat2025change,Zihao2025change}, and quantum channel discrimination~\cite{Pirandola2019,Nuradha2025}.

The paper is organized as follows. Sec.~\ref{sec: preliminary} introduces the mathematical preliminaries.
Sec.~\ref{sec: perturbative formula} derives the perturbative formulas and verifies their accuracy. Sec.~\ref{sec:boundary} analyzes slower asymptotic vanishing and develops the perturbative-analysis procedure. Sec.~\ref{sec:applications} applies the framework to quantum illumination and sub-Rayleigh source discrimination.
 Sec.~\ref{sec:conclusion} concludes the paper.

  \section{Preliminaries}\label{sec: preliminary}

\subsection{Notation}
Throughout this paper, a scalar function $f:\mathbb{C}\to\mathbb{C}$ is extended to its corresponding primary matrix function. If matrix $A\in\mathbb{C}^{N\times N}$ is diagonalizable with eigendecomposition
\begin{align}
A = U\Lambda U^{-1},
\qquad
\Lambda = \operatorname{diag}(\lambda_1,\ldots,\lambda_N),
\label{eq}
\end{align}
the corresponding primary matrix function is defined as
\begin{align}
f(A) \!=\! U f(\Lambda) U^{-1},\
f(\Lambda) \!= \!\operatorname{diag}\left(f(\lambda_1),\ldots,f(\lambda_N)\right).
\end{align}
We use the same symbol $f$ for the scalar function and its associated primary matrix function.

We use the standard asymptotic notation $f(x)=O(g(x))$, $x\to x_0$, to mean $\limsup_{x\to x_0}\left|f(x)/g(x) \right|<\infty$, $f(x)=o(g(x))$, $x\to x_0$, to mean
$\lim_{x\to x_0}f(x)/g(x)=0$, and $f(x)=\Theta(g(x))$, $x\to x_0$, to mean that there exist
constants $c_1,c_2>0$ such that
$c_1|g(x)|
\le
|f(x)|
\le
c_2|g(x)|$
for $x$ sufficiently close to $x_0$.
The limiting point $x_0$ will be clear from the context.
The first and second derivatives of the scalar function $f$ are denoted by $f'(x)$ and $f''(x)$, respectively. 

For a matrix $A$, its $kl$-th entry is denoted by $(A)_{kl}$. For a vector $d$, its $k$-th element is denoted by $(d)_k$, whereas $d_k$ denotes the mean vector of the $k$-th state.

  \subsection{Gaussian state formalism}
  Consider an $n$-mode bosonic system with annihilation operators $\left\{\hat{a}\right\}_{k=1}^n$ satisfying $[\hat{a}_k,\hat{a}_l^{\dagger} ] = \delta_{kl}$. We define the quadrature operator vector $\hat{R} = \left( \hat{q}_1,\dots,\hat{q}_n,\hat{p}_1,\dots,\hat{p}_n \right)$, where $\hat{q}_k = \frac{\hat{a}_k+\hat{a}_k^{\dagger}}{\sqrt{2}}$, $\hat{p}_k = \frac{\hat{a}_k-\hat{a}_k^{\dagger}}{\sqrt{2}i} $, and $i$ is the imaginary unit.
  The quadrature operators satisfy the canonical commutation relations
\begin{equation}
[\hat R_k,\hat R_l] = i \Omega_{kl},
\end{equation}
where
\begin{equation}
\Omega =
\begin{pmatrix}
0 & I_n\\
-I_n & 0
\end{pmatrix},
\end{equation}
is the symplectic form and $I_n$ is the $n\times n$ identity matrix.

An $n$-mode Gaussian state $\hat\rho$ of this system is completely characterized by its mean vector $d$ and covariance matrix $\sigma$, whose entries are defined by
\begin{align}
(d)_k &= \langle \hat R_k \rangle_{\hat{\rho}}\\
(\sigma)_{kl}
&=
\frac12
\left\langle
\{\hat R_k-(d)_k,\hat R_l-(d)_l\}
\right\rangle_{\hat{\rho}},
\end{align} 
where $\langle\hat{A}\rangle_{\hat{\rho}} = \trace \left(\hat{A}\hat{\rho}\right)$, and $\{\hat{A},\hat{B}\} = \hat{A}\hat{B}+\hat{B}\hat{A}$. 
Physical states satisfy the uncertainty relation that  $ \sigma + \frac{i}{2}\Omega $ is positive semi-definite.
In this paper, a Gaussian state with mean vector $d_k$ and covariance matrix
$\sigma_k$ is denoted by $\hat{\rho}_k(\sigma_k,d_k)$, or simply
by $\hat{\rho}_k$.

By Williamson's theorem, every  covariance matrix of a Gaussian state admits the symplectic decomposition
\begin{align}
 &\sigma_k = S_k \Lambda_k S_k^T, \label{eq:Williamson decomposition 1}\\
\Lambda_k =
\begin{bmatrix}
D_k & 0\\
0 & D_k
\end{bmatrix}&,
\qquad 
D_k=\diag\left(\nu_1^{(k)},\ldots,\nu_n^{(k)}\right). \label{eq:Williamson decomposition 2}
\end{align}
Here, $S_k$ is a symplectic matrix satisfying $S_k\Omega S_k^T=\Omega$, and $D_k$ contains the symplectic eigenvalues of $\sigma_k$. For a physical covariance matrix, $\nu_l^{(k)}\ge 1/2$. Equivalently, the eigenvalues of $\Omega\sigma_k$ are

$$
\left\{i\nu_1^{(k)},\ldots,i\nu_n^{(k)},-i\nu_1^{(k)},\ldots,-i\nu_n^{(k)}\right\}.
$$
Each eigenvalue pair $\pm i\nu_j$ of $\Omega\sigma$ corresponds to the Gaussian mode with symplectic eigenvalue $\nu_j$ in the Williamson decomposition of $\sigma$.

\subsection{Fr\'echet derivatives}

We briefly introduce the Fr\'echet derivatives and perturbative expansions used throughout this work.

\textit{Definition of Fr\'echet derivatives}---
Let $f$ be a matrix function defined in a neighborhood of $A$. The first Fr\'echet derivative of $f$ at $A$ with perturbation $E$ is the unique linear map $D^{(1)}_{f(A)}[E]$ satisfying  \cite{bhatia2013matrix,higham2008functions}
\begin{align}
\lim_{\|E\|\to 0}
\frac{
\|f(A+E)-f(A)- D_{f(A)}^{(1)}[E]\|
}{
\|E\|
}
=0,
\end{align}
where $\| \cdot \|$ denotes a matrix norm. The $k$th Fr\'echet derivative of $f$ at $A$ with perturbation $ [E_1,\dots,E_k] $ can be defined recursively as the unique multilinear function $D^{(k)}_{f(A)}[E_1,\dots,E_k]$ such that \cite{higham2014higher}
\begin{align}
    &D^{(k-1)}_{f(A+E_k)}[E_1,\dots,E_{k-1}] - D^{(k-1)}_{f(A)}[E_1,\dots,E_{k-1}] \nonumber \\
    &\qquad \qquad - D^{(k)}_{f(A)}[E_1,\dots,E_k] = o(\|E_k\|).
\end{align}
When all perturbations are identical, i.e., $E_k = E, \forall k$,  we use the abbreviation
\begin{align}
D^{(k)}_{f(A)}[E]
:=
D^{(k)}_{f(A)}[E,\ldots,E].
\end{align}

\textit{G\^ateaux-derivative representation}---
For a sufficiently smooth function $f$, the  Fr\'echet derivatives can equivalently be expressed through directional derivatives as \cite{higham2014higher}
\begin{align}
   &D^{(k)}_{f(A)}[E_1,\dots,E_k] \nonumber \\ 
&\qquad=
\left.
\frac{\partial^k}{\partial t_1\cdots\partial t_k}
f\left(A+\sum_{j=1}^k t_jE_j\right)
\right|_{t_1=\cdots=t_k=0}. \label{eq: directional definition}
\end{align}

\textit{Taylor expansion in Fr\'echet derivatives}---
The Fr\'echet derivatives provide a matrix analogue of the Taylor expansion and underpin the perturbative analysis developed in Sec.~III. For sufficiently smooth $f$,
up to order $K$, the Fr\'echet expansion is given by \cite[Th.~1]{Higham2010}
\begin{equation}
    f(A+E)
    =
    f(A)
    +
    \sum_{k=1}^{K}
    \frac{1}{k!}
    D_{f(A)}^{(k)}[E]
    +
    R_K(A,E),\label{eq:Frechet expansion}
\end{equation}
where $R_K(A,E)$ denotes the higher-order remainder.

\textit{Divided-difference representation}---
Fr\'echet derivatives of matrix functions can be expressed explicitly in terms of divided differences of the underlying scalar function. For Hermitian matrices, such representations are provided by the classical Dalecki\u{\i}--Kre\u{\i}n (DK) theorem \cite{daletskii1965integration}. For diagonalizable but non-Hermitian matrices, the DK theorem is extended to yield an explicit divided-difference representation of the first Fr\'echet derivative \cite[Corollary~3.12]{higham2008functions}, while conditions for the existence and continuity of higher-order Fr\'echet derivatives are established in \cite[Th.~3.5]{higham2014higher}. More generally, explicit divided-difference representations for higher-order Fr\'echet derivatives of an $N$-by-$N$ diagonalizable matrix are derived in \cite[Eq.~4.1]{Rubensson2024}, under the sufficient condition that $f$ is $(k+1)N-1$ times continuously differentiable for the existence and continuity of the $k$-th Fr\'echet derivative \cite[Th.~3.1]{Rubensson2024}. 

Since the matrix functions arising in the Gaussian-state formalism are generally evaluated at non-Hermitian but diagonalizable matrices, these generalized representations are applicable to the perturbative analysis developed in this work. Furthermore, since all matrix functions considered in this work are analytic on the relevant spectral domains, the required  continuity conditions are automatically satisfied. We therefore do not emphasize these conditions in the subsequent derivations. We summarize below the first- and second-order divided-difference representations that will be used throughout the paper.

\textit{First- and second-order formulas}
---
For a diagonal matrix $\Lambda = \operatorname{diag}\left(\lambda_1,\dots,\lambda_N\right)$, 
the first divided difference form of $f$ with respect to $ \Lambda $ is defined as 
\begin{align}
     & f^{[1]}[\lambda_k,\lambda_l]  =  
 \begin{cases}
 \dfrac{f(\lambda_k)-f(\lambda_l)}{\lambda_k-\lambda_l}, & \lambda_k\neq \lambda_l\\[6pt]
  f'(\lambda_k), & \lambda_k =\lambda_l
\end{cases}.
\label{eq:first_dd}
\end{align}
The second divided difference of $f$ with respect to $ \Lambda $ is defined recursively by
\begin{align}
      &f^{[2]}[\lambda_k,\lambda_l,\lambda_m]\nonumber\\
    &\qquad= \begin{cases}
       \frac{f^{[1]}[\lambda_k,\lambda_l] -f^{[1]}[\lambda_l,\lambda_m] }{\lambda_k - \lambda_m}
       &\ \lambda_k\ne\lambda_m \\
        \frac{f^{\prime}(\lambda_k) -f^{[1]}[\lambda_k,\lambda_l] }{\lambda_k - \lambda_l} &\ \lambda_k = \lambda_m \ne \lambda_l \\
        \frac{1}{2} f^{\prime \prime}\left( \lambda_k \right)  &\ \lambda_k=\lambda_m =\lambda_l
   \end{cases}. \label{eq:second_dd}
\end{align}
Let $A$ be diagonalizable, with eigendecomposition $A = U \Lambda U^{-1}$ and 
define
\begin{align}
    \tilde{E} = U^{-1}EU.
\end{align}
The first Fr\'echet derivative of $f$ admits the divided-difference
representation
\begin{align}
D^{(1)}_{f(A)}[E] =  U
    \left(
    f^{[1]}[\Lambda]\circ \tilde{E}
    \right)
    U^{-1}, \label{eq:First Frechet derivative Hadmard product}
\end{align}
where $\circ$ denotes the Hadamard product (entry-wise product) and the $kl$-th entry of $f^{[1]}[\Lambda]$ yields
\begin{align}
    \left(f^{[1]}[\Lambda]\right)_{kl}
    =
    f^{[1]}[\lambda_k,\lambda_l].
\end{align}
Similarly, the second Fr\'echet derivative is given by
\begin{align}
    D^{(2)}_{f(A)}[E]
    =
    2U M U^{-1},
\end{align}
where the $km$-th entry of matrix $M$ is
\begin{align}
    (M)_{km}
    =
    \sum_{l=1}^{N}
    f^{[2]}[\lambda_k,\lambda_l,\lambda_m]\,
    (\tilde{E})_{kl}(\tilde{E})_{lm}.
\end{align}
\textit{Spectral-projector form}---For the derivations in the appendices, it is often convenient to use an equivalent spectral-projector representation. Define
\begin{align}
    P_j=Ue_j e_j^\dagger U^{-1},
\end{align}
where  $e_j\in\mathbb{C}^N$ denote the $j$-th standard basis vector, whose $j$-th entry is one and all other entries are zero. 
Then
\begin{align}
    A=\sum_{j=1}^N \lambda_j P_j,
\quad
I_N=\sum_{j=1}^N P_j,
\quad
P_jP_k=\delta_{jk}P_j. \label{eq:A = P}
\end{align}
Using the spectral projectors, the first Fr\'echet derivative can equivalently be written as
\begin{align}
    D^{(1)}_{f(A)}[E]
=
\sum_{j,k=1}^N
f^{[1]}[\lambda_j,\lambda_k]P_jEP_k. \label{eq:D1f P}
\end{align}
Similarly, the second Fr\'echet derivative takes the form
\begin{align}
    D^{(2)}_{f(A)}[E]
=
2\sum_{j,k,l=1}^N
f^{[2]}[\lambda_j,\lambda_k,\lambda_l] 
P_jEP_kEP_l. \label{eq:D2f P}
\end{align}
The eigenbasis and spectral-projector forms are equivalent; we use whichever representation is more convenient.

\section{Perturbative formulas for information quantities of Gaussian states} \label{sec: perturbative formula}
In this section, we derive second-order perturbative formulas for four
Gaussian quantum information quantities. Throughout this section, we assume that all symplectic eigenvalues of the reference covariance matrix are greater than $1/2$, ensuring that the matrix functions and their derivatives appearing in the perturbative expansions are regular. The case involving symplectic eigenvalues equal to $1/2$ is analyzed in Sec.~\ref{sec:boundary}.
\subsection{Von Neumann entropy}
For a general quantum state $\hat{\varrho}$, the VNE is defined as
\begin{align}
\mathcal{S}(\hat{\varrho})
=
-\trace \left(\hat{\varrho}\log\hat{\varrho}\right).
\end{align}
The VNE of an $n$-mode Gaussian state $\hat{\rho}(\sigma,d)$ depends only on
the symplectic eigenvalues $\{\nu_k\}_{k=1}^n$ of the covariance matrix $\sigma$ \cite{weedbrook2012Gaussian}:
\begin{equation}
\mathcal{S}(\hat{\rho})
=
\sum_{k=1}^n
g(\nu_k),
\label{eq:SvN}
\end{equation}
where
\begin{equation}
g(x)
=
\left(x+\tfrac12\right)
\log\left(x+\tfrac12\right)
-
\left(x-\tfrac12\right)
\log\left(x-\tfrac12\right).\label{eq:g}
\end{equation}
To express the VNE in a form suitable for Fr\'echet expansion, we use the fact that the eigenvalues of $\Omega\sigma$ are $\{\pm i\nu_k\}_{k=1}^n$. Define the scalar function
\begin{align}
f_{\rm VN}(x)=g\left(\sqrt{-x^2}\right).
\end{align}
For $\sigma = \sigma_0 + \Delta\sigma$, define $A = \Omega \sigma_0$ and $E = \Omega\Delta \sigma$.
The VNE of $\hat{\rho}$ can be written as
\begin{align}
\mathcal{S}(\hat{\rho})
=
\frac{1}{2}
\trace\left[
f_{\rm VN}(A+E)
\right].
\end{align}

\begin{theorem}[Perturbative expansion of the von Neumann entropy]
\label{thm:VNE-perturbation}
Let $\hat{\rho}$ be the Gaussian state introduced above,  with reference
covariance matrix $\sigma_0$ whose symplectic eigenvalues are all
greater than $1/2$. Then
the second-order perturbative approximation is
\begin{align}
\mathcal{S}_{\rm pert}^{(2)}(\hat{\rho})
&=
\frac{1}{2}
\trace\left[
f_{\rm VN}(A)
\right]
+
\frac{1}{2}
\trace\left[
E f_{\rm VN}'(A)
\right]
\nonumber \\
&\qquad \qquad
+
\frac{1}{4}
\trace\left[
E
D_{f_{\rm VN}'(A)}^{(1)}[E]
\right]. \label{eq:result VNE}
\end{align}
\end{theorem}

The VNE depends only on the covariance matrix and is independent of
the mean vector. Its perturbative formula follows directly from the
Fr\'echet expansion of $f_{\rm VN}$ around $A$ in the direction $E$.

\subsection{Quantum relative entropy}

The QRE between two quantum states
$\hat{\varrho}_1$ and $\hat{\varrho}_2$ is defined as
\begin{align}
D(\hat{\varrho}_2\|\hat{\varrho}_1)
=
\trace\left[
\hat{\varrho}_2
\left(
\log\hat{\varrho}_2-\log\hat{\varrho}_1
\right)
\right].
\end{align}
For two $n$-mode Gaussian states
$\hat{\rho}_1(\sigma_1,d_1)$ and
$\hat{\rho}_2(\sigma_2,d_2)$, the QRE is given by \cite{MarkQRE}
\begin{align}
\!D(\hat{\rho}_2\|\hat{\rho}_1)\!=\!
\frac{1}{2}
\left[
\log\left(\frac{Z_1}{Z_2}\right)
-
\trace\left(\Delta G \sigma_2\right)
+
\delta_d^T G_1\delta_d
\right],
\label{eq:QRE general}
\end{align}
where
$Z_k=\det\left(\sigma_k+\frac{i}{2}\Omega\right)$,
$\Delta G  =G_2-G_1$,
$G_k=2i\Omega\,\operatorname{arccoth}\left(2i\sigma_k\Omega\right)$,
and $\delta_d=d_1-d_2$.

To express the Gaussian QRE in a form suitable for Fr\'echet expansion, let
\begin{align}
\sigma_k=\sigma_0+\Delta\sigma_k,
\
A=\Omega\sigma_0,
\
E_k=\Omega\Delta\sigma_k,
\ k=1,2, \label{eq:perturbtive setup}
\end{align}
and define
\begin{align}
f_{\rm RE}(x)
=
2i\arccoth(2ix). \label{eq:fRE}
\end{align}
With these definitions, the Gaussian QRE can be rewritten entirely in terms of matrix functions of $A+E_k$ as
\begin{align}
D(\hat{\rho}_2\|\hat{\rho}_1)
&=
\frac{1}{2}
\left[
\log\left(
\frac{
\det\left(A+E_1-\frac{i}{2}I\right)}
{
\det\left(A+E_2-\frac{i}{2}I\right)}
\right)
\right.
\nonumber\\
&\
+\trace\left[
\left(
f_{\rm RE}(A+E_1)
-
f_{\rm RE}(A+E_2)
\right)
(A+E_2)
\right]
\nonumber\\
&\ \left.
+
\delta_d^T
f_{\rm RE}(A+E_1)
\Omega\delta_d
\right].
\label{eq:QRE A E}
\end{align}
The derivation of \eqref{eq:QRE A E} from the standard Gaussian QRE expression in
\eqref{eq:QRE general} is provided in
Appendix~\ref{appendix sub:QRE CM}.

\begin{theorem}[Perturbative expansion of the QRE]
\label{thm:QRE-perturbation}
Let $\hat{\rho}_1$ and $\hat{\rho}_2$  be  the two Gaussian states introduced above,  with reference
covariance matrix $\sigma_0$ whose symplectic eigenvalues are all
greater than $1/2$.
Then the covariance-matrix contribution vanishes at zeroth and first-order in the Fr\'echet expansion, and the second-order perturbative approximation is
\begin{align}
D_{\rm pert}^{(2)}(\hat{\rho}_2\|\hat{\rho}_1)
&=
-\frac{1}{4}
\trace\left[
(E_1-E_2)
D^{(1)}_{f_{\rm RE}(A)}[E_1-E_2]
\right]
\nonumber\\
&\quad+
\frac{1}{2}
\delta_d^T f_{\rm RE}(A)\Omega\delta_d.
\label{eq:QRE}
\end{align}
\end{theorem}
The Fr\'echet expansion and the simplification leading to
Theorem~\ref{thm:QRE-perturbation} are provided in
Appendix~\ref{appendix:gaussian_qi frechet QRE}.

\subsection{Quantum Chernoff quantity}
For two quantum states $\hat\varrho_1$ and $\hat\varrho_2$, the quantum Chernoff quantity at parameter $s\in[0,1]$ is defined as
\begin{align}
\xi_s(\hat{\varrho}_1,\hat{\varrho}_2)
=
-\log\left(
\trace\left[
\hat{\varrho}_1^s\hat{\varrho}_2^{1-s}
\right]
\right),
\qquad
s\in[0,1].
\end{align}
The QCE is obtained by maximizing the quantum Chernoff
quantity over $s$:
\begin{align}
\xi(\hat{\varrho}_1,\hat{\varrho}_2)
=
\max_{s\in[0,1]}
\xi_s(\hat{\varrho}_1,\hat{\varrho}_2).
\end{align}
For two $n$-mode Gaussian states
$\hat{\rho}_1(\sigma_1,d_1)$ and
$\hat{\rho}_2(\sigma_2,d_2)$,
the quantum Chernoff quantity admits a closed-form expression in terms
of their symplectic eigenvalues and mean vectors \cite{stefanoQCB}.
Ref.~\cite{stefanoQCB} adopts a covariance-matrix convention in which a pure Gaussian state has symplectic eigenvalue $1$, whereas in our convention the corresponding symplectic eigenvalue is $1/2$. We therefore rewrite the Gaussian quantum Chernoff quantity formula in our convention as
\begin{align}
\xi_s(\hat{\rho}_1,\hat{\rho}_2)
=
-\log\left(
Q_s^{(1)}Q_s^{(2)}
\right),
\label{eq:QCE}
\end{align}
where
\begin{align}
Q_s^{(1)}
&=
\frac{
2^n
\prod_{k=1}^n
G^{(s)}(\nu_k^{(1)})
G^{(1-s)}(\nu_k^{(2)})
}{
\sqrt{
\det\left[
V_1(s)+V_2(1-s)
\right]
}
},
\label{eq:QCE 1}
\\
Q_s^{(2)}
&=
\exp\left[
- 
\delta_d^T
\left(
V_1(s)+V_2(1-s)
\right)^{-1}
\delta_d
\right],
\\
G^{(s)}(x)
&=
\frac{1}{
\left(x+\frac{1}{2}\right)^s
-
\left(x-\frac{1}{2}\right)^s
},
\label{eq: QCE Gs}
\\
F^{(s)}(x)
&=
\frac{
\left(x+\frac{1}{2}\right)^s
+
\left(x-\frac{1}{2}\right)^s
}{
\left(x+\frac{1}{2}\right)^s
-
\left(x-\frac{1}{2}\right)^s
},
\label{eq: QCE Fs}
\\
\delta_d
&=
d_1-d_2,
\\
V_1(s)
&=
S_1F^{(s)}(\Lambda_1)S_1^T \label{eq: QCE V_1} \\
V_2(1-s)
&=
S_2F^{(1-s)}(\Lambda_2)S_2^T.
\label{eq: QCE V_2}
\end{align}
To express the Gaussian quantum Chernoff quantity in a form suitable for Fr\'echet expansion, we introduce the functions
\begin{align}
f_{\rm CE}^{(s)}(A)
&=
G^{(s)}\left(\sqrt{-A^2}\right),
\\
g_{\rm CE}^{(s)}(A)
&=
-A
\left(\sqrt{-A^2}\right)^{-1}
F^{(s)}\left(\sqrt{-A^2}\right), \label{eq:QCE gs}
\end{align}
and define $\sigma_k$, $A$, and $E_k$ as in \eqref{eq:perturbtive setup}.
Then the covariance-matrix contribution to the quantum Chernoff
quantity can be written as
\begin{align}
-\log Q_s^{(1)}
&=\frac{1}{2}
\log\det\left[
g_{\rm CE}^{(s)}(A+E_1)
+
g_{\rm CE}^{(1-s)}(A+E_2)
\right]
\nonumber\\
&\quad
-\frac{1}{2}
\log\det\left[
f_{\rm CE}^{(1-s)}(A+E_2)
\right]
\nonumber\\
&\quad
-\frac{1}{2}
\log\det\left[
f_{\rm CE}^{(s)}(A+E_1)
\right]-n\log 2,
\label{eq:QCE A E CM}
\end{align}
while the mean-vector contribution is
\begin{align}
-\log Q_s^{(2)}
&=
\delta_d^T
\left[
\Omega
\left(
g_{\rm CE}^{(s)}(A+E_1)
\right.\right.\nonumber\\
&\left.\left.\qquad \qquad+
g_{\rm CE}^{(1-s)}(A+E_2)
\right)
\right]^{-1}
\delta_d.
\label{eq:QCE A E d}
\end{align}
The derivation of \eqref{eq:QCE A E CM} and
\eqref{eq:QCE A E d} is
provided in Appendix~\ref{appendix:sub gaussian AE form}. 
For convenience, define
\begin{align}
    R_{\rm CE}^{(s)}(A) = g_{\rm CE}^{(s)} +  g_{\rm CE}^{(1-s)}.
\end{align}

\begin{theorem}[Perturbative expansion of the quantum Chernoff quantity]
\label{thm:QCE-perturbation}
Let $\hat{\rho}_1$ and $\hat{\rho}_2$  be  the two Gaussian states introduced above,  with reference
covariance matrix $\sigma_0$ whose symplectic eigenvalues are all
greater than $1/2$. Then the covariance-matrix contribution vanishes at zeroth and first-order in the Fr\'echet expansion, and the second-order perturbative quantum Chernoff quantity is
\begin{align}
\xi_{{\rm pert},s}^{(2)}
(\hat{\rho}_1,\hat{\rho}_2)
=
\xi_{{\rm CE,CM},s}^{(2)}
+
\xi_{{\rm CE},d,s},
\label{eq:QCE perturb final}
\end{align}
where
\begin{align}
\xi_{{\rm CE,CM},s}^{(2)}
&=
\frac{1}{4}
\trace\Big[
\left( 1/R_{\rm CE}^{(s)}\right)(A)
D_{g_{\rm CE}^{(s)}(A)}^{(1)}[E_1-E_2]
\nonumber\\
&\quad \times
\left(1/R_{\rm CE}^{(s)}\right)(A)
D_{g_{\rm CE}^{(1-s)}(A)}^{(1)}[E_1-E_2]
\Big],
\label{eq:QCE perturb CM}
\end{align}
and
\begin{align}
\xi_{{\rm CE},d,s}
=
 - \delta_d^T
\left(1/ R_{\rm CE}^{(s)}\right)(A) \Omega
\delta_d.
\label{eq:QCE perturb d}
\end{align}
\end{theorem}
The Fr\'echet expansion and the simplification leading to
Theorem~\ref{thm:QCE-perturbation} are provided in
Appendix~\ref{appendix:gaussian_qi frechet QCE}.
\begin{corollary}[Perturbative quantum Chernoff exponent]
\label{cor:QCE-s-half}
The second-order perturbative quantum Chernoff quantity in Theorem \ref{thm:QCE-perturbation} is maximized at
$s=\frac{1}{2}$. Therefore, evaluating the perturbative quantity at $s=1/2$ yields the perturbative QCE:
\begin{align}
\xi_{\rm QCE}^{(2)}
=
\xi_{{\rm pert},1/2}^{(2)}
\label{eq:QCE s half}
\end{align}
\end{corollary}
For finite-dimensional quantum states, the second-order perturbative quantum Chernoff quantity was also shown to be maximized at
$s=\frac{1}{2}$ \cite{Grace2022}. Corollary~\ref{cor:QCE-s-half} establishes
the corresponding result for Gaussian states. The proof is provided in
Appendix~\ref{app:QCE concavity}.

\subsection{Quantum fidelity}

The QF between two quantum states $\hat{\varrho}_1$ and
$\hat{\varrho}_2$ is defined as
\begin{align}
F(\hat{\varrho}_1,\hat{\varrho}_2)
=
\trace
\sqrt{
\sqrt{\hat{\varrho}_1}\,
\hat{\varrho}_2\,
\sqrt{\hat{\varrho}_1}
}.
\end{align}
For two $n$-mode Gaussian states
$\hat{\rho}_1(\sigma_1,d_1)$ and
$\hat{\rho}_2(\sigma_2,d_2)$, the QF can be written as
\cite{LeonardoQF}
\begin{align}
F(\hat{\rho}_1,\hat{\rho}_2)
=
F_0
\exp\left[
-\frac{1}{4}
\delta_d^{T}
(\sigma_1+\sigma_2)^{-1}
\delta_d
\right],
\label{eq: QF}
\end{align}
where $\delta_d=d_1-d_2$. Here, $F_0$ contains the covariance-matrix contribution and is given by 
\begin{align}
F_0^4
=
\frac{
\det
\left[
\left(
\sqrt{4\mathbb{I}+(V_{\rm aux}\Omega)^{-2}}
+2\mathbb{I}
\right)
V_{\rm aux}
\right]
}{
\det(\sigma_1+\sigma_2)
},
\label{eq: QF F_0}
\end{align}
where
\begin{align}
V_{\rm aux}
=
\Omega^T
(\sigma_1+\sigma_2)^{-1}
\left(
\frac{\Omega}{4}
+
\sigma_2\Omega\sigma_1
\right).
\label{eq: QF V aux}
\end{align}
To apply the perturbative expansion, define $\sigma_k$, $A$, and $E_k$ as in \eqref{eq:perturbtive setup}, and define
\begin{align}
f_{\rm F}(x)
=
\left(
\sqrt{4+x^{-2}}+2
\right)x.
\end{align}
Then the covariance-matrix contribution can be expressed as
\begin{align}
F_0^4
=
\frac{
\det\left[f_{\rm F}(\Omega V_{\rm aux})\right]
}{
\det(2A+E_1+E_2)
},
\label{eq: QF CM rewritten}
\end{align}
where
\begin{align}
\Omega V_{\rm aux}
&=
(2A+E_1+E_2)^{-1} \nonumber\\
&\qquad \times
\left[
A^2-\frac{I}{4}+E_2A+AE_1+E_2E_1
\right].
\label{eq: QF Vaux main}
\end{align}
The mean-vector contribution can be written as
\begin{align}
\exp\left[
-\frac{1}{4}
\delta_d^T
(2A+E_1+E_2)^{-1}
\Omega\delta_d
\right].
\label{eq: QF mean rewritten}
\end{align}
The detailed derivation of Eqs.~\eqref{eq: QF CM rewritten}--\eqref{eq: QF mean rewritten} is provided in Appendix~\ref{appendix:sub gaussian_qi AE}.

\begin{theorem}[Perturbative expansion of the QF]
\label{thm:QF-perturbation}
Let $\hat{\rho}_1$ and $\hat{\rho}_2$  be  the two Gaussian states introduced above,  with reference
covariance matrix $\sigma_0$ whose symplectic eigenvalues are all
greater than $1/2$. Then the covariance-matrix contribution is unity at zeroth order and vanishes at first order in the Fr\'echet expansion. The second-order perturbative QF is
\begin{align}
F_{\rm pert}(\hat{\rho}_1,\hat{\rho}_2)
=
1
&-
\frac{1}{4}
\trace\left[
X(E_1-E_2)
\right] \nonumber\\
&\qquad\qquad
-
\frac{1}{4}
\delta_d^T
(2A)^{-1}
\Omega
\delta_d,
\label{eq: QF pert final}
\end{align}
where $X$ is the solution of
\begin{align}
E_1-E_2
=
X+4AXA.
\label{eq: QF X equation main}
\end{align}
\end{theorem}
The first-order cancellation and
second-order simplification are provided in
Appendix~\ref{appendix:gaussian_qi frechet QF}.


QF is related to the QFI through the Bures metric \cite{Braunstein1994,LeonardoQF}. 
This relation is reflected in our perturbative QF formula. In particular, the covariance-matrix contribution in \eqref{eq: QF pert final} is determined by the matrix $X$ in \eqref{eq: QF X equation main}, whose Lyapunov-type equation has the same structure as that appearing in Gaussian QFI formulas \cite{monras2013phase}. 

For numerical evaluation, the Lyapunov-type equation for $X$ can be converted into a linear system using the vectorization identity
\begin{align}
\operatorname{vec}(AXB)
=
(B^T\otimes A)\operatorname{vec}(X),
\end{align}
which gives
\begin{align}
\operatorname{vec}(X)
=
\left[
I+4(A^T\otimes A)
\right]^{-1}
\operatorname{vec}(E_1-E_2).
\label{eq:QF X vectorized}
\end{align}
Thus, $X$ can be obtained by \eqref{eq:QF X vectorized} and substituted directly into the perturbative fidelity formula \eqref{eq: QF pert final}.

\textit{Summary}---These formulas evaluate the perturbative information quantities directly from the reference-state moments and their perturbations, without requiring the symplectic decomposition of the perturbed covariance matrices.

\subsection{Numerical verification}\label{subsec:numerical verification}
\begin{figure}[ht]
    \centering
    \includegraphics[width=0.9\linewidth]{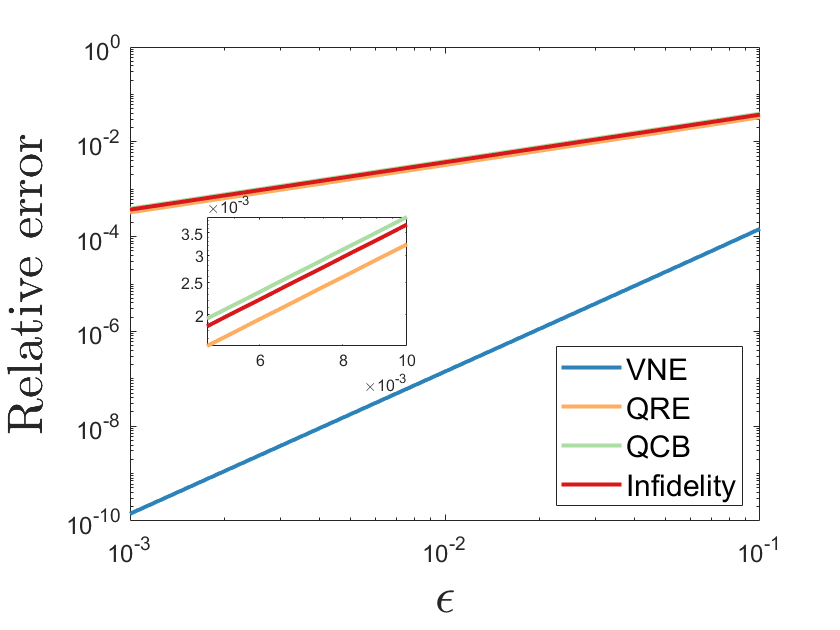} 
        \caption{Relative error $R_Q=|Q-Q_{\rm pert}^{(2)}|/|Q|$ of the second-order perturbative approximations as a function of the covariance-matrix perturbation norm $\epsilon$. Each point is averaged over $10^5$ independent random perturbation realizations. The VNE exhibits $R_Q=O(\epsilon^3)$, while the QRE, QCE, and quantum infidelity $1-F$ exhibit $R_Q=O(\epsilon)$, consistent with the predicted perturbative approximation errors.
}\label{fig:1} 
\end{figure}

In this subsection, we numerically verify the perturbative expressions derived above. To isolate the accuracy of the covariance-matrix contributions, we set all mean vectors to zero. We choose the reference covariance matrix as
\begin{align}
\sigma_0 = S D S^T,
\end{align}
where
\begin{align}
D = \diag(0.8,1.2,0.8,1.2),
\end{align}
and
\begin{align}
S=
\begin{bmatrix}
\cos\theta & -\sin\theta & 0 & 0\\
\sin\theta & \cos\theta & 0 & 0\\
0 & 0 & \cos\theta & -\sin\theta\\
0 & 0 & \sin\theta & \cos\theta
\end{bmatrix},
\end{align}
with $\theta=0.4$.
For the VNE, we consider the perturbed covariance matrix
\begin{align}
\sigma=\sigma_0+\epsilon\Delta\sigma,
\end{align}
whereas for the QRE, QCE, and QF, we consider
\begin{align}
\sigma_k=\sigma_0+\epsilon\Delta\sigma_k,
\qquad k=1,2.
\end{align}
We take the perturbation matrices $\Delta\sigma$ and $\Delta\sigma_k$ to be symmetric, with their entries generated from independent  zero-mean Gaussian random variables $\mathcal N(0,1)$. They are subsequently normalized such that
\begin{align}
\|\Delta\sigma\|=1,
\qquad
\|\Delta\sigma_1-\Delta\sigma_2\|=1,
\end{align}
where $\|\cdot\|$ denotes the operator norm.
Thus, $\epsilon$ directly specifies the norm of the covariance-matrix perturbation for the VNE and the norm of the covariance-matrix difference for the QRE, QCE, and QF.

To quantify the accuracy of the perturbative formulas, we define the relative error
\begin{align}
R_Q
=
\frac{
\left|Q-Q_{\rm pert}^{(2)}\right|
}{
|Q|
},
\end{align}
where $Q$ denotes the corresponding exact quantum information quantity, with
\begin{align}
    Q \in
\left\{
\mathcal{S}(\hat{\rho}),
D(\hat{\rho}_2\|\hat{\rho}_1),
\xi(\hat{\rho}_1,\hat{\rho}_2),
1-F(\hat{\rho}_1,\hat{\rho}_2)
\right\}.
\end{align}
For the QF, both the exact and perturbative quantities are evaluated in terms of the quantum infidelity $1-F(\hat{\rho}_1,\hat{\rho}_2)$. Each point in Fig.~\ref{fig:1} is obtained by averaging over $10^5$ independent random perturbation realizations.

For the VNE, the perturbative approximation error scales as
$\left|
Q-Q_{\rm pert}^{(2)}
\right|
=
O(\epsilon^3)$ according to \eqref{eq:Frechet expansion}, 
while $Q=O(1)$. Therefore,
$R_Q=O(\epsilon^3)$,
consistent with the slope of approximately  $3$ observed in Fig.~\ref{fig:1}. For the QRE, QCE, and quantum infidelity, the exact quantities scale as $O(\epsilon^2)$, whereas the perturbative approximation errors remain $O(\epsilon^3)$. Hence,
$R_Q=O(\epsilon)$,
consistent with the slopes close to $1$ observed in Fig.~\ref{fig:1}.

\section{Analytical framework for perturbative covariance-matrix analysis}
\label{sec:boundary}
In Sec.~\ref{sec: perturbative formula}, we derived perturbative formulas for reference covariance matrices whose symplectic eigenvalues are all greater than $1/2$. When some symplectic eigenvalues approach $1/2$, the perturbative coefficients can become singular, leading to slower asymptotic vanishing. In this section, we first characterize this behavior for reference states approaching purity. We then consider reference covariance matrices with symplectic eigenvalues both greater than and equal to $1/2$, and identify sufficient conditions under which the modes with symplectic eigenvalues greater than $1/2$ can be traced out without changing the leading perturbative contributions. Finally, we summarize a procedure for applying these results and verify the predicted scaling numerically.

\subsection{Slower asymptotic vanishing for pure reference states}
\label{subsec: pure scaling}
To characterize this singular behavior, we first consider an
$n$-mode zero-mean Gaussian reference state with covariance matrix
\begin{align}
\sigma_{0} =
\left(
\frac{1}{2}+\tau
\right)I_{2n},
\qquad
\tau\rightarrow 0,
\label{eq:singularity sigma0}
\end{align}
where $\tau$ is an auxiliary parameter.  We assume that the covariance-matrix perturbation
for the VNE satisfies $\|E\|=\Theta(\epsilon)$, with
$\epsilon\rightarrow0$. For the QRE, quantum Chernoff quantity, and QF,
the perturbations $E_1$ and $E_2$ satisfy
$\|E_1-E_2\|=\Theta(\epsilon)$. We first keep $\epsilon$ and $\tau$
independent and determine the joint scaling of the perturbative terms.

For the VNE, the first- and second-order covariance-matrix
contributions scale as
\begin{align}
S_{\rm VN}^{(1)}
&=
O\left(
\epsilon |\log\tau|
\right),
&
S_{\rm VN}^{(2)}
&=
O\left(
\epsilon^2\tau^{-1}
\right).
\label{eq:VN singular scaling}
\end{align}
For the QRE, quantum Chernoff quantity, and QF, the leading
covariance-matrix contributions satisfy
\begin{align}
D_{\rm RE,CM}^{(2)},
\quad
\xi_{{\rm CE,CM},s}^{(2)},
\quad
1-F_{\rm pert}
=
O\left(
\epsilon^2\tau^{-1}
\right).
\label{eq:singularity common scaling}
\end{align}
The detailed derivations of
Eqs.~\eqref{eq:VN singular scaling}--\eqref{eq:singularity common scaling}
are provided in Appendix~\ref{app:singularity}.

These singular coefficients do not imply that the underlying quantum
information quantities diverge. Instead, they indicate that the
asymptotic scaling with the covariance-matrix perturbation changes when the symplectic eigenvalues approach $1/2$. In accordance with the basic premise of perturbation theory, the auxiliary reference state should remain asymptotically close to the perturbed state. For a covariance-matrix perturbation of order $\Theta(\epsilon)$ acting on the modes with eigenvalues $1/2$, we therefore take
\begin{align}
\tau = \Theta(\epsilon).
\label{eq:tau epsilon scaling}
\end{align}
Substituting \eqref{eq:tau epsilon scaling} into
\eqref{eq:VN singular scaling} and
\eqref{eq:singularity common scaling}, the VNE contributions change
from
\begin{align}
S_{\rm VN}^{(1)} &= O(\epsilon),
&
S_{\rm VN}^{(2)} &= O(\epsilon^2),
\end{align}
when all the reference state symplectic eigenvalues are greater than $1/2$ to
\begin{align}
S_{\rm VN}^{(1)} &= O(\epsilon|\log\epsilon|),
&
S_{\rm VN}^{(2)} &= O(\epsilon),
\end{align}
for a pure reference state. Likewise, the covariance-matrix contributions to the
QRE, quantum Chernoff quantity, and quantum infidelity change from
\begin{align}
D_{\rm RE,CM}^{(2)},
\quad
\xi_{{\rm CE,CM},s}^{(2)},
\quad
1-F_{\rm pert}
=
O(\epsilon^2)
\end{align}
when all the reference state symplectic eigenvalues are greater than $1/2$ to
\begin{align}
D_{\rm RE,CM}^{(2)},
\quad
\xi_{{\rm CE,CM},s}^{(2)},
\quad
1-F_{\rm pert}
=
O(\epsilon)
\end{align}
for a pure reference state. Thus, reference covariance matrices with symplectic eigenvalues equal to $1/2$ exhibit a slower vanishing rate under small covariance-matrix perturbations than those whose symplectic eigenvalues are greater than $1/2$.

This slower asymptotic vanishing can improve the fundamental limits of quantum-information tasks governed by these quantities. In particular, a nonzero leading contribution to the QCE proportional to $\epsilon$, rather than $\epsilon^2$, yields a larger asymptotic discrimination exponent for sufficiently small $\epsilon$. This behavior is illustrated by the sub-Rayleigh optical source discrimination problem in Sec.~\ref{subsec:application 2}.

\subsection{Leading perturbations with symplectic eigenvalues both greater than and equal to $1/2$}
\label{subsec:leading pure modes}

We considered reference covariance matrices whose symplectic eigenvalues are all greater than  $1/2$  in Sec.~\ref{sec: perturbative formula} and all equal to  $1/2$  in Sec.~\ref{subsec: pure scaling}. When $\sigma_0$ contains both types of symplectic eigenvalues, the leading asymptotic behavior depends not only on the perturbation magnitude but also on the modes on which the perturbation acts. Perturbations on modes with symplectic eigenvalue  $1/2$  can exhibit the slower asymptotic vanishing and may dominate even when subleading in norm, whereas perturbations confined to modes with symplectic eigenvalues greater than $1/2$ retain the regular scaling.

Let $n_{\rm g}$ and $n_{\rm e}$ denote the numbers of modes with symplectic eigenvalue greater than $1/2$ and those equal to $1/2$, respectively, with
\begin{align}
n_{\rm g}+n_{\rm e}=n.
\end{align}
Let $\{\nu_j\}_{j=1}^{n_{\rm g}}$ be the symplectic eigenvalues greater than $1/2$,
\begin{align}
\nu_j>\frac{1}{2},
\qquad
j=1,\ldots,n_{\rm g}.
\end{align}
Without loss of generality, we group the two kinds of symplectic eigenvalues as follows. Let
\begin{align}
A=\Omega \sigma_0 = U\Gamma U^{-1},
\end{align}
where the eigenvalues are ordered as
\begin{align}
\Gamma
=
\operatorname{diag}
\left(
i\Lambda_{\rm g},
\frac{i}{2}I_{n_{\rm e}},
-i\Lambda_{\rm g},
-\frac{i}{2}I_{n_{\rm e}}
\right),
\label{eq:mixed pure Gamma}
\end{align}
with
\begin{align}
\Lambda_{\rm g}
=
\operatorname{diag}
(\nu_1,\ldots,\nu_{n_{\rm g}}).
\end{align}

For the VNE, we write the covariance-matrix perturbation in the
eigenbasis of $A$ as
\begin{align}
\widetilde E
=
U^{-1}EU
=
\begin{pmatrix}
\widetilde E_{11} & \widetilde E_{12}
& \widetilde E_{13} & \widetilde E_{14}\\
\widetilde E_{21} & \widetilde E_{22}
& \widetilde E_{23} & \widetilde E_{24}\\
\widetilde E_{31} & \widetilde E_{32}
& \widetilde E_{33} & \widetilde E_{34}\\
\widetilde E_{41} & \widetilde E_{42}
& \widetilde E_{43} & \widetilde E_{44}
\end{pmatrix}.
\label{eq:tilde E mixed pure blocks}
\end{align}
For the QRE, quantum Chernoff quantity, and QF, we similarly define
\begin{align}
\Delta\widetilde E
=
U^{-1}(E_1-E_2)U,
\end{align}
and partition $\Delta\widetilde E$ into the same $4\times4$ block
structure as in \eqref{eq:tilde E mixed pure blocks}. Since $U$ is
independent of $\epsilon$, the norm of the perturbation matrices in the eigenbasis of $A$ remains as
\begin{align}
\|\widetilde E\|=\Theta(\epsilon),
\qquad
\|\Delta\widetilde E\|=\Theta(\epsilon).
\end{align}
Although the overall perturbations scale as $\Theta(\epsilon)$, individual blocks may vanish more rapidly or identically. For example, if only the $11$ and $33$ blocks of $\widetilde E$ or $\Delta \widetilde{ E}$ are nonzero, the second-order covariance-matrix contributions involve only modes whose reference symplectic eigenvalues are greater than $1/2$. Therefore, there is no slower vanishing of these covariance matrix contributions, and they can be evaluated using the perturbative formulas in Sec.~\ref{sec: perturbative formula} restricted to those modes. 

In contrast, the $22$ and $44$ blocks of $\widetilde E$ and $\Delta \widetilde{ E}$ correspond to perturbations within the modes whose reference symplectic eigenvalue $1/2$. The perturbative coefficients associated with these blocks exhibit the strongest singular behavior as the symplectic eigenvalues approach $1/2$, so their contributions can dominate those involving modes with symplectic eigenvalues greater than $1/2$. We therefore identify sufficient conditions for this dominance rather than classifying all possible relative block scalings.

\begin{criterion}
\label{thm:leading pure mode perturbation}
Consider the reference matrix $A$ and the perturbations
$\widetilde E$ and $\Delta\widetilde E$ defined above. For the VNE,
suppose
\begin{align}
\left|
\operatorname{tr}\widetilde E_{22}
-
\operatorname{tr}\widetilde E_{44}
\right|
=
\Theta(\epsilon). \label{eq:criterion 1}
\end{align}
For the QRE, quantum Chernoff quantity, and QF, suppose
\begin{align}
\max\left\{
\|\Delta\widetilde E_{22}\|,
\|\Delta\widetilde E_{44}\|
\right\}
=
\Theta(\epsilon).  \label{eq:criterion 2}
\end{align}
Then, within the second-order perturbative analysis, the leading perturbative covariance-matrix contribution is determined by the modes with symplectic eigenvalues $1/2$, while contributions from modes with symplectic eigenvalues $\nu_j>1/2$ are strictly not
asymptotically leading (and can therefore be ignored in the perturbative regime). Consequently, tracing out the modes whose
symplectic eigenvalues $\nu_j>1/2$ does not change the coefficients of the leading
perturbative scaling.
\end{criterion}
The proof is given in Appendix~\ref{app:proof mixed pure theorem}.  Condition \eqref{eq:criterion 1} prevents cancellation of the leading first-order VNE contribution, whereas condition \eqref{eq:criterion 2} ensures that at least one of the corresponding perturbation blocks contributes at leading order for the QRE, quantum Chernoff quantity, and quantum infidelity.

For this reduced state, we evaluate the VNE variation, QRE, quantum
Chernoff quantities, and quantum infidelity directly from their
definitions and extract their leading asymptotic terms. These terms
also determine the leading asymptotic behavior of the corresponding
quantities for the full state. This reduction simplifies the theoretical analysis and may inform receiver design by identifying the modes relevant at leading order.
Alternatively, standard matrix eigenvalue perturbation theory \cite[Sec.~V]{stewart1990matrix} may be applied to approximate the perturbed symplectic eigenvalues of the reduced pure state. However, the resulting approximations to the quantum information quantities are not guaranteed to reproduce their leading asymptotic scalings or coefficients.

If the leading-order perturbation does not act within the eigenspaces
associated with symplectic eigenvalues equal to $1/2$, different
perturbation blocks may compete at the leading asymptotic order. In
this regime, no universal reduction is available, and the problem must
be analyzed case by case using the blockwise singularity analysis
developed in Appendix~\ref{app:proof mixed pure theorem}.

\subsection{Procedure for leading-order perturbative analysis}
\label{subsec:perturbative procedure}

We summarize a practical procedure for obtaining the leading
perturbative behavior of the information quantities of Gaussian states considered
in this work.

\begin{enumerate}
    \item \textit{Identify the leading contributions.}
    Determine the asymptotic scaling of the mean-vector and
    covariance-matrix perturbations. For the QRE, quantum Chernoff
    exponent, and QF, compare the corresponding contributions and retain
    only those that can contribute at leading order. The VNE depends
    only on the covariance matrix.

    \item \textit{Reference covariance matrix with all symplectic eigenvalues greater than $1/2$.}
    When the covariance matrix contribution is relevant, 
    we first take the limit $\epsilon\to0$ and determine the symplectic
    eigenvalues of the reference covariance matrix. If all symplectic
    eigenvalues are greater than $1/2$, the perturbative formulas in
    Sec.~\ref{sec: perturbative formula} apply directly and give both
    the leading scaling and coefficient.

   \item \textit{Reference covariance matrix with symplectic eigenvalues equal to $1/2$.} 
If all reference symplectic eigenvalues equal to $1/2$, we identify slower asymptotic vanishing and evaluate the corresponding quantum information quantities directly from their definitions. 
If the reference covariance matrix contains symplectic eigenvalues
both equal to and greater than $1/2$, and Criterion~\ref{thm:leading pure mode perturbation} is satisfied, the
modes whose symplectic eigenvalues are greater than $1/2$ may be
traced out without changing the leading perturbative scaling. 
The reduced-state information quantities can then be evaluated directly to extract their leading terms. When Criterion~\ref{thm:leading pure mode perturbation} is not satisfied, the competing perturbation blocks must be analyzed case by case using Appendix~\ref{app:proof mixed pure theorem}.
\end{enumerate}

\subsection{Numerical verification}
\label{sec:numerical pure boundary}

\begin{figure*}[t]
    \centering
    \subfloat[Full formula]{
        \includegraphics[width=0.49\textwidth]{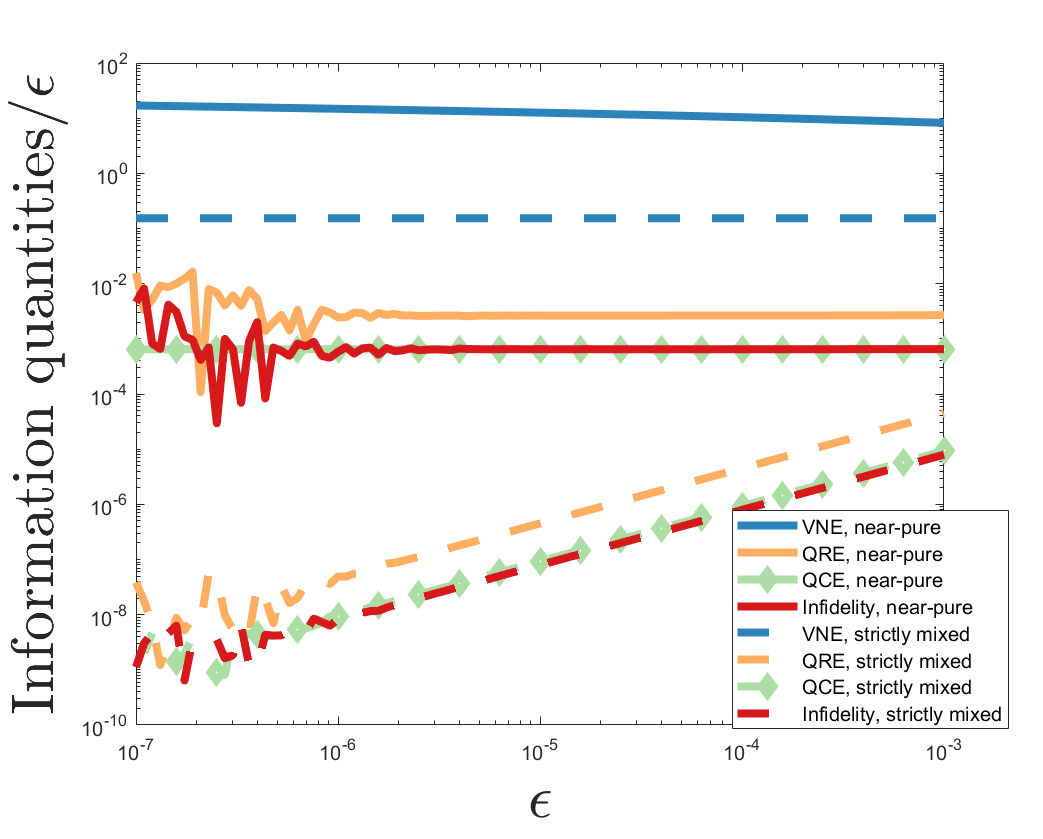}
        \label{fig:sub1}
    }
    \hfill
    \subfloat[perturbative formula]{
        \includegraphics[width=0.48\textwidth]{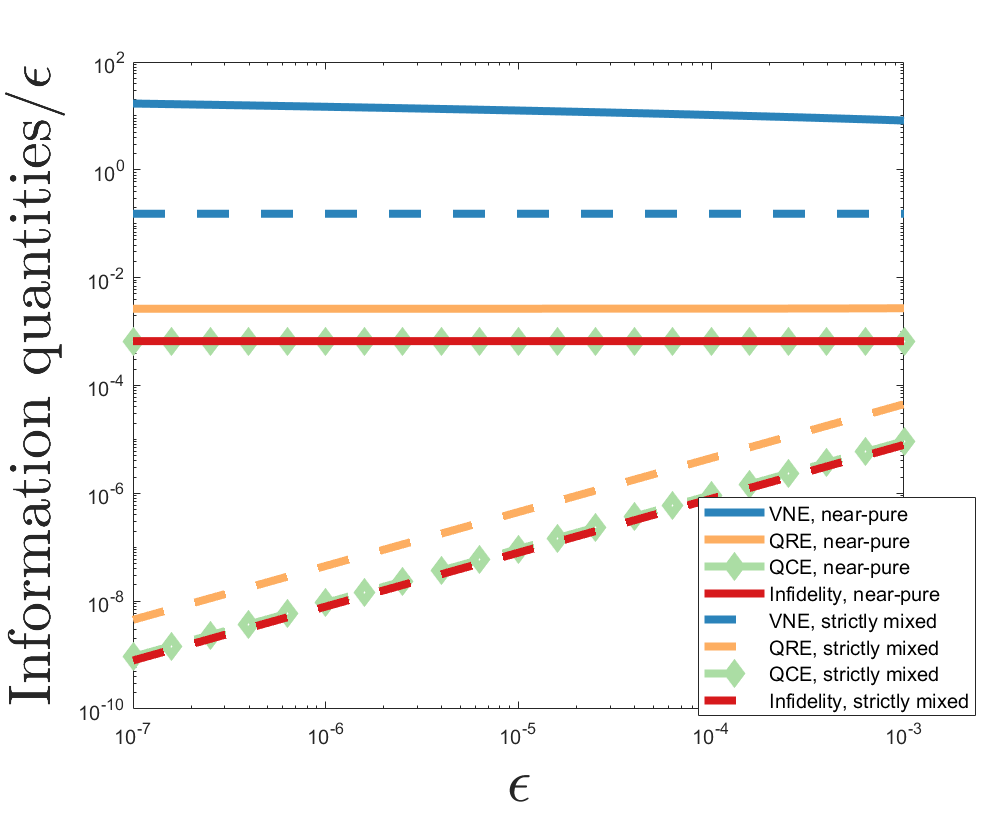}
        \label{fig:sub2}
    }

    \caption{Scaling of the information quantities of Gaussian states for reference covariance
matrices with symplectic eigenvalues greater than $1/2$ and equal to
$1/2$ as a function of the covariance-matrix perturbation norm
$\epsilon$. The plotted quantities are normalized by $\epsilon$. Each
point is averaged over $10^5$ independent random perturbation
realizations. (a) Direct evaluation using the full Gaussian formulas.
(b) Evaluation using the perturbative formulas derived in
Sec.~\ref{sec: perturbative formula}. When all symplectic eigenvalues
of the reference covariance matrix are greater than $1/2$, the VNE
variation scales as $O(\epsilon)$, while the QRE, QCE, and quantum
infidelity scale as $O(\epsilon^2)$. For the pure reference state, approximated by reference covariance matrices with symplectic eigenvalues approaching $1/2$, the VNE variation scales as $O(\epsilon|\log\epsilon|)$, while the QRE, QCE, and quantum infidelity scale as $O(\epsilon)$. The agreement between the two figures confirms the
predicted change in asymptotic scaling. }
    \label{fig:2}
\end{figure*}
We next numerically verify the predicted change in asymptotic scaling by considering reference states approaching purity. We use the same random perturbation matrices and normalization procedure as in Sec.~\ref{subsec:numerical verification}. To compare reference covariance matrices with symplectic eigenvalues greater than $1/2$ and approaching  $1/2$, we consider
\begin{align}
    \sigma_{0,\rm mix} = 0.6 I_4, \qquad
    \sigma_{0, \rm p} = \left(\frac{1}{2}+\epsilon\right) I_4 ,
\end{align}
respectively. Here $\epsilon \to 0$.

We consider the following quantities: VNE variation $ \Delta S_{\rm VN}$,
QRE $ D(\hat{\rho}_2\Vert\hat{\rho}_1) $, QCE $\xi(\hat{\rho}_1,\hat{\rho}_2)$, and quantum infidelity $1-F(\hat{\rho}_1,\hat{\rho}_2)$.
Figure~\ref{fig:2} plots these information quantities divided by $\epsilon$.

For perturbations of  $\sigma_{0,\rm mix}$, the VNE variation
$\Delta S_{\rm VN}=O(\epsilon)$, whereas the QRE, QCE, and quantum 
infidelity scale as $O(\epsilon^2)$. Consequently,
$\Delta S_{\rm VN}/\epsilon$ approaches a constant, while the
remaining normalized quantities scale as $O(\epsilon)$.
For a pure reference state, the VNE variation exhibits the
logarithmically enhanced scaling
  $  \Delta S_{\rm VN}
    =
    O\!\left(\epsilon|\log\epsilon|\right)$,
whereas the QRE, QCE, and quantum infidelity scale as $O(\epsilon)$.
Their normalized values therefore remain constant as
$\epsilon\rightarrow0$, while
$\Delta S_{\rm VN}/\epsilon=O(|\log\epsilon|)$.

Figure~\ref{fig:sub1}  evaluates the
information quantities directly from the full Gaussian formulas,
whereas Fig.~\ref{fig:sub2} evaluates the
corresponding perturbative formulas derived in Sec.~\ref{sec: perturbative formula}. Each point is averaged over $10^5$ independent random perturbation realizations.
The two evaluations in Fig.~\ref{fig:2} exhibit the same asymptotic behavior, confirming the slower vanishing of the information quantities when the reference state is pure.

For $\epsilon \le 10^{-6}$, the direct evaluations in
Fig.~\ref{fig:sub1} begin to exhibit
numerical fluctuations. In this regime, the information quantities become small and the  errors from the eigen-decomposition and symplectic-decomposition calculations become comparably large. In contrast, the perturbative approximation error continues to decrease with  $\epsilon$,  demonstrating the numerical advantage of the perturbative formulas derived in Sec.~\ref{sec: perturbative formula}.

\section{Applications}\label{sec:applications}

In this section, we apply the perturbative-analysis procedure developed in Sec.~\ref{subsec:perturbative procedure} to quantum illumination and sub-Rayleigh optical source discrimination, recovering known asymptotic QCEs directly within the Gaussian-state formalism.

\subsection{Quantum illumination with coherent and TMSV inputs}

As the first application, we revisit the quantum-illumination problem
considered by Cox {\em et al.} in \cite{AliCox2024}, in which one seeks to discriminate among a library of weakly reflecting targets using either coherent-state or two-mode-squeezed-vacuum (TMSV) states in the presence of thermal background noise. We denote by $\bar n$ the mean photon number per mode transmitted by the optical probe, by $\bar n_{\rm B}$ the mean photon number injected by the thermal background per mode, and by $\kappa^{(k)}$ the target reflectivity under hypothesis $k$. Cox {\em et al.} further defined $C^{(k)}$ as a real-valued vector of spatial-mode overlaps between the receiver-pupil-truncated spatial mode excited by object $k$ and the spatial functions corresponding to the elements of a fixed spatial-mode basis that contains the excitations corresponding to all the object hypotheses in the library.
The original analysis in \cite{AliCox2024} applied the perturbative framework of  \cite{Grace2022}, which was developed for finite-dimensional quantum states.
Here, we show that the same asymptotic quantum Chernoff exponents can be recovered directly from the Gaussian-state description using the perturbative QCE derived in Theorem~\ref{thm:QCE-perturbation} and Corollary~\ref{cor:QCE-s-half}. 

Following \cite{AliCox2024}, we consider the asymptotic regime
\begin{align}
\kappa^{(k)}
&\in
o\left(
\frac{\bar{n}}{\bar{n}_{\rm B}}
\right),
\quad
\bar{n}
\in
O\left(
e^{-\bar{n}_{\rm B}}
\right),
\quad
1\in o(\bar{n}_{\rm B}), \label{eq:limit regime}
\end{align}
For the coherent-state input, the leading perturbative contribution
arises from the mean-vector difference between the two hypotheses.
Applying Theorem~\ref{thm:QCE-perturbation} at
$s=\frac{1}{2}$ gives
\begin{align}
\xi_{C}^{(2)}
\simeq
\frac{\bar{n}}{4\bar{n}_{\rm B}}
\left\|
\sqrt{\kappa^{(1)}}\,C^{(1)}
-
\sqrt{\kappa^{(2)}}\,C^{(2)}
\right\|_F^2,
\label{eq:Ali coherent main}
\end{align}
where $\|A\|_{\rm F}$ denotes the Frobenius norm,
defined by
\begin{align}
\|A\|_{\rm F}
=
\sqrt{\trace\left(A^\dagger A\right)}.
\end{align}
For the TMSV input (with the transmitter retaining the idler mode while sending the signal mode of the TMSV to illuminate the target), the mean vectors are zero and the Chernoff exponent
is determined entirely by the covariance-matrix contribution. Since the reference state has symplectic eigenvalues all greater than $1/2$, applying Theorem~\ref{thm:QCE-perturbation} at
$s=\frac{1}{2}$ gives
\begin{align}
\xi_{Q}^{(2)}
\simeq
\frac{\bar{n}}{\bar{n}_{\rm B}}
\left\|
\sqrt{\kappa^{(1)}}\,C^{(1)}
-
\sqrt{\kappa^{(2)}}\,C^{(2)}
\right\|_F^2.
\label{eq:Ali TMSV main}
\end{align}
Therefore, 
\begin{align}
\xi_{Q}^{(2)}
=
4\xi_{C}^{(2)},
\end{align}
showing the factor-of-four enhancement in the quantum Chernoff
exponent provided by the TMSV input.
Furthermore, $\xi_{Q}^{(2)}$ retains the usual quadratic scaling
with the covariance-matrix difference when all reference symplectic eigenvalues are greater than $1/2$. The derivations of
\eqref{eq:Ali coherent main} and
\eqref{eq:Ali TMSV main}, as well as this quadratic scaling, are
provided in Appendix~\ref{appendix:Ali}.

\subsection{Sub-Rayleigh imaging} \label{subsec:application 2}

As the second application, we revisit the sub-Rayleigh object
discrimination problem studied by Krovi {\em et al.}~\cite{krovi2016attainingquantumlimitpassive} and Lu {\em et al.} \cite{lu2018quantum}. 
The task is to determine whether the received optical field originates from a single point source or from two closely spaced point sources radiating incoherently with respect to one another, situated at a known Rayleigh-unit separation $\mu$. 
In \cite{lu2018quantum}, the weak-source approximation was applied to the received thermal states, under which the output state contains only the zero- and one-photon components and therefore has finite-dimensional support. Applying the finite-dimensional perturbation theory of \cite{Grace2022} will yield the perturbative quantum Chernoff exponent, albeit whose validity will hold true only under the weak-source approximation.

Since our perturbation theory applies to Gaussian states, we can use the exact analysis of this problem developed by Krovi {\em et al.}~ \cite{krovi2016attainingquantumlimitpassive}, where the quantum state under each hypotheses---for the one-versus-two source discrimination problem---was expressed as a multimode correlated zero-mean thermal state (i.e., Gaussian states). This analysis, performed directly within the Gaussian-state formalism, thereby eliminates the need to invoke the weak-source approximation. Let  $N_0$ denote the mean photon number of the received quantum state, and $\mu$ denote the ratio of the angular separation under the two-source hypothesis at the image plane to the width of the point-spread function. 
The corresponding covariance matrices are given by
\begin{align}
V^{(1)}
&=
\begin{bmatrix}
V_1^{(1)} & 0\\
0 & V_1^{(1)}
\end{bmatrix},
&
V^{(2)}
&=
\begin{bmatrix}
V_1^{(2)} & 0\\
0 & V_1^{(2)}
\end{bmatrix},
\end{align}
where
\begin{align}
V_1^{(1)}
&=
\frac{1}{2}
\begin{bmatrix}
2N_0+1 & 0 & 0\\
0 & 1 & 0\\
0 & 0 & 1
\end{bmatrix},
\\
V_1^{(2)}
&=
\frac{1}{2}
\begin{bmatrix}
N_0A_0^2+1 & 0 & N_0A_0B\\
0 & 2\kappa+1 & 0\\
N_0A_0B & 0 & N_0B^2+1
\end{bmatrix},
\end{align}
and
\begin{align}
A_0
&=
\sqrt{2}\sinc(\mu),
\\
B
&=
\sqrt{
\left[
1+\sinc(2\mu)-2\sinc^2(\mu)
\right]
}\\
\kappa & = \frac{N_0}{2}[1-\sinc(2\mu)],
\end{align}
where $\sinc(x) = \sin(x)/x$.
In the sub-Rayleigh regime $\mu\to0$, the two covariance matrices approach the same reference covariance matrix, given by
\begin{align}
V_1^{(1)},V_1^{(2)}
\rightarrow
\operatorname{diag}\left(
N_0+\frac{1}{2},
\frac{1}{2},
\frac{1}{2}
\right).
\end{align}
The first mode has symplectic eigenvalue greater than $1/2$ while the remaining two modes have symplectic eigenvalues equal to $1/2$. Meanwhile, since the overall perturbation is 
\begin{align}
\left\|V^{(1)}-V^{(2)}\right\|=\Theta(\mu^2),
\end{align}
and the perturbation on second and third modes is
\begin{align}
    \left\|\diag\left(\kappa, N_0 B^2/2\right)\right\| = \Theta(\mu^2),
\end{align}
the condition for applying Criterion~\ref{thm:leading pure mode perturbation} is satisfied. We can therefore trace out the first mode. The reduced reference state is a pure state, and the third-mode perturbation scales as $N_0B^2=\Theta(\mu^4)$ and is negligible at leading order compared with the second-mode perturbation $\kappa=\Theta(\mu^2)$. We therefore trace out the  third mode as well and retain only the second mode. The remaining single-mode covariance matrices are
\begin{align}
    V^{(1)}_{\rm rem} = \frac{1}{2} I_2, \qquad V^{(2)}_{\rm rem} = \left(\frac{1}{2}+\kappa\right) I_2.
\end{align}
The corresponding states are therefore the vacuum state and a thermal state with mean photon number $\kappa$. By definition, the QCE of the remaining single-mode problem is
\begin{align}
\xi_{\rm rem}
&=
\!\max_{s\in[0,1]}
\! - \log
\left[
\bra{0}
\sum_{k=0}^{\infty}
\frac{\kappa^{k(1-s)}}
{(1+\kappa)^{(1+k)(1-s)}}
\ket{k}\bra{k}
\ket{0}
\right]
\nonumber\\
&=
\log(1+\kappa)
\nonumber\\
&=
\frac{N_0}{3}\mu^2+O(\mu^4).
\end{align}
Here, $\xi_{\rm rem}$ is maximized at $s = 0$, which does not contradict  Corollary \ref{cor:QCE-s-half}, since the latter applies to reference states whose symplectic eigenvalues are all greater than $1/2$.
Using $\delta(d)=\sinc(2\mu)$ in the notation of \cite{lu2018quantum}, the small-separation expansion of its Eq.~(5) yields the same leading term $N_0\mu^2/3$. Thus, a straightforward single-mode calculation reproduces the known asymptotic coefficient, without invoking the weak-source approximation. 

This reduction also suggests a simple receiver design: by retaining only the second mode and performing photon-number-resolving detector, we achieve the leading asymptotic QCE of the original discrimination problem.

This quadratic scaling outperforms the  $O(\mu^3)$ scaling achieved by conventional direct imaging in the image plane (a suboptimal receiver for this problem) \cite{darekar2026fundamentallimitversuspoint}, demonstrating a quantum advantage in the sub-Rayleigh regime.

\section{Conclusion}
\label{sec:conclusion}

In this work, we developed a perturbation theory for quantum information quantities of Gaussian states directly in terms of perturbations of their moments. For reference covariance matrices whose symplectic eigenvalues are all greater than $1/2$, we derived compact second-order perturbative formulas for the VNE, QRE, quantum Chernoff quantity, and QF using primary matrix functions and Fréchet derivatives. We also showed that the perturbative quantum Chernoff quantity is maximized at $s=1/2$, yielding the perturbative QCE. Furthermore, the derivations establish general trace identities for Fréchet derivatives of primary matrix functions of diagonalizable matrices, providing analytical tools for perturbative evaluations of quantum information quantities.

We next identified slower asymptotic vanishing associated with symplectic eigenvalues equal to $1/2$ for all four information quantities. Specifically, for covariance-matrix perturbations of order $O(\epsilon)$, the VNE variation can scale as $O(\epsilon|\log\epsilon|)$ rather than $O(\epsilon)$, while the covariance-matrix contributions to the QRE, QCE, and quantum infidelity can scale as $O(\epsilon)$ rather than $O(\epsilon^2)$. Since these quantities govern fundamental limits across quantum-information tasks, slower asymptotic vanishing can improve performance, as illustrated by the larger quantum Chernoff exponent in state discrimination. The leading asymptotic behavior also depends on the symplectic eigenspaces affected by the perturbation, rather than only its magnitude.

For reference covariance matrices containing symplectic eigenvalues both equal to and greater than $1/2$, we identified sufficient conditions under which the modes with symplectic eigenvalues greater than $1/2$ can be traced out without changing the leading asymptotic coefficients of the covariance-matrix contributions. This reduction simplifies theoretical calculations and may also inform receiver design by identifying the modes relevant at leading order. 

Applying our framework to quantum illumination and sub-Rayleigh incoherent optical source discrimination, we recovered known asymptotic QCEs directly within the Gaussian-state formalism. In quantum illumination, we recovered the factor-of-four enhancement provided by TMSV probes. In sub-Rayleigh discrimination, the mode reduction yielded a single-mode calculation reproducing the known quadratic scaling and asymptotic coefficient in the QCE without invoking the weak-source approximation, while also identifying a corresponding simple quantum-optimal receiver design that uses a spatial-mode resolved photon detection.

\section{Acknowledgments}
This work was supported by the Office of Naval Research (ONR) under grant number N00014-19-1-2189.

\bibliographystyle{IEEEtran}
\bibliography{papers}

\appendix
\section{Derivation of Lemmas} \label{app:lemma}

\begin{lemma} \label{lemma: trace}
Let $A,E\in\mathbb{C}^{N\times N}$, with $A$ diagonalizable, and let $f$ and $g$ be $2N-1$ times continuously differentiable on a domain containing the eigenvalues of $A$. Then
\begin{align}
\trace\left(f(A) D^{(1)}_{g(A)}[E]\right)
&=
\trace\left(E f(A) g^{\prime}(A)\right). \label{eq: t fa d1}
\end{align}
\end{lemma}
\begin{proof}
Since $g$ is $2N-1$ times continuously differentiable, the first-order Fr\'echet derivative $D^{(1)}_{g(A)}[E]$ exists and admits a
divided-difference representation.
Using $A = \sum_{j=1}^{N} \lambda_j P_j $ as in \eqref{eq:A = P}, we have
    \begin{align}
        &\trace\left(f(A) D^{(1)}_{g(A)}[E]\right)\nonumber\\
        & \qquad= \trace\left(f(A) \sum_{k,l} g^{[1]}[\lambda_k,\lambda_l] P_k E P_l \right) \label{eq:Lemma 1 e 1}\\ 
        & \qquad = \trace\left( \sum_{k} f(\lambda_k)g^{\prime}(\lambda_k) P_k E P_k \right) \label{eq:U_k and U_l}\\ 
        & \qquad = \trace\left( E \sum_{k} f(\lambda_k)g^{\prime}(\lambda_k) P_k  \right) \label{eq:Lemma 1 e 2}\\ 
       & \qquad = \trace\left( E f(A)g^{\prime}(A)  \right). \label{eq: t fa d1 1}
    \end{align}
    Here,  \eqref{eq:Lemma 1 e 1} follows directly from the first-order divided-difference  representation in  \eqref{eq:D1f P}. 
    To obtain \eqref{eq:U_k and U_l}, we use $f(A)P_k=f(\lambda_k)P_k$, and $P_lP_k=\delta_{lk}P_k$, so that only the terms with $l=k$ remain, for which $g^{[1]}[\lambda_k,\lambda_k]=g'(\lambda_k)$. Equation~\eqref{eq:Lemma 1 e 2} then follows from the cyclic property of the trace and $P_k^2=P_k$.  Finally, \eqref{eq: t fa d1 1} follows from the definition of the  corresponding primary matrix function
$\sum_k f(\lambda_k)g'(\lambda_k)P_k=f(A)g'(A)$. Thus,  \eqref{eq: t fa d1} follows.
\end{proof}
When $f = (g^{\prime})^{-1}$, Lemma~\ref{lemma: trace} reduces to \cite[Lemma~3]{Grace2022}.

\begin{lemma} \label{lemma: trace 2}
Let $A,E\in\mathbb{C}^{N\times N}$, with $A$ diagonalizable,  and let $ f $ and $g$   be $3N-1$ times continuously differentiable on a domain containing the eigenvalues of $A$. Then
\begin{align}
&\trace\left( f(A) D^{(2)}_{g(A)}[E]\right) \nonumber \\
&= \trace\left( E D_{(fg^{\prime})(A)}^{(1)} [E] -  D^{(1)}_{f(A)}[E]D^{(1)}_{g(A)}[E]\right) \label{eq:trace lemma2}
\end{align}
\end{lemma}

\begin{proof}
Since $f$ and $g$ are $3N-1$ times continuously differentiable, their first- and second-order Fréchet derivatives exist and admit divided-difference representations. We first simplify the left-hand side of \eqref{eq:trace lemma2}
 \begin{align}
       &\trace(f(A) D^{(2)}_{g(A)}[E]) \nonumber \\
  & = \trace\left(2 f(A)  \sum_{k,l,m} g^{[2]}(\lambda_k,\lambda_l,\lambda_m) P_k E P_l E P_m \right) \label{eq:Lemma 2 p1}\\
  & = \trace\left(2 \sum_{k,l} f(\lambda_k) g^{[2]}(\lambda_k,\lambda_l,\lambda_k) P_k E P_l E \right) \label{eq:Lemma 2 p2} \\
  & = \trace\left( \sum_{k,l} f(\lambda_k) g^{[2]}(\lambda_k,\lambda_l,\lambda_k) P_k E P_l E \right) \nonumber\\
&\qquad+\trace\left( \sum_{l,k} f(\lambda_l) g^{[2]}(\lambda_l,\lambda_k,\lambda_l) P_l E P_k E \right) \label{eq:Lemma 2 p3} \\
& = \trace\left( E \sum_{k,l}  h_{k,l} P_k E P_l  \right).\label{eq:lemma 2 left}
    \end{align}
   Here, \eqref{eq:Lemma 2 p1} follows from the second-order divided-difference representation in \eqref{eq:D2f P}. Equation \eqref{eq:Lemma 2 p2} follows from the spectral-projector representation  $A = \sum_{j=1}^{N}\lambda_j P_j$ in \eqref{eq:A = P},  $f(A)P_k = f(\lambda_k)P_k$, and $\trace\left(P_k E P_l E P_m\right) = \trace\left(P_m P_k E P_l E \right) = \trace\left(\delta_{mk} P_k E P_l E \right) $. Equation \eqref{eq:Lemma 2 p3} follows  by splitting the factor of two into two identical sums and interchanging the indices  $k$ and $l$ in the second sum. For \eqref{eq:lemma 2 left}, we define
    \begin{align}
     \!\! h_{k,l}  =   f(\lambda_k) g^{[2]}(\lambda_k,\lambda_l,\lambda_k)+ f(\lambda_l) g^{[2]}(\lambda_l,\lambda_k,\lambda_l).
    \end{align}
When $k \ne l$, $h_{k,l} $ yields
\begin{align}
     &h_{k,l}\nonumber\\
     &= f(\lambda_k) \frac{g^{\prime}(\lambda_k) - g^{[1]}\left( \lambda_k, \lambda_l \right) }{\lambda_k-\lambda_l} + f(\lambda_l) \frac{g^{\prime}(\lambda_l) - g^{[1]}\left( \lambda_k, \lambda_l \right) }{\lambda_l-\lambda_k} \\
     & = \frac{f(\lambda_k)g^{\prime}(\lambda_k) - f(\lambda_l)g^{\prime}(\lambda_l) }{\lambda_k-\lambda_l}  - \frac{f(\lambda_k) - f(\lambda_l)}{\lambda_k- \lambda_l} g^{[1]}\left( \lambda_k, \lambda_l \right) \\
     & = (fg^{\prime})^{[1]}(\lambda_k,\lambda_l) - f^{[1]}\left( \lambda_k, \lambda_l \right)g^{[1]}\left( \lambda_k, \lambda_l \right).\label{eq:h t 1}
\end{align}
    When $k = l$, $h_{k,l}$ yields
    \begin{align}
        h_{k,l} &= (fg^{\prime\prime})(\lambda_k) \\
        &= \left(\left(fg^{\prime}\right)^{\prime} - f^{\prime}g^{\prime}\right)(\lambda_k)\\
        & = (fg^{\prime})^{[1]}(\lambda_k,\lambda_l) - f^{[1]}\left( \lambda_k, \lambda_l \right)g^{[1]}\left( \lambda_k, \lambda_l \right). \label{eq:h t 2}
    \end{align}
Combining \eqref{eq:h t 1} and \eqref{eq:h t 2}, we obtain that, for all $k$ and $l$, $h_{k,l}$ yields
\begin{align}
    h_{k,l} =  (fg^{\prime})^{[1]}(\lambda_k,\lambda_l) - f^{[1]}\left( \lambda_k, \lambda_l \right)g^{[1]}\left( \lambda_k, \lambda_l \right). \label{eq:lemma 2 p8}
\end{align}
We next simplify the right-hand side of \eqref{eq:trace lemma2} and show that it reduces to the same expression as \eqref{eq:lemma 2 left}. Using the first-order Fr\'echet derivative representation in \eqref{eq:D1f P}, the first term of \eqref{eq:trace lemma2} yields
\begin{align}
    \trace\left( E D_{(fg^{\prime})(A)}^{(1)} [E] \right) = \trace\left( E \sum_{k,l} (fg^{\prime})^{[1]}(\lambda_k,\lambda_l) P_k E P_l  \right). \label{eq:lemma 2 p9}
\end{align}
Similarly, using \eqref{eq:D1f P},  the second term of \eqref{eq:trace lemma2} yields
\begin{align}
    &-\trace\left(   D^{(1)}_{f(A)}[E]D^{(1)}_{g(A)}[E]\right) \nonumber\\
    &= -\!\trace \!\! \left( \sum_{k,l} f ^{[1]}(\lambda_k,\lambda_l) P_k E P_l  \sum_{k',l'} g^{[1]}(\lambda_{k'},\lambda_{l'}) P_{k'} E P_{l'} \right) \\
    & = - \!\trace \!\! \left( E \sum_{k,l} f ^{[1]}(\lambda_k,\lambda_l)  g^{[1]}(\lambda_k,\lambda_l) P_k E P_l      \right), \label{eq:lemma 2 p10}
\end{align}
where the last equality follows from
$P_lP_{k'}=\delta_{lk'}P_l$, the cyclic property of the trace,
and the symmetry
$g^{[1]}(\lambda_k,\lambda_l)
=g^{[1]}(\lambda_l,\lambda_k)$.
Substituting \eqref{eq:lemma 2 p9} and \eqref{eq:lemma 2 p10} into the right-hand side of \eqref{eq:trace lemma2}, we obtain the same expression as that obtained by substituting \eqref{eq:lemma 2 left} and \eqref{eq:lemma 2 p8} into the left-hand side of \eqref{eq:trace lemma2}. Thus, \eqref{eq:trace lemma2} follows.

\end{proof}
When $f = 1/(g^{\prime})$, this lemma reduces to \cite[Lemma~4]{Grace2022}.
We now turn to the Fr\'echet derivatives of $1/f$.

\begin{lemma}
\label{lemma:Df-1}
Let $A\in\mathbb{C}^{N\times N}$ be diagonalizable, and let $f$ and $1/f$ be $2N-1$ times
continuously differentiable on a domain containing the eigenvalues of $A$. Then
\begin{align}
D_{(1/f)(A)}^{(1)}[E]
=
-(1/f)(A)D_{f(A)}^{(1)}[E](1/f)(A).
\label{eq:df-1}
\end{align}
\end{lemma}

\begin{proof}
Since $f$ and $(1/f)$ are $2N-1$ times continuously differentiable, the first-order Fr\'echet derivatives $D_{(1/f)(A)}^{(1)}[E]$ and $ D_{f(A)}^{(1)}[E] $  exist and admit a
divided-difference representation. Using the spectral-projector representation
$A=\sum_{k=1}^{N}\lambda_kP_k$,
we have
\begin{align}
D_{(1/f)(A)}^{(1)}[E]
=
\sum_{k,l}
(1/f)^{[1]}(\lambda_k,\lambda_l)
P_kEP_l.
\label{eq:df-1 0}
\end{align}
From the definition of the first divided difference representation in \eqref{eq:first_dd},
\begin{align}
&(1/f)^{[1]}(\lambda_k,\lambda_l)
\nonumber\\
&=
\begin{cases}
\dfrac{(1/f)(\lambda_k)-(1/f)(\lambda_l)}
{\lambda_k-\lambda_l},
& \lambda_k\neq\lambda_l,\\[6pt]
(1/f)'(\lambda_k),
& \lambda_k=\lambda_l,
\end{cases}
\nonumber\\
&=
\begin{cases}
-(1/f)(\lambda_k)
\dfrac{f(\lambda_k)-f(\lambda_l)}
{\lambda_k-\lambda_l} (1/f)(\lambda_l) ,
& \lambda_k\neq\lambda_l,\\[6pt]
-(1/f^2)(\lambda_k)f'(\lambda_k),
& \lambda_k=\lambda_l,
\end{cases}
\nonumber\\
&=
-(1/f)(\lambda_k)
f^{[1]}(\lambda_k,\lambda_l)(1/f)(\lambda_l).
\label{eq:df-1 1}
\end{align}
Substituting \eqref{eq:df-1 1} into \eqref{eq:df-1 0}, we obtain
\begin{align}
& D_{(1/f)(A)}^{(1)}[E] \nonumber\\
&=
-\sum_{k,l}
(1/f)(\lambda_k)
f^{[1]}(\lambda_k,\lambda_l)
(1/f)(\lambda_l)
P_kEP_l
 \\
&=
-\left(
\sum_{k'}(1/f)(\lambda_{k'})P_{k'}
\right)
\left(
\sum_{k,l}
f^{[1]}(\lambda_k,\lambda_l)P_kEP_l
\right)
\nonumber\\
&\qquad\times
\left(
\sum_{l'}(1/f)(\lambda_{l'})P_{l'}
\right)
\label{eq:df-1 21} \\
&=
-(1/f)(A)D_{f(A)}^{(1)}[E](1/f)(A). \label{eq:df-1 22}
\end{align}
Equation~\eqref{eq:df-1 21} follows from the projector relation $P_jP_k=\delta_{jk}P_j$. Equation~\eqref{eq:df-1 22} follows from the primary-matrix-function representation of $(1/f)(A)$ and the first-order divided-difference representation of $D^{(1)}_{f(A)}[E]$. Thus, \eqref{eq:df-1} follows.
\end{proof}

\begin{lemma}\label{lemma:Frechet log det A}[Fr\'echet derivatives of $\log\det A$]
Let $A\in\mathbb{C}^{N\times N}$ be diagonalizable and invertible.  For any perturbation direction $E\in\mathbb{C}^{N\times N}$, we have
\begin{align}
D^{(1)}_{\log\det (A)}[E] &= \trace(EA^{-1}), \label{eq: d1 log det A} \\
D^{(2)}_{\log\det(A)}[E] &= -\trace(E A^{-1}EA^{-1}).\label{eq: d2 log det A}
\end{align}
\end{lemma}
\begin{proof} 
We first prove \eqref{eq: d1 log det A}.  $\log\det$ is twice Fr\'echet differentiable. Using the G\^ateaux-derivative representation in \eqref{eq: directional definition}, we have
    \begin{align}
        D^{(1)}_{\log\det(A)}[E] & = \frac{d}{dt}\Big|_{t = 0} \log \det (A+tE) \\
        & = \trace\left(\frac{d}{dt}\Big|_{t = 0}\log (A+tE)\right)\label{eq:lemma log det1}\\
        & =  \trace \left(  D^{(1)}_{\log(A)}[E] \right) \\
        & =\trace \left(  E A^{-1} \right) \label{eq:lemma log det2},
    \end{align}
 where \eqref{eq:lemma log det1} follows from $\log\det(A) = \trace(\log(A))$, and \eqref{eq:lemma log det2} follows from Lemma~\ref{lemma: trace} with $f(x) = 1$ and $g(x) = \log x$.
Now we prove \eqref{eq: d2 log det A}.  By G\^ateaux-derivative representation in \eqref{eq: directional definition}, \eqref{eq: d2 log det A} yields
\begin{align}
     D^{(2)}_{\log\det(A)}[E] & = \frac{d^2}{dt^2}\Big|_{t = 0} \log \det (A+tE) \\ 
     & =  \trace \left(\frac{d^2}{dt^2}\Big|_{t = 0}\log  (A+tE) \right) \label{eq:lemma 3 p1} \\
     & = \trace\left(D^{(2)}_{\log(A)}[E]\right) \\
     & = \trace\left(E D^{(1)}_{A^{-1}}[E]\right)  \label{eq:lemma 3 p2}\\
     & = -\trace\left(  E A^{-1}E A^{-1} \right) \label{eq:lemma 3 p3},
\end{align}
where \eqref{eq:lemma 3 p1} follows from   $\log\det(A) = \trace(\log(A))$, \eqref{eq:lemma 3 p2} follows from Lemma~\ref{lemma: trace 2} by taking $f(x) = 1$ and $g(x) = \log x$, and \eqref{eq:lemma 3 p3} follows from Lemma~\ref{lemma:Df-1} with $f(x) = x$.
 \end{proof}

We are now ready to derive the Fr\'echet derivatives of $\log \det f(A)$, which will be used repeatedly in the perturbative formulas.
\begin{lemma}[Fr\'echet derivatives of $\log\det f(A)$] \label{lemma:D logdet f}
Let $A\in\mathbb{C}^{N\times N}$ be diagonalizable, and let $f$ be $3N-1$ times continuously differentiable on a domain containing the eigenvalues of $A$. Assume that $f(A)$ is invertible.
Define
\begin{align}
    h(A)=\log\det f(A).
\end{align}
Letting $g(x)=\log(f(x))$, we have
\begin{align}
D^{(1)}_{h(A)}[E] &= \trace\left(E\,g'(A)\right), \label{eq:DF_trace_g}\\
D^{(2)}_{h(A)}[E] &= \trace\left(E\,D^{(1)}_{g'(A)}[E]\right). \label{eq:D2F_trace_g}
\end{align}
\end{lemma}
\begin{proof}
 Using  $\log\det(A) = \trace(\log(A))$, we have
   \begin{align}
       D^{(1)}_{h(A)}[E] & = \frac{d}{dt}\Big|_{t=0} \log\det f(A+tE) \\
       & = \trace\left(  \frac{d}{dt}\Big|_{t=0} \log f(A+tE)\right) \\
       & = \trace\left(  D^{(1)}_{g(A)}[E]\right) \\
       & = \trace\left(  E g^{\prime}(A)\right),\label{eq:lemma log det f1}
   \end{align} 
 where \eqref{eq:lemma log det f1} follows from Lemma~\ref{lemma: trace} with $f(x) =1$, which proves \eqref{eq:DF_trace_g}. Taking another Fr\'echet derivative in the same direction $E$, \eqref{eq:lemma log det f1} yields
   \begin{align}
      D^{(2)}_{h(A)}[E] &= \trace\left(E\,D^{(1)}_{g'(A)}[E]\right).
      \end{align}
\end{proof}

\section{Quantum relative entropy} \label{appendix:gaussian_qi frechet QRE}
Throughout the appendices, $\simeq$ denotes equality up to the second-order of $E_1$ and $E_2$ in the perturbative expansion of the quantum information quantities.
\subsection{Gaussian covariance-matrix representation}
\label{appendix sub:QRE CM}
We first write \eqref{eq:QRE general} in terms of
$\Omega\sigma_k$, $k = 1,2$. For the first term of \eqref{eq:QRE general},
we have
\begin{align}
\log\left(\frac{Z_1}{Z_2}\right)
&=
\log\left(
\frac{
\det\left(\sigma_1+\frac{i}{2}\Omega\right)}
{\det\left(\sigma_2+\frac{i}{2}\Omega\right)}
\right)
\\
&=
\log\left(
\frac{
\det\Omega\,
\det\left(\sigma_1+\frac{i}{2}\Omega\right)}
{\det\Omega\,
\det\left(\sigma_2+\frac{i}{2}\Omega\right)}
\right)
\\
&=
\log\left(
\frac{
\det\left(\Omega\sigma_1-\frac{i}{2}I\right)}
{\det\left(\Omega\sigma_2-\frac{i}{2}I\right)}
\right),
\label{eq:QRE t1}
\end{align}
where \eqref{eq:QRE t1} follows from $\Omega^2 = -I$ and $\det(AB) = \det(A)\det(B)$.
Next, we rewrite $G_k$ in \eqref{eq:QRE general} as a primary matrix function of 
$\Omega\sigma_k$:
\begin{align}
G_k
&=
2i\Omega\arccoth(2i\,\sigma_k\Omega)
\\
&=
2i\Omega\arccoth
\left(
2i\Omega^{-1}(\Omega\sigma_k)\Omega
\right)
\\
&=
f_{\rm RE}(\Omega\sigma_k)\Omega,
\label{eq:QRE G}
\end{align}
where we used the similarity invariance of primary matrix functions and defined
\begin{align}
f_{\rm RE}(x)
=
2i\arccoth(2ix).
\end{align}
Thus, the second term in \eqref{eq:QRE general} becomes
\begin{align}
&\trace\left((G_1-G_2)\sigma_2\right)
\nonumber\\
&\qquad=
\trace\left[
\left(
f_{\rm RE}(\Omega\sigma_1)
-
f_{\rm RE}(\Omega\sigma_2)
\right)
\Omega\sigma_2
\right].
\label{eq:QRE t2}
\end{align}
The mean-vector term in \eqref{eq:QRE general} becomes
\begin{align}
\delta_d^T G_1\delta_d
=
\delta_d^T
f_{\rm RE}(\Omega\sigma_1)
\Omega\delta_d.
\label{eq:QRE t3}
\end{align}
Combining \eqref{eq:QRE t1}, \eqref{eq:QRE t2}, and
\eqref{eq:QRE t3}, we obtain the QRE in terms of $\Omega \sigma_1$ and $\Omega \sigma_2$:
\begin{align}
D(\hat{\rho}_2\|\hat{\rho}_1)
&=
\frac{1}{2}
\left[
\log\left(
\frac{
\det\left(\Omega\sigma_1-\frac{i}{2}I\right)}
{\det\left(\Omega\sigma_2-\frac{i}{2}I\right)}
\right)
\right.
\nonumber\\
&\quad
+\trace\left[
\left(
f_{\rm RE}(\Omega\sigma_1)
-
f_{\rm RE}(\Omega\sigma_2)
\right)
\Omega\sigma_2
\right]
\nonumber\\
&\qquad\left.
+
\delta_d^T
f_{\rm RE}(\Omega\sigma_1)
\Omega\delta_d
\right].
\label{eq:QRE omega sigma}
\end{align}
Finally, using 
$\Omega\sigma_1=A+E_1$ and
$\Omega\sigma_2=A+E_2$, \eqref{eq:QRE omega sigma} becomes
\begin{align}
D(\hat{\rho}_2\|\hat{\rho}_1)
&=
\frac{1}{2}
\left[
\log\left(
\frac{
\det\left(A+E_1-\frac{i}{2}I\right)}
{\det\left(A+E_2-\frac{i}{2}I\right)}
\right)
\right.
\nonumber\\
&\
+\trace\left[
\left(
f_{\rm RE}(A+E_1)
-
f_{\rm RE}(A+E_2)
\right)
(A+E_2)
\right]
\nonumber\\
&\qquad\left.
+
\delta_d^T
f_{\rm RE}(A+E_1)
\Omega\delta_d
\right].
\label{eq:QRE Perturb}
\end{align}
Since $f_{\rm RE}'(A)$ is singular when $\sigma_0$ has a symplectic
eigenvalue equal to $1/2$, we first restrict the analysis to $\sigma_0$ whose symplectic eigenvalues are all greater than
$1/2$.

\subsection{Fr\'echet expansion}
We now derive the Fr\'echet expansion of \eqref{eq:QRE Perturb}
around $A$ in the directions $E_1$ and $E_2$. Up to the second order,
we have
\begin{align}
D_{\rm pert}^{(2)}(\hat{\rho}_2\|\hat{\rho}_1)
=
D_{\rm RE,CM}^{(1)}
+
D_{\rm RE,CM}^{(2)}
+
D_{\rm RE,d},
\label{eq:QRE perturb 2}
\end{align}
where $D_{\rm RE,CM}^{(1)}$ and $D_{\rm RE,CM}^{(2)}$ denote the
first- and second-order covariance-matrix contributions, respectively,
and $D_{\rm RE,d}$ denotes the mean-vector contribution. The zeroth order term is zero since the QRE vanishes when $\hat{\rho}_1$ and $\hat{\rho}_2$ are identical.
The first two terms of \eqref{eq:QRE Perturb} are the covariance-matrix
contribution. The Fr\'echet expansion of the first term of \eqref{eq:QRE Perturb} up to the
second-order is
\begin{align}
&
\log\left(
\frac{
\det\left(A+E_1-\frac{i}{2}I\right)}
{\det\left(A+E_2-\frac{i}{2}I\right)}
\right)
\nonumber\\
&\simeq
D^{(1)}_{\log\det(A-\frac{i}{2}I)}[E_1-E_2]
\nonumber\\
&\quad
+\frac{1}{2}
\left(
D^{(2)}_{\log\det(A-\frac{i}{2}I)}[E_1]
-
D^{(2)}_{\log\det(A-\frac{i}{2}I)}[E_2]
\right) \label{eq:QRE 0 expansion}
 \\
&=
\trace\left[
\left(A-\frac{i}{2}I\right)^{-1}(E_1-E_2)
\right]
\nonumber\\
&\quad
-\frac{1}{2}
\trace\left[
E_1\left(A-\frac{i}{2}I\right)^{-1}
E_1\left(A-\frac{i}{2}I\right)^{-1}
\right]
\nonumber\\
&\quad
+\frac{1}{2}
\trace\left[
E_2\left(A-\frac{i}{2}I\right)^{-1}
E_2\left(A-\frac{i}{2}I\right)^{-1}
\right],
\label{eq:QRE 1 expansion}
\end{align}
where \eqref{eq:QRE 0 expansion} follows from the second-order Fr\'echet expansion in \eqref{eq:Frechet expansion} with the zeroth order contribution vanishing because $\log\left(
\frac{
\det\left(A-\frac{i}{2}I\right)}
{\det\left(A-\frac{i}{2}I\right)}
\right)=0$, and \eqref{eq:QRE 1 expansion} follows from
Lemma~\ref{lemma:Frechet log det A}.

We next expand the second term of
\eqref{eq:QRE Perturb} up to the
second-order:
\begin{align}
&
\trace\left[
\left(
f_{\rm RE}(A+E_1)-f_{\rm RE}(A+E_2)
\right)(A+E_2)
\right]
\nonumber\\
&\simeq
\trace\left(
D_{f_{\rm RE}(A)}^{(1)}[E_1-E_2]A
\right)
+
\trace\left(
D_{f_{\rm RE}(A)}^{(1)}[E_1-E_2]E_2
\right)
\nonumber\\
&\quad
+\frac{1}{2}
\trace\left[
\left(
D_{f_{\rm RE}(A)}^{(2)}[E_1]
-
D_{f_{\rm RE}(A)}^{(2)}[E_2]
\right)A
\right]
\label{eq: QRE 2 expansion 2}\\
&=
\trace\left[
(E_1-E_2)A f'_{\rm RE}(A)
\right]
+
\trace\left(
D_{f_{\rm RE}(A)}^{(1)}[E_1-E_2]E_2
\right)
\nonumber\\
&\quad
+\frac{1}{2}
\trace\left[
E_1D_{(xf'_{\rm RE})(A)}^{(1)}[E_1]
-
E_2D_{(xf'_{\rm RE})(A)}^{(1)}[E_2]
\right]
\nonumber\\
&\quad
-\frac{1}{2}
\trace\left[
E_1D_{f_{\rm RE}(A)}^{(1)}[E_1]
-
E_2D_{f_{\rm RE}(A)}^{(1)}[E_2]
\right].
\label{eq:QRE 2 expansion}
\end{align}
Here, \eqref{eq: QRE 2 expansion 2} follows from the second-order  Fr\'echet expansion in \eqref{eq:Frechet expansion}, with the zeroth order term vanishing. In obtaining \eqref{eq:QRE 2 expansion}, the first and third terms of
\eqref{eq: QRE 2 expansion 2} are simplified using
Lemmas~\ref{lemma: trace} and \ref{lemma: trace 2}, respectively.

Collecting the first term in \eqref{eq:QRE 1 expansion} and \eqref{eq:QRE 2 expansion}, we obtain the first-order covariance-matrix contribution
\begin{align}
&D_{\rm RE,CM}^{(1)} \nonumber \\
&=
\frac{1}{2}
\trace\left[
\left(
\left(A-\frac{i}{2}I\right)^{-1}
+
Af'_{\rm RE}(A)
\right)
(E_1-E_2)
\right].
\label{eq:QRE D1}
\end{align}
Collecting the remaining terms in \eqref{eq:QRE 1 expansion} and \eqref{eq:QRE 2 expansion} gives the second-order covariance-matrix contribution
\begin{align}
&D_{\rm RE,CM}^{(2)} \nonumber \\
&=
-\frac{1}{4}
\trace\left[
E_1\left(A-\frac{i}{2}I\right)^{-1}
E_1\left(A-\frac{i}{2}I\right)^{-1}
\right]
\nonumber\\
&\quad
+\frac{1}{4}
\trace\left[
E_2\left(A-\frac{i}{2}I\right)^{-1}
E_2\left(A-\frac{i}{2}I\right)^{-1}
\right]
\nonumber\\
&\quad
+\frac{1}{2}
\trace\left(
D_{f_{\rm RE}(A)}^{(1)}[E_1]E_2
\right)
\nonumber\\
&\quad
+\frac{1}{4}
\trace\left[
E_1D_{(xf'_{\rm RE})(A)}^{(1)}[E_1]
-
E_2D_{(xf'_{\rm RE})(A)}^{(1)}[E_2]
\right]
\nonumber\\
&\quad
-\frac{1}{4}
\trace\left[
E_1D_{f_{\rm RE}(A)}^{(1)}[E_1]
+
E_2D_{f_{\rm RE}(A)}^{(1)}[E_2]
\right].
\label{eq:QRE D2}
\end{align}
Finally, the last term of \eqref{eq:QRE Perturb} gives the mean-vector
contribution. Since $\delta_d^T \delta_d$ is already second order in $ \|\delta_d \|$, its leading contribution is obtained by evaluating the matrix function at the reference matrix $A$:
\begin{align}
D_{\rm RE,d}
=
\frac{1}{2}
\delta_d^T f_{\rm RE}(A)\Omega\delta_d.
\label{eq:QRE 3 expansion}
\end{align}

\subsection{First-order cancellation}
In this subsection, we show that the first-order term in the Fr\'echet expansion of the covariance-matrix contribution vanishes
\begin{align}
    D_{\rm RE,CM}^{(1)}=0, \label{eq:QRE first p = 0}
\end{align}
where $D_{\rm RE,CM}^{(1)}$ is given in \eqref{eq:QRE D1}.
Recall that
\begin{align}
f_{\rm RE}(x)
=
2i\,\operatorname{arccoth}(2ix),
\end{align}
we have
\begin{align}
f_{\rm RE}'(A)
=
-\left(A-\frac{i}{2}I\right)^{-1}
\left(A+\frac{i}{2}I\right)^{-1}.
\end{align}
Therefore,
\begin{align}
&\left(A-\frac{i}{2}I\right)^{-1}
+
A f_{\rm RE}'(A)
\nonumber\\
&=
\left(A-\frac{i}{2}I\right)^{-1}
-
A
\left(A-\frac{i}{2}I\right)^{-1}
\left(A+\frac{i}{2}I\right)^{-1}
 \\
&=
\left(A-\frac{i}{2}I\right)^{-1}
\left(A+\frac{i}{2}I\right)^{-1}
\left(A+\frac{i}{2}I-A\right)
 \\
&=
\frac{i}{2}
\left(A^2+\frac{1}{4}I\right)^{-1}.
\label{eq: QRE t1 e1}
\end{align}
Substituting \eqref{eq: QRE t1 e1} into \eqref{eq:QRE D1},
we obtain
\begin{align}
D_{\rm RE,CM}^{(1)}
&=
\frac{i}{4}
\trace\left[
\left(A^2+\frac{1}{4}I\right)^{-1}
(E_1-E_2)
\right]
\nonumber\\
&=
\frac{i}{4}
\trace\left[
\left(A^2+\frac{1}{4}I\right)^{-1}
\Omega(\Delta\sigma_1-\Delta\sigma_2)
\right],
\label{eq:QRE t1 = 0}
\end{align}
where we used $E_k=\Omega\Delta\sigma_k$.
The matrix $\Delta\sigma_1-\Delta\sigma_2$ is symmetric, and we further
show that
$\left(A^2+\frac{1}{4}I\right)^{-1}\Omega$
is antisymmetric.
Using $A=\Omega\sigma_0$ and $\sigma_0^T=\sigma_0$, we have
\begin{align}
\left(A^2+\frac{1}{4}I\right)\Omega
&=
\Omega\sigma_0\Omega\sigma_0\Omega
+\frac{1}{4}\Omega
\nonumber\\
&=
\Omega\left((A^T)^2+\frac{1}{4}I\right).
\label{eq:QRE step 1}
\end{align}
Taking the inverse of both sides of \eqref{eq:QRE step 1}, and multiplying by $\Omega$ from both left and right, we have
\begin{align}
\left(A^2+\frac{1}{4}I\right)^{-1}\Omega
&=
\Omega
\left((A^T)^2+\frac{1}{4}I\right)^{-1}
\nonumber\\
&=
-\Omega^T
\left[
\left(A^2+\frac{1}{4}I\right)^{-1}
\right]^T
\nonumber\\
&=
-\left[
\left(A^2+\frac{1}{4}I\right)^{-1}\Omega
\right]^T.
\end{align}
Therefore,
$\left(A^2+\frac{1}{4}I\right)^{-1}\Omega$
is antisymmetric. Since the trace of the product of an antisymmetric
matrix and a symmetric matrix is zero, we obtain
\begin{align}
\trace\left[
\left(A^2+\frac{1}{4}I\right)^{-1}
\Omega(\Delta\sigma_1-\Delta\sigma_2)
\right]
=0.
\label{eq: QRE t1 0}
\end{align}
Substituting \eqref{eq: QRE t1 0} into \eqref{eq:QRE t1 = 0} yields \eqref{eq:QRE first p = 0}.

\subsection{Second-order simplification}
In this subsection, we simplify the second-order term $D_{\rm RE,CM}^{(2)}$.
The eigen-decomposition of $A = \Omega \sigma_0$ satisfies
\begin{align}
    A=U\Gamma U^{-1}, \label{eq:QRE t2=0 80}
\end{align}
where
\begin{align}
    \Gamma=\operatorname{diag}\{a_k,-a_k\}_{k=1}^{n}, \label{eq:QRE t2=0 81}
\end{align}
and $n$ is the number of modes.
Define the perturbation in the eigenbasis of $A$ as
\begin{align}
    \tilde{E}_k=U^{-1}E_kU. \label{eq:QRE t2=0 82}
\end{align}
Using the linearity of the Fréchet derivative with respect to the matrix function, and
\begin{align}
x f'_{\rm RE}(x)
=
-\frac{4x}{1+4x^2}
=
-\frac{1}{2x+i}
-\frac{1}{2x-i},
\end{align}
we obtain
\begin{align}
D^{(1)}_{(xf'_{\rm RE})(A)}[E_k]
&=
D^{(1)}_{-\frac{1}{2A+iI}}[E_k]
+
D^{(1)}_{-\frac{1}{2A-iI}}[E_k]. \label{eq:QRE t2=0 88}
\end{align}
Applying Lemma~\ref{lemma:Df-1} and using $D^{(1)}_{- (2A\pm iI)}[E_k] = 2 E_k $, the terms in \eqref{eq:QRE t2=0 88} yield
\begin{align}
D^{(1)}_{-\frac{1}{2A\pm iI}}[E_k]=
\left(2A\pm iI\right)^{-1}
(2E_k)
\left(2A\pm iI\right)^{-1}. \label{eq:QRE t2=0 89}
\end{align}
For $k'=1,2$, we have
\begin{align}
&\trace\left[
E_{k'}
D^{(1)}_{-\frac{1}{2A+iI}}[E_{k'}]
\right]
\nonumber\\
&=
\frac{1}{2}
\trace\left[
E_{k'}
\left(A+\frac{i}{2}I\right)^{-1}
E_{k'}
\left(A+\frac{i}{2}I\right)^{-1}
\right] \label{eq:QRE t2=0 90}
\\
&=
\frac{1}{2}
\sum_{k,l=1}^{n}
(\widetilde E_{k'})_{kl}
(\widetilde E_{k'})_{lk}
\left(a_l+\frac{i}{2}\right)^{-1}
\left(a_k+\frac{i}{2}\right)^{-1}
\nonumber\\
&\quad+
\frac{1}{2}
\sum_{k,l=1}^{n}
(\widetilde E_{k'})_{kl}
(\widetilde E_{k'})_{lk}
\left(-a_l+\frac{i}{2}\right)^{-1}
\left(-a_k+\frac{i}{2}\right)^{-1} \label{eq:QRE t2=0 91}
\\
&=
\frac{1}{2}
\sum_{k,l=1}^{n}
(\widetilde E_{k'})_{kl}
(\widetilde E_{k'})_{lk}
\left(-a_l-\frac{i}{2}\right)^{-1}
\left(-a_k-\frac{i}{2}\right)^{-1}
\nonumber\\
&\quad+
\frac{1}{2}
\sum_{k,l=1}^{n}
(\widetilde E_{k'})_{kl}
(\widetilde E_{k'})_{lk}
\left(a_l-\frac{i}{2}\right)^{-1}
\left( a_k-\frac{i}{2}\right)^{-1}
\\
&=
\frac{1}{2}
\trace\left[
E_{k'}
\left(A-\frac{i}{2}I\right)^{-1}
E_{k'}
\left(A-\frac{i}{2}I\right)^{-1}
\right]\label{eq:QRE t2=0 100} \\
&=\trace\left[
E_{k'}
D^{(1)}_{-\frac{1}{2A-iI}}[E_{k'}]
\right].\label{eq:QRE t2=0 101}
\end{align}
Here, \eqref{eq:QRE t2=0 90} and \eqref{eq:QRE t2=0 101}
follow from \eqref{eq:QRE t2=0 89}. Equation
\eqref{eq:QRE t2=0 91} follows by expressing the trace in the
eigenbasis of $A$ using \eqref{eq:QRE t2=0 80}--\eqref{eq:QRE t2=0 82}.
Combining \eqref{eq:QRE t2=0 88}, \eqref{eq:QRE t2=0 100}, and
\eqref{eq:QRE t2=0 101}, we obtain
\begin{align}
&\trace\left[
E_{k'}
D^{(1)}_{(xf'_{\rm RE})(A)}[E_{k'}]
\right]
\nonumber\\
&=
\trace\left[
E_{k'}
\left(A-\frac{i}{2}I\right)^{-1}
E_{k'}
\left(A-\frac{i}{2}I\right)^{-1}
\right].
\end{align}
Thus, the first two terms of
\eqref{eq:QRE D2} cancel the fourth term of
\eqref{eq:QRE D2}.
Furthermore, we consider the third and fifth terms of \eqref{eq:QRE D2}:
\begin{align}
&\frac{1}{2}
\trace\left(
D^{(1)}_{f_{\rm RE}(A)}[E_1]E_2
\right)
\nonumber\\
&\quad
-\frac{1}{4}
\trace\left[
E_1D^{(1)}_{f_{\rm RE}(A)}[E_1]
+
E_2D^{(1)}_{f_{\rm RE}(A)}[E_2]
\right]
\nonumber\\
&=
\frac{1}{4}
\sum_{k,l}
\left(f_{\rm RE}^{[1]}(\Gamma)\right)_{kl}
\Big[
2(\widetilde E_1)_{kl}(\widetilde E_2)_{lk}
\nonumber\\
&\qquad
-(\widetilde E_1)_{kl}(\widetilde E_1)_{lk}
-(\widetilde E_2)_{kl}(\widetilde E_2)_{lk}
\Big] \label{eq:QRE t2=0 111}
\\
&=
\frac{1}{4}
\sum_{k,l}
\left(f_{\rm RE}^{[1]}(\Gamma)\right)_{kl}
\Big[
(\widetilde E_1)_{kl}(\widetilde E_2)_{lk}
+
(\widetilde E_2)_{kl}(\widetilde E_1)_{lk}
\nonumber\\
&\qquad
-(\widetilde E_1)_{kl}(\widetilde E_1)_{lk}
-(\widetilde E_2)_{kl}(\widetilde E_2)_{lk}
\Big] \label{eq:QRE t2=0 112}
\\
&=
-\frac{1}{4}
\sum_{k,l}
\left(f_{\rm RE}^{[1]}(\Gamma)\right)_{kl}
(\widetilde E_1-\widetilde E_2)_{kl}
(\widetilde E_1-\widetilde E_2)_{lk}
\\
&=
-\frac{1}{4}
\trace\left[
(E_1-E_2)
D^{(1)}_{f_{\rm RE}(A)}[E_1-E_2]
\right],\label{eq:QRE t2=0 113}
\end{align}
where \eqref{eq:QRE t2=0 111} is obtained  by the definition of the divided difference representation in \eqref{eq:First Frechet derivative Hadmard product} and expressing the trace in the eigenbasis of $A$ using \eqref{eq:QRE t2=0 80}--\eqref{eq:QRE t2=0 82}. Since $f_{\rm RE}^{[1]}(\Gamma)$ is symmetric, \eqref{eq:QRE t2=0 112} follows from splitting the first term of \eqref{eq:QRE t2=0 111} into two identical sums and relabeling the second one.

Therefore, after the cancellation of the first, second, and fourth
terms of \eqref{eq:QRE D2}, substituting \eqref{eq:QRE t2=0 113} into \eqref{eq:QRE D2} yields
\begin{align}
D_{\rm RE,CM}^{(2)}
=
-\frac{1}{4}
\trace\left[
(E_1-E_2)
D^{(1)}_{f_{\rm RE}(A)}[E_1-E_2]
\right].
\label{eq:QRE D2 final}
\end{align}

\subsection{Perturbative formula}

Combining the first-order cancellation, the second-order
covariance-matrix contribution, and the mean-vector contribution,
the Fr\'echet expansion of the QRE up to the second-order is
\begin{align}
D_{\rm pert}^{(2)}
\left(\hat{\rho}_2\|\hat{\rho}_1\right)
&=
-\frac{1}{4}
\trace\left[
(E_1-E_2)
D^{(1)}_{f_{\rm RE}(A)}[E_1-E_2]
\right]
\nonumber\\
&\quad+
\frac{1}{2}
\delta_d^T f_{\rm RE}(A)\Omega\delta_d.
\end{align}

\section{Quantum Chernoff quantity}\label{appendix:gaussian_qi frechet QCE}

\subsection{Gaussian covariance-matrix representation}
\label{appendix:sub gaussian AE form}
We first rewrite the quantum Chernoff quantity in terms of
$\Omega\sigma_1$ and $\Omega\sigma_2$. From \eqref{eq:QCE}, the
quantum Chernoff quantity is
\begin{align}
\xi_s(\rho_1,\rho_2)
=
-\log\left(Q_s^{(1)}Q_s^{(2)}\right),
\end{align}
where
\begin{align}
Q_s^{(1)}
&=
\frac{
2^n
\prod_{k=1}^n
G^{(s)}(\nu_k^{(1)})
G^{(1-s)}(\nu_k^{(2)})
}{
\sqrt{\det[V_1(s)+V_2(1-s)]}
},
\label{eq:app QCE 1}
\\
Q_s^{(2)}
&=
\exp\left[
-\delta_d^T
\left(V_1(s)+V_2(1-s)\right)^{-1}
\delta_d
\right].
\label{eq:app QCE 2}
\end{align}
We first consider $-\log Q_s^{(1)}$ and recall that
\begin{align}
f_{\rm CE}^{(s)}(x)
=
G^{(s)}\left(\sqrt{-x^2}\right),
\end{align}
where $G^{(s)}$ is defined in \eqref{eq: QCE Gs}.
The first product in the numerator of \eqref{eq:app QCE 1} can then
be written as
\begin{align}
\prod_{k=1}^n G^{(s)}(\nu_k^{(1)})
&=
\sqrt{
\det\left[
G^{(s)}\left(\sqrt{-(\Omega\sigma_1)^2}\right)
\right]
}
\label{eq:QCE AE Q1 1}
\\
&=
\sqrt{
\det\left[
f_{\rm CE}^{(s)}(\Omega\sigma_1)
\right]
},
\end{align}
where \eqref{eq:QCE AE Q1 1} follows from the fact that
$\sqrt{-(\Omega\sigma_1)^2}$ has eigenvalues
$\{\nu_k^{(1)}\}_{k=1}^n$, each with multiplicity two.
Similarly,
\begin{align}
\prod_{k=1}^n G^{(1-s)}(\nu_k^{(2)})
=
\sqrt{
\det\left[
f_{\rm CE}^{(1-s)}(\Omega\sigma_2)
\right]
}.
\end{align}
Therefore, the numerator of
\eqref{eq:app QCE 1} to $-\log Q_s^{(1)}$ yields
\begin{align}
&-n\log 2
-\frac{1}{2}
\log\det\left[
f_{\rm CE}^{(s)}(\Omega\sigma_1) 
f_{\rm CE}^{(1-s)}(\Omega\sigma_2)
\right].
\label{eq:QCE 1 numerator}
\end{align}
We next consider the denominator of \eqref{eq:app QCE 1}. For the symplectic eigen-decomposition of $\Omega \sigma_1$ in \eqref{eq:Williamson decomposition 1} and \eqref{eq:Williamson decomposition 2},
since
$\Lambda_1
=
-\Omega\Lambda_1\Omega$, 
we have
\begin{align}
\Lambda_1
=
\sqrt{\Lambda_1^2}
=
\sqrt{-\Lambda_1\Omega\Lambda_1\Omega}
=
\sqrt{-(\Lambda_1\Omega)^2}.
\label{eq:QCE V1 -3}
\end{align}
Consequently,
\begin{align}
S_1\Lambda_1S_1^{-1}
&=
S_1\sqrt{-(\Lambda_1\Omega)^2}S_1^{-1}
\label{eq:QCE V1 -2}
\\
&=
\sqrt{
-\left(S_1\Lambda_1\Omega S_1^{-1}\right)^2
}
\label{eq:QCE V1 -1}
\\
&=
\sqrt{
-\left(S_1\Lambda_1S_1^T\Omega\right)^2
}
\label{eq:QCE V1 0}
\\
&=
\sqrt{-(\sigma_1\Omega)^2}.
\label{eq:QCE V1 00}
\end{align}
Here, \eqref{eq:QCE V1 -2} follows from
\eqref{eq:QCE V1 -3}, \eqref{eq:QCE V1 -1} follows from the
similarity-invariance of primary matrix functions, and
\eqref{eq:QCE V1 0} follows from
$\Omega=S_1^T \Omega S_1$.
We can now rewrite $V_1(s)$ defined in \eqref{eq: QCE V_1} as
\begin{align}
V_1(s)
&=
S_1F^{(s)}(\Lambda_1)S_1^T
\\
&=
S_1F^{(s)}(\Omega\Lambda_1\Omega^{-1})S_1^T
\label{eq: V1 1}
\\
&=
S_1\Omega F^{(s)}(\Lambda_1)\Omega^{-1}S_1^T
\label{eq: V1 2}
\\
&=
-S_1\Omega F^{(s)}(\Lambda_1)\Omega S_1^T
\label{eq: V1 3}
\\
&=
-S_1\Omega F^{(s)}(\Lambda_1)S_1^{-1}\Omega
\label{eq: V1 4}
\\
&=
-S_1\Omega S_1^{-1}
S_1F^{(s)}(\Lambda_1)S_1^{-1}\Omega
\label{eq: V1 5}
\\
&=
-S_1\Omega S_1^{-1}
F^{(s)}(S_1\Lambda_1S_1^{-1})\Omega
\label{eq: V1 6}
\\
&=
-S_1\Omega S_1^{-1}
F^{(s)}\left(\sqrt{-(\sigma_1\Omega)^2}\right)\Omega,
\label{eq: V1 7}
\end{align}
where \eqref{eq: V1 1} follows from $\Omega\Lambda_1 \Omega^{-1} = \Lambda_1$, \eqref{eq: V1 2} follows from the similarity-invariance, \eqref{eq: V1 3} follows from $\Omega^{-1} = -\Omega$, \eqref{eq: V1 4} follows from $\Omega S_1^T = S_1^{-1}\Omega$, \eqref{eq: V1 5} follows from inserting $S^{-1}_1 S_1 = I$, \eqref{eq: V1 6} follows from similarity invariance, and \eqref{eq: V1 7} follows from \eqref{eq:QCE V1 00}.

It remains to rewrite $S_1\Omega S_1^{-1}$  in \eqref{eq: V1 7} in terms of
$\sigma_1\Omega$. We have
\begin{align}
S_1\Omega S_1^{-1}
&=
S_1\Lambda_1\Omega\Lambda_1^{-1}S_1^{-1}
\label{eq:QCE V1 6}
\\
&=
S_1\Lambda_1S_1^T
\Omega
S_1\Lambda_1^{-1}S_1^{-1}
\label{eq:QCE V1 7}
\\
&=
\sigma_1\Omega
\left(S_1\Lambda_1S_1^{-1}\right)^{-1}
\label{eq:V1 6}
\\
&=
\sigma_1\Omega
\left(\sqrt{-(\sigma_1\Omega)^2}\right)^{-1},
\label{eq:V1 7}
\end{align}
where \eqref{eq:QCE V1 6} follows from $ \Omega = \Lambda_1\Omega \Lambda_1^{-1}$, \eqref{eq:QCE V1 7} follows from $S_1^T\Omega S_1 = \Omega $, and \eqref{eq:V1 7} follows from \eqref{eq:QCE V1 00}.
Substituting \eqref{eq:V1 7} into \eqref{eq: V1 7} gives
\begin{align}
V_1(s)
=
-(\sigma_1\Omega)
\left(\sqrt{-(\sigma_1\Omega)^2}\right)^{-1}
F^{(s)}\left(\sqrt{-(\sigma_1\Omega)^2}\right)
\Omega.
\label{eq:QCE V1}
\end{align}
Define the scalar function 
\begin{align}
g_{\rm CE}^{(s)}(x)
=
-x
\left(\sqrt{-x^2}\right)^{-1}
F^{(s)}\left(\sqrt{-x^2}\right).
\end{align}
Then, \eqref{eq:QCE V1} becomes
\begin{align}
V_1(s)
&=
g_{\rm CE}^{(s)}(\sigma_1\Omega)\Omega
\\
&=
g_{\rm CE}^{(s)}
\left(
\Omega(\Omega\sigma_1)\Omega^{-1}
\right)\Omega
\\
&=
\Omega g_{\rm CE}^{(s)}(\Omega\sigma_1).
\label{eq:QCE V1 2}
\end{align}
Similarly,
\begin{align}
V_2(1-s)
=
\Omega g_{\rm CE}^{(1-s)}(\Omega\sigma_2).
\label{eq:QCE V2}
\end{align}
Substituting \eqref{eq:QCE V1 2} and \eqref{eq:QCE V2} into the denominator of \eqref{eq:app QCE 1} yields
\begin{align}
 \sqrt{\det\left[
g_{\rm CE}^{(s)}(\Omega\sigma_1)
+
g_{\rm CE}^{(1-s)}(\Omega\sigma_2)
\right]}.
\label{eq: QCE denominator}
\end{align}
Combining \eqref{eq:QCE 1 numerator} and
\eqref{eq: QCE denominator}, we obtain the covariance-matrix contribution in terms of $\Omega \sigma_1$ and $\Omega \sigma_2$:
 \begin{align}
     -\log Q_s^{(1)} 
    & = - n\log 2- \frac{1}{2}\left( \log \det\left( f^{(s)}_{\rm CE}(\Omega\sigma_1) \right)  \right.\nonumber \\
    &\left.\quad -  \log \det\left( f^{(1-s)}_{\rm CE}(\Omega\sigma_2)  \right) \right.\nonumber \\
    & \left.\quad + \log \det[  g_{\rm CE}^{(s)}( \Omega\sigma_1)  + g_{\rm CE}^{(1-s)}( \Omega\sigma_2)  ] \right).\label{eq:QCE CM omega}
\end{align}
For the mean-vector contribution, substituting 
\eqref{eq:QCE V1 2}  and \eqref{eq:QCE V2} into \eqref{eq:app QCE 2}, we obtain
\begin{align}
-\log Q_s^{(2)}
&=
\delta_d^T
\left[
V_1(s)+V_2(1-s)
\right]^{-1}
\delta_d
\nonumber\\
&=
\delta_d^T
\left[
\Omega
\left(
g_{\rm CE}^{(s)}(\Omega\sigma_1)
+
g_{\rm CE}^{(1-s)}(\Omega\sigma_2)
\right)
\right]^{-1}
\delta_d.
\label{eq:QCE d omega}
\end{align}
Finally, writing
$\Omega\sigma_1=A+E_1$, 
$\Omega\sigma_2=A+E_2$, 
the covariance-matrix contribution becomes
\begin{align}
-\log Q_s^{(1)}
&= \frac{1}{2}
\log\det\left[
g_{\rm CE}^{(s)}(A+E_1)
+
g_{\rm CE}^{(1-s)}(A+E_2)
\right]\nonumber\\
&\ -\frac{1}{2}
\log\det\left[
f_{\rm CE}^{(1-s)}(A+E_2)\right]
\nonumber\\
&\
-\frac{1}{2}
\log\det\left[
f_{\rm CE}^{(s)}(A+E_1)
\right]-n\log 2,
\label{eq:log Q bar}
\end{align}
while the mean-vector contribution becomes
\begin{align}
-\log Q_s^{(2)}
&=
\delta_d^T
\left[
\Omega
\left(
g_{\rm CE}^{(s)}(A+E_1) \right.\right.\nonumber\\
&\qquad\qquad\left.\left.+
g_{\rm CE}^{(1-s)}(A+E_2)
\right)
\right]^{-1}
\delta_d.
\label{eq:log Q bar 2}
\end{align}
Since $(g_{\rm CE}^{(s)})'(A)$ is singular when $\sigma_0$ has a symplectic
eigenvalue equal to $1/2$, we first restrict the analysis to $\sigma_0$ whose symplectic eigenvalues are all greater than
$1/2$.

\subsection{Fr\'echet expansion}

We now derive the Fr\'echet expansion of the quantum Chernoff quantity
up to the second order. We write
\begin{align}
\xi_{\rm pert,s}^{(2)}(\hat{\rho}_1,\hat{\rho}_2)
=
\sum_{k=0}^{2}\xi_{\rm CE,CM,s}^{(k)}
+
\xi_{\rm CE,d,s},
\label{eq:QCE perturb expansion}
\end{align}
where $\xi_{\rm CE,CM,s}^{(k)}$ denotes the $k$th-order
covariance-matrix contribution,
and $\xi_{\rm CE,d,s}$ denotes the mean-vector contribution.

We first consider the mean-vector contribution $-\log Q_s^{(2)} $ in \eqref{eq:log Q bar 2}. Since the
mean-vector term is already second order in $\|\delta_d\|$, we retain
only the leading-order covariance matrix term,
\begin{align}
\xi_{\rm CE,d,s}
=  
\delta_d^T
\left[
\Omega
\left(
g_{\rm CE}^{(s)}(A)
+
g_{\rm CE}^{(1-s)}(A)
\right)
\right]^{-1}
\delta_d.
\label{eq:QCE perturb mean-vector}
\end{align}
We next consider the covariance-matrix contribution in
\eqref{eq:log Q bar}. The second and third terms of \eqref{eq:log Q bar} admit the following second-order 
Fr\'echet expansions
\begin{align}
&-\frac{1}{2}
\left[
\log\det\!\left(f_{\rm CE}^{(s)}(A+E_1)\right)
\!+\!
\log\det\!\left(f_{\rm CE}^{(1-s)}(A+E_2)\right)
\right]
\nonumber\\
&\simeq
-\frac{1}{2}
\left[
\log\det\left(f_{\rm CE}^{(s)}(A)\right)
+
\log\det\left(f_{\rm CE}^{(1-s)}(A)\right)
\right]
\nonumber\\
&\
-\frac{1}{2}
\left[
D^{(1)}_{\log\det(f_{\rm CE}^{(s)}(A))}[E_1]
+
D^{(1)}_{\log\det(f_{\rm CE}^{(1-s)}(A))}[E_2]
\right]
\nonumber\\
&\
-\frac{1}{4}  \!
\left[
D^{(2)}_{\log\det(f_{\rm CE}^{(s)}(A))}[E_1]
\!+\!
D^{(2)}_{\log\det(f_{\rm CE}^{(1-s)}(A))}[E_2]
\right].
\label{eq:QCE t1 0}
\end{align}
For the first term of \eqref{eq:log Q bar}, we first expand
$g_{\rm CE}^{(s)}(A+E_1)+g_{\rm CE}^{(1-s)}(A+E_2)$ and then
apply Lemma~\ref{lemma:Frechet log det A}. Up to the second order,
\begin{align}
&g_{\rm CE}^{(s)}(A+E_1)
+
g_{\rm CE}^{(1-s)}(A+E_2)
\nonumber\\
&\simeq
g_{\rm CE}^{(s)}(A)
+
D^{(1)}_{g_{\rm CE}^{(s)}(A)}[E_1]
+
\frac{1}{2}
D^{(2)}_{g_{\rm CE}^{(s)}(A)}[E_1]
\nonumber\\
&\quad
+
g_{\rm CE}^{(1-s)}(A)
+
D^{(1)}_{g_{\rm CE}^{(1-s)}(A)}[E_2]
+
\frac{1}{2}
D^{(2)}_{g_{\rm CE}^{(1-s)}(A)}[E_2]
\nonumber\\
&=
R_{\rm CE}^{(s)}(A)
+
\Delta_{\rm CE}^{(1)}
+
\Delta_{\rm CE}^{(2)},
\label{eq:QCE t1 1}
\end{align}
where
\begin{align}
R_{\rm CE}^{(s)}(x)
&=
g_{\rm CE}^{(s)}(x)
+
g_{\rm CE}^{(1-s)}(x),
\label{eq:QCE RCB}\\
\Delta_{\rm CE}^{(1)}
&=
D^{(1)}_{g_{\rm CE}^{(s)}(A)}[E_1]
+
D^{(1)}_{g_{\rm CE}^{(1-s)}(A)}[E_2],
\label{eq:QCE delta 1}\\
\Delta_{\rm CE}^{(2)}
&=
\frac{1}{2}
D^{(2)}_{g_{\rm CE}^{(s)}(A)}[E_1]
+
\frac{1}{2}
D^{(2)}_{g_{\rm CE}^{(1-s)}(A)}[E_2].
\label{eq:QCE delta 2}
\end{align}
Substituting \eqref{eq:QCE t1 1} into the first term of
\eqref{eq:log Q bar} and applying
Lemma~\ref{lemma:Frechet log det A}, we obtain
\begin{align}
&\frac{1}{2}\log\det\left[
R_{\rm CE}^{(s)}(A)
+
\Delta_{\rm CE}^{(1)}
+
\Delta_{\rm CE}^{(2)}
\right]
\nonumber\\
&\simeq\frac{1}{2}\left(
\log\det\left[R_{\rm CE}^{(s)}(A)\right]
+
\trace\left[
(1/R_{\rm CE}^{(s)})(A)
\Delta_{\rm CE}^{(1)}
\right] \right.
\nonumber\\
&\left.\quad
-\frac{1}{2}
\trace\left[
(1/R_{\rm CE}^{(s)})(A)
\Delta_{\rm CE}^{(1)}
(1/R_{\rm CE}^{(s)})(A)
\Delta_{\rm CE}^{(1)}
\right]\right.
\nonumber\\
&\left.\quad
+
\trace\left[
(1/R_{\rm CE}^{(s)})(A)
\Delta_{\rm CE}^{(2)}
\right]\right).
\label{eq:QCE t4}
\end{align}
We now collect terms according to their perturbative order.
The zeroth-order covariance-matrix contribution consists of the
first two terms in \eqref{eq:QCE t1 0} and the first term in \eqref{eq:QCE t4},
together with the constant $-n\log 2$ in \eqref{eq:log Q bar}. Since the two covariance
matrices coincide at zeroth order, the corresponding quantum
Chernoff exponent vanishes. Therefore,
\begin{align}
\xi_{\rm CE,CM,s}^{(0)}=0.
\end{align}
The first-order covariance-matrix contribution consists of the
third and the fourth terms in \eqref{eq:QCE t1 0} and the second term in \eqref{eq:QCE t4},
which gives
\begin{align}
&\xi_{\rm CE,CM,s}^{(1)}
=
\frac{1}{2}
\left[
\trace\left(
(1/R_{\rm CE}^{(s)})(A)
\Delta_{\rm CE}^{(1)}
\right)
\right.
\nonumber\\
&\qquad\qquad \qquad\left.
-
D^{(1)}_{\log\det(f_{\rm CE}^{(s)}(A))}[E_1]\right.
\nonumber\\
&\qquad \qquad \qquad \left.
-
D^{(1)}_{\log\det(f_{\rm CE}^{(1-s)}(A))}[E_2]
\right].
\label{eq:QCE term 1}
\end{align}
Similarly, the second-order covariance-matrix contribution consists
of the fifth and the sixth terms in \eqref{eq:QCE t1 0} and third and fourth terms in 
\eqref{eq:QCE t4}, yielding
\begin{align}
\xi_{\rm CE,CM,s}^{(2)}
&=
\frac{1}{2}
\trace\left[
(1/R_{\rm CE}^{(s)})(A)
\Delta_{\rm CE}^{(2)}
\right] \nonumber\\
&\quad-\frac{1}{4}
D^{(2)}_{\log\det(f_{\rm CE}^{(s)}(A))}[E_1]
\nonumber\\
&\quad
-\frac{1}{4}
D^{(2)}_{\log\det(f_{\rm CE}^{(1-s)}(A))}[E_2]
\nonumber\\
&\quad
-\frac{1}{4}
\trace\left[
\left(
(1/R_{\rm CE}^{(s)})(A)
\Delta_{\rm CE}^{(1)}
\right)^2
\right].
\label{eq:QCE term 2}
\end{align}

\subsection{First-order cancellation}

We now show that the first-order covariance-matrix contribution
vanishes, i.e.,
\begin{align}
\xi_{\rm CE,CM,s}^{(1)}=0.
\end{align}
Using Lemma~\ref{lemma:D logdet f}, the second and third terms of
\eqref{eq:QCE term 1} can be written as
\begin{align}
D^{(1)}_{\log\det(f_{\rm CE}^{(s)}(A))}[E_1]
&=
\trace\left[
E_1
\left(\log f_{\rm CE}^{(s)}\right)'(A)
\right], \label{eq:QCE t1-2}
\\
D^{(1)}_{\log\det(f_{\rm CE}^{(1-s)}(A))}[E_2]
&=\trace\left[
E_2
\left(\log f_{\rm CE}^{(1-s)}\right)'(A)
\right].  \label{eq:QCE t1-1}
\end{align}
Next, substituting \eqref{eq:QCE delta 1} into \eqref{eq:QCE term 1}, the first term of \eqref{eq:QCE term 1} yields
\begin{align}
&
\frac{1}{2}
\Bigg[
\trace\left[
(1/R_{\rm CE}^{(s)})(A)
\left(
D^{(1)}_{g_{\rm CE}^{(s)}(A)}[E_1]
+
D^{(1)}_{g_{\rm CE}^{(1-s)}(A)}[E_2]
\right)
\right] \nonumber \\
& =\frac{1}{2}
\left[
\trace\left[
E_1
(1/R_{\rm CE}^{(s)})(A)
\left(g_{\rm CE}^{(s)}\right)'(A)
\right]\right.
\nonumber\\
&\left.\qquad
+
\trace\left[
E_2
(1/R_{\rm CE}^{(s)})(A)
\left(g_{\rm CE}^{(1-s)}\right)'(A)
\right]\right] \label{eq:QCE proof t11} \\
& =\frac{1}{2}
\left[
\trace\left[
E_1
\left(\log f_{\rm CE}^{(s)}\right)'(A)
\right]\right.
\nonumber\\
&\left.\qquad
+
\trace\left[
E_2
\left(\log f_{\rm CE}^{(1-s)}\right)'(A)
\right]\right], \label{eq:QCE proof t12}
\end{align}
where \eqref{eq:QCE proof t11} follows from applying Lemma~\ref{lemma: trace} to the first two terms. Eq.~\eqref{eq:QCE proof t12} follows from
\begin{align}
(1/R_{\rm CE}^{(s)})(x)
\left(g_{\rm CE}^{(s)}\right)'(x)
&=
\left(\log f_{\rm CE}^{(s)}\right)'(x),
\label{eq:QCE equal 1}
\\
(1/R_{\rm CE}^{(s)})(x
\left(g_{\rm CE}^{(1-s)}\right)'(x)
&=
\left(\log f_{\rm CE}^{(1-s)}\right)'(x).
\label{eq:QCE equal 2}
\end{align}
Substituting \eqref{eq:QCE t1-2}, \eqref{eq:QCE t1-1}, and \eqref{eq:QCE proof t12} into \eqref{eq:QCE term 1}, the four terms are canceled pairwise. Therefore, 
\begin{align}
\xi_{\rm CE,CM,s}^{(1)}=0.
\end{align}

\subsection{Second-order simplification}
We now simplify the second-order covariance-matrix contribution
$\xi_{\rm CE,CM,s}^{(2)}$ given in \eqref{eq:QCE term 2}.
Using Lemma~\ref{lemma:D logdet f}, the second and third terms of
\eqref{eq:QCE term 2} become
\begin{align}
&-\frac{1}{4}
D^{(2)}_{\log\det(f_{\rm CE}^{(s)}(A))}[E_1] \nonumber \\
&\qquad\qquad\qquad =
-\frac{1}{4}
\trace\left[
E_1
D^{(1)}_{(\log f_{\rm CE}^{(s)})'(A)}[E_1]
\right],
\label{eq:QCE t7}
\\
&-\frac{1}{4}
D^{(2)}_{\log\det(f_{\rm CE}^{(1-s)}(A))}[E_2] \nonumber \\
&\qquad\qquad\qquad=
-\frac{1}{4}
\trace\left[
E_2
D^{(1)}_{(\log f_{\rm CE}^{(1-s)})'(A)}[E_2]
\right].
\label{eq:QCE t8}
\end{align}
Next, substituting \eqref{eq:QCE delta 2} into the first term of
\eqref{eq:QCE term 2} and applying Lemma~\ref{lemma: trace 2}, we have
\begin{align}
&\frac{1}{2}\trace\left[
(1/R_{\rm CE}^{(s)})(A)
\Delta_{\rm CE}^{(2)}
\right]
\nonumber\\
&=
\frac{1}{4}
\trace\left[
(1/R_{\rm CE}^{(s)})(A)
\left(
D^{(2)}_{g_{\rm CE}^{(s)}(A)}[E_1]
+
D^{(2)}_{g_{\rm CE}^{(1-s)}(A)}[E_2]
\right)
\right]
\nonumber\\
&=
\frac{1}{4}
\trace\left[
E_1
D^{(1)}_{
\left(
(1/R_{\rm CE}^{(s)})
(g_{\rm CE}^{(s)})'
\right)(A)
}[E_1]
\right]
\nonumber\\
&\quad
-\frac{1}{4}
\trace\left[
D^{(1)}_{(1/R_{\rm CE}^{(s)})(A)}[E_1]
D^{(1)}_{g_{\rm CE}^{(s)}(A)}[E_1]
\right]
\nonumber\\
&\quad
+\frac{1}{4}
\trace\left[
E_2
D^{(1)}_{
\left(
(1/R_{\rm CE}^{(s)})
(g_{\rm CE}^{(1-s)})'
\right)(A)
}[E_2]
\right]
\nonumber\\
&\quad
-\frac{1}{4}
\trace\left[
D^{(1)}_{(1/R_{\rm CE}^{(s)})(A)}[E_2]
D^{(1)}_{g_{\rm CE}^{(1-s)}(A)}[E_2]
\right].
\label{eq:QCE t9}
\end{align}
 Applying
Lemma~\ref{lemma:Df-1} to the second
and fourth terms of \eqref{eq:QCE t9}, we have
\begin{align}
   & -\frac{1}{4}
\trace\left[
D^{(1)}_{(1/R_{\rm CE}^{(s)})(A)}[E_1]
D^{(1)}_{g_{\rm CE}^{(s)}(A)}[E_1]
\right]  \nonumber\\
& = \frac{1}{4}
\trace\left[
(1/R_{\rm CE}^{(s)})(A)
\left(
D^{(1)}_{g_{\rm CE}^{(s)}(A)}[E_1]
+
D^{(1)}_{g_{\rm CE}^{(1-s)}(A)}[E_1]
\right)
\right.
\nonumber\\
&\left.
\qquad\qquad\times
(1/R_{\rm CE}^{(s)})(A)
D^{(1)}_{g_{\rm CE}^{(s)}(A)}[E_1]
\right], \label{eq:QCE t2-1}
\end{align}
and
\begin{align}
&-\frac{1}{4}
\trace\left[
D^{(1)}_{(1/R_{\rm CE}^{(s)})(A)}[E_2]
D^{(1)}_{g_{\rm CE}^{(1-s)}(A)}[E_2]
\right] \nonumber \\
& = \frac{1}{4}
\trace\left[
(1/R_{\rm CE}^{(s)})(A)
\left(
D^{(1)}_{g_{\rm CE}^{(s)}(A)}[E_2]
+
D^{(1)}_{g_{\rm CE}^{(1-s)}(A)}[E_2]
\right)
\right.
\nonumber\\
&\left.
\qquad\qquad\times
(1/R_{\rm CE}^{(s)})(A)
D^{(1)}_{g_{\rm CE}^{(1-s)}(A)}[E_2]
\right]. \label{eq:QCE t2-2}
\end{align}
Furthermore, substituting
\eqref{eq:QCE RCB} and \eqref{eq:QCE delta 1} into the fourth term of
\eqref{eq:QCE term 2}, we obtain
\begin{align}
    &-\frac{1}{4}
\trace\left[
\left(
(1/R_{\rm CE}^{(s)})(A)
\Delta_{\rm CE}^{(1)}
\right)^2
\right] \nonumber\\
& \quad =-\frac{1}{4}
\trace\left[
\left(
(1/R_{\rm CE}^{(s)})(A)
\left(
D^{(1)}_{g_{\rm CE}^{(s)}(A)}[E_1] \right.\right.\right.\nonumber\\
&\left.\left.\left.\qquad\qquad\qquad\qquad+
D^{(1)}_{g_{\rm CE}^{(1-s)}(A)}[E_2]
\right)
\right)^2
\right] \label{eq:QCE t2-3}
\end{align}
Using \eqref{eq:QCE equal 1} and \eqref{eq:QCE equal 2}, the first
and third terms of \eqref{eq:QCE t9} cancel \eqref{eq:QCE t7} and
\eqref{eq:QCE t8}, respectively. Thus, combining \eqref{eq:QCE t2-1}, \eqref{eq:QCE t2-2} and \eqref{eq:QCE t2-3}, we simplify the second-order covariance contribution by
\begin{align}
\!\!\!\xi_{\rm CE,CM,s}^{(2)} 
&=
\mathcal{B}_s(E_1,E_1)
+
\mathcal{B}_s(E_2,E_2)
-
2\mathcal{B}_s(E_1,E_2),
\label{eq:QCE t10}
\end{align}
where
\begin{align}
\mathcal{B}_s(X,Y)
&=
\frac{1}{4}
\trace\left[
(1/R_{\rm CE}^{(s)})(A)
D^{(1)}_{g_{\rm CE}^{(s)}(A)}[X]
\right.
\nonumber\\
&\qquad\left.
\times
(1/R_{\rm CE}^{(s)})(A)
D^{(1)}_{g_{\rm CE}^{(1-s)}(A)}[Y]
\right].
\label{eq:QCE symmetry 0}
\end{align}
To further simplify \eqref{eq:QCE t10}, we show that
$\mathcal{B}_s(X,Y)$ is symmetric in $X$ and $Y$, i.e.,
\begin{align}
\mathcal{B}_s(X,Y)
=
\mathcal{B}_s(Y,X).
\label{eq:QCE symmetry}
\end{align}
Let the eigen-decomposition of $(1/R_{\rm CE}^{(s)})(A)$ be
\begin{align}
(1/R_{\rm CE}^{(s)})(A)
=
\sum_k q_k P_k.
\label{eq:QCE symmetry 1}
\end{align}
Using the divided-difference representation in \eqref{eq:D1f P},
we also write
\begin{align}
D^{(1)}_{g_{\rm CE}^{(s)}(A)}[X]
&=
\sum_{k,l}
(g_{\rm CE}^{(s)})^{[1]}(\lambda_k,\lambda_l)
P_kXP_l
\nonumber\\
&=
\sum_{k,l}
g_s^{(kl)}X_{kl},
\label{eq:QCE symmetry 2}
\end{align}
where $ g_s^{(kl)}=
(g_{\rm CE}^{(s)})^{[1]}(\lambda_k,\lambda_l) $, and $X_{kl}=
P_kXP_l$.
Substituting \eqref{eq:QCE symmetry 1} and
\eqref{eq:QCE symmetry 2} into \eqref{eq:QCE symmetry 0} gives
\begin{align}
\mathcal{B}_s(X,Y)
&=
\frac{1}{4}
\trace\left[
\sum_{k,l}
q_lq_k
g_s^{(kl)}
g_{1-s}^{(kl)}
X_{lk}Y_{kl}
\right]
\label{eq:QCE symmetry 3}
\\
&=
\frac{1}{4}
\trace\left[
\sum_{k,l}
q_kq_l
g_s^{(lk)}
g_{1-s}^{(lk)}
X_{kl}Y_{lk}
\right]
\label{eq:QCE symmetry 4}
\\
&=
\mathcal{B}_s(Y,X).
\label{eq:QCE symmetry 5}
\end{align}
Here, \eqref{eq:QCE symmetry 3} follows from the orthogonality of the
spectral projectors  $ P_kP_l = \delta_{kl}P_k $, while \eqref{eq:QCE symmetry 4} follows by
relabeling $k$ and $l$ and using the symmetry property of the first divided
difference,
$g_s^{(kl)}=g_s^{(lk)}$.
Using \eqref{eq:QCE symmetry}, \eqref{eq:QCE t10} reduces to
\begin{align}
\xi_{\rm CE,CM,s}^{(2)}
=
\mathcal{B}_s(E_1-E_2,E_1-E_2).
\label{eq:QCE perturb t2}
\end{align}

\subsection{Perturbative formula}

Combining the zeroth- and first-order cancellations, the second-order
covariance-matrix contribution in \eqref{eq:QCE perturb t2}, and the
mean-vector contribution in \eqref{eq:QCE perturb mean-vector},
the quantum Chernoff quantity up to the second order is
\begin{align}
\xi_{\rm{pert},s}^{(2)}(\hat{\rho}_1,\hat{\rho}_2)
=
\xi_{\rm CE,CM}^{(2)}
+
\xi_{\rm CE,d,s},
\end{align}
where
\begin{align}
\xi_{\rm CE,CM,s}^{(2)}
&=
\mathcal{B}_s(E_1-E_2,E_1-E_2), \label{eq:QCE xi 2 CB CM}
\\
\xi_{\rm CE,d,s}
&=
\delta_d^T
\left[
\Omega
\left(
g_{\rm CE}^{(s)}(A)
+
g_{\rm CE}^{(1-s)}(A)
\right)
\right]^{-1}
\delta_d,
\end{align}
with $\mathcal{B}_s(X,Y)$ defined in
\eqref{eq:QCE symmetry 0}.

\subsection{Optimal Chernoff parameter}
\label{app:QCE concavity}

In this subsection, we show that the leading perturbative quantum
Chernoff exponent is maximized at $s=\frac{1}{2}$.
We first establish its concavity with respect to $s$.

Let $\epsilon\rightarrow0$ denote the asymptotic parameter governing
the perturbations, and let $\xi_s(\epsilon)$ denote the exact quantum
Chernoff quantity. We
write its leading asymptotic behavior as
\begin{align}
    \xi_s(\epsilon)
    =
    a(\epsilon)\,
    \xi_{{\rm lead},s}
    +
    o\!\left(a(\epsilon)\right).
\end{align}
Here, $a(\epsilon)>0$ is a common asymptotic scale independent of $s$. Since the leading covariance-matrix and mean-vector contributions are quadratic in $E_1-E_2$ and $\delta d$, respectively, we choose $a(\epsilon)$ as the dominant of these two scales, allowing either contribution or both to determine the leading term. 
Therefore,
\begin{align}
    \xi_{{\rm lead},s}
    =
    \lim_{\epsilon\rightarrow0}
    \frac{\xi_s(\epsilon)}{a(\epsilon)}. \label{eq:QCE epsilon limit}
\end{align}
Since the exact quantum Chernoff quantity is concave in $s$
\cite[Lemma~20(i)]{Beigi2026}, for any
$s_1,s_2\in[0,1]$ and $\lambda\in[0,1]$,
\begin{align}
\xi_{\lambda s_1+(1-\lambda)s_2}(\epsilon)
\geq
\lambda\xi_{s_1}(\epsilon)
+
(1-\lambda)\xi_{s_2}(\epsilon).
\label{eq:QCE exact concavity}
\end{align}
Since $a(\epsilon)>0$, dividing both sides of
\eqref{eq:QCE exact concavity} by $a(\epsilon)$ gives
\begin{align}
\frac{
\xi_{\lambda s_1+(1-\lambda)s_2}(\epsilon)
}{a(\epsilon)}
&\geq
\lambda
\frac{\xi_{s_1}(\epsilon)}{a(\epsilon)}
+
(1-\lambda)
\frac{\xi_{s_2}(\epsilon)}{a(\epsilon)}.
\label{eq: QCE concavity divide e}
\end{align}
Taking the limit $\epsilon\rightarrow0$ and using
\eqref{eq:QCE epsilon limit}, we obtain
\begin{align}
\xi_{{\rm lead},
\lambda s_1+(1-\lambda)s_2}
&\geq
\lambda
\xi_{{\rm lead},s_1}
+
(1-\lambda)
\xi_{{\rm lead},s_2}.
\label{eq:QCE perturb concavity}
\end{align}
Here, the leading-order term $\xi_{\rm lead,s}$ is $\xi_{\rm pert,s}^{(2)}$.
Therefore, the leading perturbative quantum Chernoff quantity
$\xi_{\rm pert,s}^{(2)}$ is concave in $s$.

We next show that $s = \frac{1}{2}$ is a stationary point of $\xi_{{\rm pert},s}^{(2)}$. From the definition
\begin{align}
R_{\rm CE}^{(s)}(A)
=
g_{\rm CE}^{(s)}(A)
+
g_{\rm CE}^{(1-s)}(A),
\end{align}
we have
\begin{align}
R_{\rm CE}^{(s)}(A)
=
R_{\rm CE}^{(1-s)}(A).
\end{align}
Moreover, exchanging $s$ and $1-s$ in
$\mathcal{B}_s(X,X)$ only interchanges the two Fr\'echet-derivative
factors. By the cyclic property of the trace,
\begin{align}
\mathcal{B}_s(X,X)
=
\mathcal{B}_{1-s}(X,X).
\end{align}
The mean-vector contribution is also invariant under interchanging $s$ with $1-s$. Therefore,
\begin{align}
\xi_{{\rm pert},s}^{(2)}
=
\xi_{{\rm pert},1-s}^{(2)}.
\label{eq:QCE perturb symmetry s}
\end{align}
Differentiating \eqref{eq:QCE perturb symmetry s} with respect to $s$
gives
\begin{align}
\partial_s\xi_{{\rm pert},s}^{(2)}
=
-
\partial_s\xi_{{\rm pert},1-s}^{(2)}.
\end{align}
Setting $s=\frac{1}{2}$, we obtain
\begin{align}
\partial_s\xi_{{\rm pert},\frac{1}{2}}^{(2)}
=
-
\partial_s\xi_{{\rm pert},\frac{1}{2}}^{(2)} = 0
\end{align}
Since $\xi_{{\rm pert},s}^{(2)}$ is concave in $s$, this stationary
point is a global maximum. Hence,
\begin{align}
\max_{s\in[0,1]}
\xi_{{\rm pert},s}^{(2)}
=
\xi_{{\rm pert},1/2}^{(2)}.
\end{align}

\section{Quantum fidelity} \label{appendix:gaussian_qi frechet QF}
\subsection{Gaussian covariance-matrix representation} \label{appendix:sub gaussian_qi AE}

For two $n$-mode Gaussian states
$\hat{\rho}_1(\sigma_1,d_1)$ and
$\hat{\rho}_2(\sigma_2,d_2)$, the QF can be written as
\cite{LeonardoQF}
\begin{align}
F(\hat{\rho}_1,\hat{\rho}_2)
=
F_0
\exp\left[
-\frac{1}{4}
\delta_d^{T}
(\sigma_1+\sigma_2)^{-1}
\delta_d
\right],
\label{eq: app QF}
\end{align}
where $\delta_d=d_1-d_2$,
and $F_0$ depends only on the covariance matrices. In particular,
\begin{align}
F_0^4
=
\frac{
\det
\left[
\left(
\sqrt{4\mathbb{I}+(V_{\rm aux}\Omega)^{-2}}
+2\mathbb{I}
\right)
V_{\rm aux}
\right]
}{
\det(\sigma_1+\sigma_2)
},
\label{eq: app QF F_0}
\end{align}
where
\begin{align}
V_{\rm aux}
=
\Omega^T
(\sigma_1+\sigma_2)^{-1}
\left(
\frac{\Omega}{4}
+
\sigma_2\Omega\sigma_1
\right).
\label{eq: app QF V aux}
\end{align}
We first express $\Omega V_{\rm aux}$ in terms of
$\Omega\sigma_1 = A+E_1$ and $\Omega\sigma_2 = A+E_2$,
\begin{align}
\Omega V_{\rm aux}
&=
(\sigma_1+\sigma_2)^{-1}
\Omega^{-1}\Omega\left(
\frac{\Omega}{4}
+
\sigma_2\Omega\sigma_1
\right)
\nonumber\\
&=
(\Omega\sigma_1+\Omega\sigma_2)^{-1}
\left(
\Omega\sigma_2\Omega\sigma_1
-
\frac{\mathbb{I}}{4}
\right)
\label{eq: QF Omega Vaux} \\
&=
(2A+E_1+E_2)^{-1} \nonumber\\
&\qquad\times\Big[
A^2-\frac{\mathbb{I}}{4}
+
E_2A
+
AE_1
+
E_2E_1
\Big].
\label{eq: QF Vaux AE}
\end{align}
We next rewrite $F_0$ in terms of $A$, $E_1$, and $E_2$. Multiplying both numerator and denominator by $\det(\Omega)$, using $\det(A)\det(B) = \det(AB)$, and using the similarity invariance, \eqref{eq: app QF F_0} yields
\begin{align}
F_0^4
&=
\frac{
\det
\left[
\Omega
\left(
\sqrt{
4\mathbb{I}
+
\left[
\Omega^{-1}
(\Omega V_{\rm aux})
\Omega
\right]^{-2}
}
+
2\mathbb{I}
\right)
V_{\rm aux}
\right]
}{
\det(\Omega\sigma_1+\Omega\sigma_2)
}
\nonumber\\
&=
\frac{
\det
\left[
\left(
\sqrt{
4\mathbb{I}
+
(\Omega V_{\rm aux})^{-2}
}
+
2\mathbb{I}
\right)
\Omega V_{\rm aux}
\right]
}{
\det(\Omega\sigma_1+\Omega\sigma_2)
} \\
& = \frac{
\det\left[
f_{\rm F}(\Omega V_{\rm aux})
\right]
}{
\det(2A+E_1+E_2)
}, \label{eq: QF F_0^4}
\end{align}
where we define
\begin{align}
f_{\rm F}(x)
=
\left(
\sqrt{4+x^{-2}}+2
\right)x.
\label{eq: QF fF}
\end{align}
Since the fidelity is unity at zeroth order and the first-order contribution vanishes, the quantum infidelity is second-order. Therefore, $-\log F=1-F$ up to second order in the perturbations.
Taking the logarithm of \eqref{eq: QF F_0^4}, the covariance-matrix
contribution to the logarithmic fidelity becomes
\begin{align}
\log F_0
&=
\frac{1}{4}
\log\det
\left[
f_{\rm F}(\Omega V_{\rm aux})
\right] \nonumber\\
&\qquad\qquad
-
\frac{1}{4}
\log\det
\left(
2A+E_1+E_2
\right).
\label{eq: QF log F_0 0}
\end{align}
Finally, the mean-vector contribution to the logarithmic fidelity can
be written as
\begin{align}
\!\!\!-\frac{1}{4}
\delta_d^T
(\sigma_1+\sigma_2)^{-1}
\delta_d
&=
-\frac{1}{4}
\delta_d^T
(\Omega\sigma_1+\Omega\sigma_2)^{-1}
\Omega
\delta_d
\nonumber\\
&=
-\frac{1}{4}
\delta_d^T
(2A+E_1+E_2)^{-1}
\Omega
\delta_d.
\label{eq:QF Displace}
\end{align}
Since $h(A)$ in \eqref{eq:QF h0} is singular when $\sigma_0$ has a symplectic
eigenvalue equal to $1/2$, we first restrict the analysis to $\sigma_0$ whose symplectic eigenvalues are all greater than
$1/2$.

\subsection{Fr\'echet expansion}
We first expand $\Omega V_{\rm aux}$ in \eqref{eq: QF Vaux AE}
around $A$ up to the second order in terms of $E_1$ and $E_2$. First, using the matrix inverse perturbation expansion
\begin{align}
(2A+E_1+E_2)^{-1}
\simeq{}&
\frac{1}{2}A^{-1}
-\frac{1}{4}
A^{-1}(E_1+E_2)A^{-1}
\nonumber\\
&\ +
\frac{1}{8}
A^{-1}\left[(E_1+E_2)
A^{-1}\right]^2,
\end{align}
we show that \eqref{eq: QF Vaux AE} yields
\begin{align}
\Omega V_{\rm aux}
&\simeq
r_{\rm F}(A)
+
\Delta_{\rm F}^{(1)}
+
\Delta_{\rm F}^{(2)},
\label{eq: QF V aux pert}
\end{align}
where
\begin{align}
r_{\rm F}(A)
&=
\frac{1}{2}A
-
\frac{1}{8}A^{-1},
\label{eq:QF rF}
\\
\Delta_{\rm F}^{(1)}
&=
\frac{1}{2}
A^{-1}(E_2A+AE_1)
\nonumber\\
&\quad
-\frac{1}{4}
A^{-1}(E_1+E_2)A^{-1}
\left(
A^2-\frac{\mathbb{I}}{4}
\right),
\label{eq:QF Delta1}
\\
\Delta_{\rm F}^{(2)}
&=
\frac{1}{2}A^{-1}E_2E_1
\nonumber\\
&\quad
-\frac{1}{4}
A^{-1}(E_1+E_2)A^{-1}
(E_2A+AE_1)
\nonumber\\
&\quad
+\frac{1}{8}
A^{-1}\left[(E_1+E_2)
A^{-1}\right]^2
\left(
A^2-\frac{\mathbb{I}}{4}
\right).
\label{eq:QF Delta2}
\end{align}
Substituting \eqref{eq: QF V aux pert} into
\eqref{eq: QF log F_0 0}, and using the Fr\'echet expansion of
$\log\det f_{\rm F}(\cdot)$ in Lemma~\ref{lemma:D logdet f} around $r_{\rm F}(A)$, the first term of \eqref{eq: QF log F_0 0} yields
\begin{align}
&\frac{1}{4}\log\det
f_{\rm F}\!\left(
r_{\rm F}(A)
+
\Delta_{\rm F}^{(1)}
+
\Delta_{\rm F}^{(2)}
\right) \nonumber\\
&\qquad\simeq{}
\frac{1}{4}\log\det
f_{\rm F}\!\left(r_{\rm F}(A)\right)+
\frac{1}{4}\trace\left[
(\log f_{\rm F})'
\!\left(r_{\rm F}(A)\right)
\Delta_{\rm F}^{(1)}
\right]
\nonumber\\
&\qquad\qquad+
\frac{1}{4}\trace\left[
(\log f_{\rm F})'
\!\left(r_{\rm F}(A)\right)
\Delta_{\rm F}^{(2)}
\right]
\nonumber\\
&\qquad\qquad+
\frac{1}{8}
\trace\left[
\Delta_{\rm F}^{(1)}
D^{(1)}_{(\log f_{\rm F})'
(r_{\rm F}(A))}
[\Delta_{\rm F}^{(1)}]
\right].
\label{eq:QF logdet f expansion}
\end{align}
Similarly, using Lemma~\ref{lemma:D logdet f}, the second term of \eqref{eq: QF log F_0 0} yields
\begin{align}
&\frac{1}{4}\log\det(2A+E_1+E_2) \nonumber \\
&\quad\simeq{}
\frac{1}{4} \log\det(2A)
+
\frac{1}{8}
\trace\left[
A^{-1}(E_1+E_2)
\right]
\nonumber\\
&\qquad\qquad-
\frac{1}{32}
\trace\left[
A^{-1}(E_1+E_2)
A^{-1}(E_1+E_2)
\right].
\label{eq:QF logdet denominator expansion}
\end{align}
Substituting \eqref{eq:QF logdet f expansion} and \eqref{eq:QF logdet denominator expansion} into
\eqref{eq: QF log F_0 0}, we obtain
\begin{align}
\log F_0
&\simeq{}
\frac{1}{4}
\bigg[
\log\det
f_{\rm F}\!\left(r_{\rm F}(A)\right)
-
\log\det(2A)
\nonumber\\
&\
+
\trace\left(
(\log f_{\rm F})'
\!\left(r_{\rm F}(A)\right)
\Delta_{\rm F}^{(1)}
-
\frac{1}{2}
A^{-1}(E_1+E_2)
\right)
\nonumber\\
&\ 
+
\trace\left(
(\log f_{\rm F})'
\!\left(r_{\rm F}(A)\right)
\Delta_{\rm F}^{(2)}
\right)
\nonumber\\
&\ 
+
\frac{1}{2}
\trace\left[
\Delta_{\rm F}^{(1)}
D^{(1)}_{(\log f_{\rm F})'
(r_{\rm F}(A))}
[\Delta_{\rm F}^{(1)}]
\right]
\nonumber\\
&\ 
+
\frac{1}{8}
\trace\left[
A^{-1}(E_1+E_2)
A^{-1}(E_1+E_2)
\right]
\bigg].
\label{eq: QF log F0 expansion}
\end{align}
Accordingly, we write
\begin{align}
\log F_0
\simeq
F_{\rm CM}^{(0)}
+
F_{\rm CM}^{(1)}
+
F_{\rm CM}^{(2)},
\end{align}
where
\begin{align}
F_{\rm CM}^{(0)}
&=
\frac{1}{4}
\left[
\log\det
f_{\rm F}\!\left(r_{\rm F}(A)\right)
-
\log\det(2A)
\right] = 0,
\label{eq:QF perturb0 0}
\\
F_{\rm CM}^{(1)}
&=
\frac{1}{4}
\trace\left[
(\log f_{\rm F})'
\!\left(r_{\rm F}(A)\right)
\Delta_{\rm F}^{(1)} \right.\nonumber\\
&\left.\qquad\qquad\qquad\qquad
-
\frac{1}{2}
A^{-1}(E_1+E_2)
\right],
\label{eq:QF perturb1 0}
\\
F_{\rm CM}^{(2)}
&=
\frac{1}{4}
\trace\left[
(\log f_{\rm F})'
\!\left(r_{\rm F}(A)\right)
\Delta_{\rm F}^{(2)}
\right]
\nonumber\\
&\quad
+
\frac{1}{8}
\trace\left[
\Delta_{\rm F}^{(1)}
D^{(1)}_{(\log f_{\rm F})'
(r_{\rm F}(A))}
[\Delta_{\rm F}^{(1)}]
\right]
\nonumber\\
&\quad
+
\frac{1}{32}
\trace\left[
A^{-1}(E_1+E_2)
A^{-1}(E_1+E_2)
\right].
\label{eq:QF perturb2 0}
\end{align}
The zeroth-order covariance-matrix contribution $F_{\rm CM}^{(0)}$
vanishes, since the fidelity between two identical Gaussian states is unity and its logarithm therefore vanishes.
The mean-vector contribution in \eqref{eq:QF Displace} is
already small. Therefore, we retain only its leading-order term and obtain
\begin{align}
F_{\rm d}
=
-\frac{1}{4}
\delta_d^T
(2A)^{-1}
\Omega
\delta_d.
\label{eq:QF mean-vector pert}
\end{align}
Hence, the logarithmic fidelity up to the second order is
\begin{align}
\log F_{\rm pert}
=
F_{\rm CM}^{(1)}
+
F_{\rm CM}^{(2)}
+
F_{\rm d}.
\label{eq:QF total perturb}
\end{align}

\subsection{First-order cancellation}

We now show that the first-order covariance-matrix contribution  $ F_{\rm CM}^{(1)}$ in
\eqref{eq:QF perturb1 0} vanishes. From \eqref{eq: QF fF}, 
we have
\begin{align}
\left(\log f_{\rm F}\right)'(x)
=
\frac{2}{\sqrt{1+4x^2}}.
\end{align}
Define
\begin{align}
h(x)
=
\left(\log f_{\rm F}\right)'
\left(r_{\rm F}(x)\right) \label{eq:QF h0}.
\end{align}
Substituting \eqref{eq:QF rF} into \eqref{eq:QF h0}, $h(x)$ yields
\begin{align}
h(x)
=
\frac{8x}{1+4x^2}.
\label{eq:QF h}
\end{align}
Substituting \eqref{eq:QF h} and \eqref{eq:QF Delta1} into the first term in
\eqref{eq:QF perturb1 0}, we obtain
\begin{align}
&\frac{1}{4}
\trace\left[
h(A)\Delta_{\rm F}^{(1)}
\right]
\nonumber\\
&=
\frac{1}{4}
\trace\left[
\frac{8A}{1+4A^2}
\left(
\frac{1}{2}
A^{-1}(E_2A+AE_1)
\right.\right.
\nonumber\\
&\left.\left.
\qquad\qquad
-\frac{1}{4}
A^{-1}(E_1+E_2)A^{-1}
\left(
A^2-\frac{\mathbb{I}}{4}
\right)
\right)
\right] \\
&=
\trace\left[
\frac{2A(E_1+E_2)}{1+4A^2}
\left(
\frac{\mathbb{I}}{2}
-
\frac{1}{4}
A^{-2}
\left(
A^2-\frac{\mathbb{I}}{4}
\right)
\right)
\right] \label{eq:QF t11 10}
\\
&=
\trace\left[
\frac{2A(E_1+E_2)}{1+4A^2}
\frac{4A^2+\mathbb{I}}{16A^2}
\right]
\\
&=
\frac{1}{8}
\trace\left[
A^{-1}(E_1+E_2)
\right],
\label{eq:QF first term cancellation}
\end{align}
 where \eqref{eq:QF t11 10} follows from using the cyclic property of the trace multiple times.
Substituting \eqref{eq:QF first term cancellation} into
\eqref{eq:QF perturb1 0}, we obtain
\begin{align}
F_{\rm CM}^{(1)}=
0. \label{eq:QF t1 = 0}
\end{align}
Therefore, the first-order covariance-matrix contribution to the
logarithmic fidelity vanishes.

\subsection{Second-order simplification}

We now simplify the second-order covariance-matrix contribution
$F_{\rm CM}^{(2)}$ in \eqref{eq:QF perturb2 0}. Let the eigen-decomposition of $A$ be
\begin{align}
A
=
U\Gamma U^{-1}, \label{eq:QF eigen A}
\end{align}
where
\begin{align}
\Gamma
=
\operatorname{diag}\{a_k\}_{k=1}^{2n},
\end{align}
and define the perturbations in the eigenbasis of $A$ as
\begin{align}
\widetilde{E}_k
=
U^{-1}E_kU,
\qquad
k=1,2. \label{eq:QF eigen E}
\end{align}
Consider the first term in
\eqref{eq:QF perturb2 0}. Using  \eqref{eq:QF Delta2} and \eqref{eq:QF h}, we obtain
\begin{align}
&\frac{1}{4}
\trace\left[
h(A)\Delta_{\rm F}^{(2)}
\right] \nonumber \\
&=
\frac{1}{4}
\trace\left[
\frac{8A}{1+4A^2}
\Delta_{\rm F}^{(2)}
\right]\\
&=
\frac{1}{4}
\trace\left[
\frac{8\Gamma}{1+4\Gamma^2}
\left(
U^{-1}\Delta_{\rm F}^{(2)}U
\right)
\right]\\
&=
\trace\left[
\frac{2}{1+4\Gamma^2}
\left(
\frac{1}{2}
\widetilde{E}_2\widetilde{E}_1
-\frac{1}{4}
(\widetilde{E}_1+\widetilde{E}_2)
\Gamma^{-1}
(\widetilde{E}_2\Gamma+\Gamma\widetilde{E}_1)
\right.\right.
\nonumber\\
&\left.\left.
\qquad
+\frac{1}{8}
(\widetilde{E}_1+\widetilde{E}_2)
\Gamma^{-1}
(\widetilde{E}_1+\widetilde{E}_2)
\Gamma^{-1}
\left(
\Gamma^2-\frac{\mathbb{I}}{4}
\right)
\right)
\right]\label{eq:QF perturb2 t-1}\\
&=
\sum_{k,l}
\Bigg[
\frac{
(\widetilde{E}_2)_{kl}
(\widetilde{E}_1)_{lk}
}{
1+4a_k^2
}
\nonumber\\
&\
-
\frac{1}{2}
\frac{
(\widetilde{E}_1+\widetilde{E}_2)_{kl}
a_l^{-1}
}{
1+4a_k^2
}
\left[
(\widetilde{E}_2)_{lk}a_k
+
a_l(\widetilde{E}_1)_{lk}
\right]
\nonumber\\
&\
+
\frac{1}{4}
\frac{
(\widetilde{E}_1+\widetilde{E}_2)_{kl}
a_l^{-1}
}{
1+4a_k^2
}
(\widetilde{E}_1+\widetilde{E}_2)_{lk}
a_k^{-1}
\left(
a_k^2-\frac{1}{4}
\right)
\Bigg].
\label{eq:QF perturb2 t0}
\end{align}
Here, \eqref{eq:QF perturb2 t-1} is obtained by transforming $\Delta_{\rm F}^{(2)}$ into the eigenbasis of
$A$, and \eqref{eq:QF perturb2 t0} follows from using the diagonal structure of $\Gamma$ to write the trace
entrywise as a double summation.

We next consider the second term in
\eqref{eq:QF perturb2 0}. 
\begin{align}
&\frac{1}{8}
\trace\left[
\Delta_{\rm F}^{(1)}
D^{(1)}_{(\log f_{\rm F})'
(r_{\rm F}(A))}
[\Delta_{\rm F}^{(1)}]
\right]
\nonumber\\
&=
\frac{1}{8}
\trace\left[
\left(
U^{-1}\Delta_{\rm F}^{(1)}U
\right)
U^{-1}
D^{(1)}_{(\log f_{\rm F})'
(r_{\rm F}(A))}
[\Delta_{\rm F}^{(1)}]
U
\right]
\\
&=
\frac{1}{8}
\sum_{k,l}
\left(
U^{-1}\Delta_{\rm F}^{(1)}U
\right)_{kl}
\nonumber\\
&\qquad\times
\left(
U^{-1}
D^{(1)}_{(\log f_{\rm F})'
(r_{\rm F}(A))}
[\Delta_{\rm F}^{(1)}]
U
\right)_{lk}
\\
&=
\frac{1}{8}
\sum_{k,l}
\left(
U^{-1}\Delta_{\rm F}^{(1)}U
\right)_{kl}
\nonumber\\
&\qquad\times
\left[
\left(
\left((\log f_{\rm F})'\right)^{[1]}
[r_{\rm F}(\Gamma)]
\right)
\circ
\left(
U^{-1}\Delta_{\rm F}^{(1)}U
\right)
\right]_{lk},
\label{eq:QF perturb2 t1}
\end{align}
where \eqref{eq:QF perturb2 t1} follows from the divided difference form of the first Fr\'echet derivative in \eqref{eq:First Frechet derivative Hadmard product}. Furthermore, from \eqref{eq:QF Delta1}, the $kl$-th entry of $\Delta_{\rm F}^{(1)}$ in the eigenbasis of $A$ is given by
\begin{align}
\left(
U^{-1}\Delta_{\rm F}^{(1)}U
\right)_{kl}
&=
\left[
\frac{1}{2}
\Gamma^{-1}
(\widetilde{E}_2\Gamma+\Gamma\widetilde{E}_1)
\right]_{kl}
\nonumber\\
&\quad
-
\left[
\frac{1}{4}
\Gamma^{-1}
(\widetilde{E}_1+\widetilde{E}_2)
\Gamma^{-1}
\left(
\Gamma^2-\frac{\mathbb{I}}{4}
\right)
\right]_{kl}
\nonumber\\
&=
\frac{
a_k^{-1}(\widetilde{E}_2)_{kl}a_l
+
(\widetilde{E}_1)_{kl}
}{2} \nonumber\\
&
\quad-
\frac{
(\widetilde{E}_1+\widetilde{E}_2)_{kl}
\left(
a_l^2-\frac{1}{4}
\right)
}{
4a_ka_l
}.
\label{eq:QF Delta1 eigenbasis}
\end{align}
The entries of the divided-difference matrix in \eqref{eq:QF perturb2 t1} are
\begin{align}
&\left(
\left((\log f_{\rm F})'\right)^{[1]}
[r_{\rm F}(\Gamma)]
\right)_{kl}
\nonumber\\
&=
\begin{cases}
\displaystyle
\frac{
h(a_k)-h(a_l)
}{
r_{\rm F}(a_k)-r_{\rm F}(a_l)
},
& k\neq l,
\\[10pt]
\displaystyle
(\log f_{\rm F})''
\left(
r_{\rm F}(a_k)
\right),
& k=l,
\end{cases}
\nonumber\\
&=
\begin{cases}
\displaystyle
\frac{
64a_ka_l(1-4a_ka_l)
}{
(1+4a_k^2)
(1+4a_l^2)
(1+4a_ka_l)
},
& k\neq l,
\\[12pt]
\displaystyle
\frac{
64a_k^2(1-4a_k^2)
}{
(1+4a_k^2)^3
},
& k=l,
\end{cases}
\nonumber\\
&=
\frac{
64a_ka_l(1-4a_ka_l)
}{
(1+4a_k^2)
(1+4a_l^2)
(1+4a_ka_l)
},
\label{eq:QF divided difference h}
\end{align}
where the last expression applies to all $k$ and $l$.
Consider the third term in \eqref{eq:QF perturb2 0}:
\begin{align}
&\frac{1}{32}
\trace\left[
A^{-1}(E_1+E_2)
A^{-1}(E_1+E_2)
\right]
\nonumber\\
&=
\frac{1}{32}
\trace\left[
\Gamma^{-1}
(\widetilde{E}_1+\widetilde{E}_2)
\Gamma^{-1}
(\widetilde{E}_1+\widetilde{E}_2)
\right]
\nonumber\\
&=
\frac{1}{32}
\sum_{k,l}
\frac{
(\widetilde{E}_1+\widetilde{E}_2)_{kl}
(\widetilde{E}_1+\widetilde{E}_2)_{lk}
}{
a_ka_l
}.
\label{eq:QF perturb2 t2}
\end{align}
Here, we again transform the perturbations $E_1$ and $E_2$ into the eigenbasis of $A$ using \eqref{eq:QF eigen A}--\eqref{eq:QF eigen E}, then write the trace as entrywise summation.
Substituting \eqref{eq:QF perturb2 t0},
\eqref{eq:QF perturb2 t1}, and
\eqref{eq:QF perturb2 t2} into
\eqref{eq:QF perturb2 0}, and using
\eqref{eq:QF Delta1 eigenbasis} and
\eqref{eq:QF divided difference h}, we obtain
\begin{align}
F_{\rm CM}^{(2)}
=
\sum_{k,l}
\frac{
C_1+C_2+C_3
}{
4a_k(1+4a_k^2)
a_l(1+4a_ka_l)
(1+4a_l^2)
},
\label{eq:QF perturb2 1}
\end{align}
where
\begin{align}
C_1
&=
\left[
a_k^2
-
a_k(1-4a_k^2)a_l
-
(1-16a_k^4)a_l^2
\right.
\nonumber\\
&\left.
\qquad
-
4a_k(3+4a_k^2)a_l^3
-
16a_k^2a_l^4
\right]
(\widetilde{E}_1)_{kl}
(\widetilde{E}_1)_{lk},
\\
C_2
&=
a_k(1+4a_l^2)
\left[
a_k+2a_l
+
4a_ka_l(2a_k+a_l)
\right]
\nonumber\\
&\qquad\times
(\widetilde{E}_2)_{kl}
(\widetilde{E}_1)_{lk},
\\
C_3
&=
-(1+4a_k^2)a_l
\left[
a_l(1+4a_k^2)
(\widetilde{E}_1)_{kl}
\right.
\nonumber\\
&\left.
\qquad\qquad
+
a_k(1+4a_l^2)
(\widetilde{E}_2)_{kl}
\right]
(\widetilde{E}_2)_{lk}.
\end{align}
The double sum in~\eqref{eq:QF perturb2 1} can be simplified by symmetrizing the summand. Let $s_{kl}$ denote the summand. Since the summation is invariant under interchanging $k$ and $l$,
\begin{align}
F_{\rm CM}^{(2)}
=
\sum_{k,l}s_{kl}
=
\frac{1}{2}
\sum_{k,l}
\left(
s_{kl}+s_{lk}
\right).
\label{eq:QF perturb2 2}
\end{align}
Simplifying the symmetrized summand gives
\begin{align}
F_{\rm CM}^{(2)}
&=
-\sum_{k,l}
\frac{
\left(
\widetilde{E}_1-\widetilde{E}_2
\right)_{kl}
\left(
\widetilde{E}_1-\widetilde{E}_2
\right)_{lk}
}{
4+16a_ka_l
} \label{eq:QF trace representation}\\
&=
-\frac{1}{4}
\trace\left[
\widetilde{X}
\left(
\widetilde{E}_1-\widetilde{E}_2
\right)
\right],
\label{eq:QF perturb2 3}
\end{align}
where $\widetilde{X}$ is defined entrywise as
\begin{align}
(\widetilde{X})_{kl}
=
\frac{
\left(
\widetilde{E}_1-\widetilde{E}_2
\right)_{kl}
}{
1+4a_ka_l
}.
\label{eq:QF Xtilde}
\end{align}
Equation \eqref{eq:QF Xtilde} is equivalently expressed entrywise as
\begin{align}
\left(
\widetilde{E}_1-\widetilde{E}_2
\right)_{kl}
&=
(\widetilde{X})_{kl}
+
4a_k(\widetilde{X})_{kl}a_l
\nonumber\\
&=
\left(
\widetilde{X}
+
4\Gamma\widetilde{X}\Gamma
\right)_{kl}.
\end{align}
This entry-wise relation is equivalent to the Lyapunov equation
\begin{align}
\widetilde{E}_1-\widetilde{E}_2
=
\widetilde{X}
+
4\Gamma\widetilde{X}\Gamma.
\label{eq:QF Xtilde equation}
\end{align}
Now we show that similar equality holds not only in the eigenbasis of $A$. Defining
\begin{align}
X
=
U\widetilde{X}U^{-1},
\end{align}
and transforming \eqref{eq:QF Xtilde equation} back to the original
basis gives
\begin{align}
E_1-E_2
&=
U
\left(
\widetilde{X}
+
4\Gamma\widetilde{X}\Gamma
\right)
U^{-1}
\nonumber\\
&=
X+4AXA.
\label{eq:QF X equation}
\end{align}
Therefore, the second-order covariance-matrix contribution depends only
on the difference $E_1-E_2$, 
\begin{align}
F_{\rm CM}^{(2)}
&=
-\frac{1}{4}
\trace\left[
U\widetilde{X}U^{-1}
U
\left(
\widetilde{E}_1-\widetilde{E}_2
\right)
U^{-1}
\right]
\nonumber\\
&=
-\frac{1}{4}
\trace\left[
X(E_1-E_2)
\right],
\label{eq:QF perturb2 final}
\end{align}
 where $X$ is determined by the linear
matrix equation \eqref{eq:QF X equation}.

\subsection{Perturbative formula}
From \eqref{eq:QF mean-vector pert} and \eqref{eq:QF perturb2 final}, the logarithmic fidelity up to the second order is
\begin{align}
\log F_{\rm pert}
=
-\frac{1}{4}
\trace\left[
X(E_1-E_2)
\right]
-
\frac{1}{4}
\delta_d^T
(2A)^{-1}
\Omega
\delta_d,
\label{eq:QF log final}
\end{align}
where $X$ satisfies
\begin{align}
E_1-E_2
=
X+4AXA.
\label{eq:QF X final}
\end{align}
Since the fidelity is unity at zeroth order and the first-order contribution vanishes, $\log F$ is second order in the perturbations. Therefore, up to second-order,
\begin{align}
F
=
\exp(\log F)
\simeq
1+\log F,
\end{align}
or equivalently,
\begin{align}
-\log F \simeq 1-F.
\end{align}
Here, $1-F$ denotes the quantum infidelity. 
Using \eqref{eq:QF log final}, the second-order perturbative fidelity is therefore
\begin{align}
F_{\rm pert}
=
1-\frac{1}{4}
\trace\left[
X(E_1-E_2)
\right]
-
\frac{1}{4}
\delta_d^T
(2A)^{-1}
\Omega
\delta_d.
\label{eq:QF final}
\end{align}

\section{Singularity analysis of the perturbative coefficients}
\label{app:singularity}

In this appendix, we derive the singular behavior of the
perturbative coefficients of the quantum information quantities that leads to the slower asymptotic vanishing
for a pure reference states summarized in
Sec.~\ref{subsec: pure scaling}. Since the behavior considered
here is determined by the covariance matrix, we set all mean vectors to
zero and focus on the covariance-matrix contributions.

 To characterize the
singularity of the perturbative coefficients when the reference state is a pure state, consider an $n$-mode auxiliary reference covariance
matrix
\begin{align}
    \sigma_0
    =
    \left(
    \frac{1}{2}+\tau
    \right)I_{2n},
    \qquad
    \tau\rightarrow0,
\end{align}
so that all symplectic eigenvalues approach $1/2$ with the same
auxiliary parameter $\tau$. The corresponding reference matrix is
\begin{align}
    A
    =
    \Omega\sigma_0
    =
    \left(
    \frac{1}{2}+\tau
    \right)\Omega,
\end{align}
whose eigenvalues are
\begin{align}
    a_{\pm}
    =
    \pm i
    \left(
    \frac{1}{2}+\tau
    \right),
\end{align}
each with multiplicity $n$.

We first determine how the perturbative coefficients diverge as
$\tau\rightarrow0$ for each quantum information quantity.
For the VNE, we write the covariance-matrix perturbation in the
$\Omega\sigma$ representation as
\begin{align}
    E
    =
    \Omega\Delta\sigma,
    \qquad
    \|E\|
    =
    \Theta(\epsilon).
\end{align}
For the QRE, quantum Chernoff quantity, and QF, we similarly define
\begin{align}
    E_k
    =
    \Omega\Delta\sigma_k,
    \qquad
    k=1,2,
    \qquad
    \|E_1-E_2\|
    =
    \Theta(\epsilon).
\end{align}
Throughout this appendix, the number of modes $n$ is fixed as
$\tau\rightarrow0$ and $\epsilon\rightarrow0$.

In the following derivations, we transform the perturbation matrices to
the eigenbasis of $A$. Since the eigenvectors of $A$ are independent of
$\tau$ and $\epsilon$, this similarity transformation does not change
the asymptotic order of the perturbations. Hence,
\begin{align}
    \|\widetilde E\|
    =
    \Theta(\epsilon),
    \qquad
    \| \widetilde{E}_1 - \widetilde{E}_2\|
    =
    \Theta(\epsilon).
\end{align}
For fixed matrix dimension, the singular behavior is therefore governed
by the asymptotic scaling of the relevant divided differences of the eigenvalues of $A$.

\subsection{Von Neumann entropy}
The perturbative expansion of the VNE is
\begin{align}
S_{\rm pert}^{(2)}
&=
\frac{1}{2}\trace[f_{\rm VN}(A)]
+
\frac{1}{2}\trace[
E f_{\rm VN}'(A)
]
\nonumber\\
&\quad+
\frac{1}{4}\trace\left[
E
D^{(1)}_{f_{\rm VN}'(A)}[E]
\right].
\end{align}
Recall that
\begin{align}
f_{\rm VN}(x)
=
g\left(\sqrt{-x^2}\right),
\end{align}
where $g(x)$ is given in \eqref{eq:g}.
We first consider the first-order contribution. Evaluating  $f'_{\rm VN}(x)$ at the eigenvalues $a_\pm$ gives
\begin{align}
f'_{\rm VN}(a_{\pm})
=
\pm i
\log\left(
\frac{\tau}{1+\tau}
\right)
=
\Theta\left(
|\log\tau|
\right).
\label{eq:VN first derivative singular}
\end{align}
Since $\|E\|=\Theta(\epsilon)$, the largest entry of $E$ is $\Theta(\epsilon)$, and the first-order contribution therefore satisfies
\begin{align}
S_{\rm VN}^{(1)} = \frac{1}{2}
\trace\left[
E f_{\rm VN}'(A)
\right]
=
O\left(
\epsilon|\log\tau|
\right).
\label{eq:app VN first singularity}
\end{align}
We next consider the second-order contribution. In the eigenbasis of
$A$, the first Fr\'echet derivative satisfies
\begin{align}
\left(
D^{(1)}_{f_{\rm VN}'(A)}[E]
\right)_{kl}
=
(f'_{\rm VN})^{[1]}[a_k,a_l] 
\left(\widetilde{E}\right)_{kl},
\label{eq:VN second divided difference}
\end{align}
where $\widetilde{E}=U^{-1}EU$, 
and $U$ diagonalizes $A$. Since $U$ is 
independent of $\tau$,  we have $\|\widetilde{E}\|=\Theta(\epsilon)$, and the largest entry of $\widetilde{E}$ is $\Theta(\epsilon)$.
For equal eigenvalues,
\begin{align}
\!\!\!(f'_{\rm VN})^{[1]}[a_\pm,a_\pm]
\!=\!
f''_{\rm VN}(a_\pm)
\!=\!
\frac{1}{\tau(1+\tau)}
=
\Theta\left(\frac{1}{\tau}\right),
\label{eq:VN diagonal DD singular}
\end{align}
whereas for opposite eigenvalues,
\begin{align}
(f'_{\rm VN})^{[1]}[a_\pm,a_\mp]
&=
\frac{
f'_{\rm VN}(a_\pm)-f'_{\rm VN}(a_\mp)
}{
a_\pm-a_\mp
} \\
& =\frac{2}{1+2\tau}\log\left( \frac{\tau}{1+\tau} \right) \\
&=
\Theta(|\log\tau|). \label{eq: VNE opposite}
\end{align}
Hence, the most singular divided-difference entries scale as
$\Theta(\tau^{-1})$. It follows from \eqref{eq:VN second divided difference} that for all $k$ and $l$,
\begin{align}
\left(
D^{(1)}_{f_{\rm VN}'(A)}[E]
\right)_{kl}
=
O\left(
\frac{\epsilon}{\tau}
\right).
\end{align}
For fixed matrix dimension, the second-order contribution therefore satisfies
\begin{align}
S_{\rm VN}^{(2)} = \frac{1}{4}
\trace\left[
E
D^{(1)}_{f_{\rm VN}'(A)}[E]
\right]
=
O\left(
\frac{\epsilon^2}{\tau}
\right).
\label{eq:app VN second singularity}
\end{align}
Therefore, the leading singular scalings of the first- and second-order VNE contributions are
\begin{align}
S_{\rm VN}^{(1)}
=
O\left(
\epsilon|\log\tau|
\right), \qquad
S_{\rm VN}^{(2)}
=
O\left(
\frac{\epsilon^2}{\tau}
\right).
\end{align}

\subsection{Quantum relative entropy}
For the QRE, we use the simplified second-order expression derived in
Theorem~\ref{thm:QRE-perturbation}. In the eigenbasis of
$A$,
\begin{align}
\left(
D^{(1)}_{f_{\rm RE}(A)}[E_1-E_2]
\right)_{kl}
=
f_{\rm RE}^{[1]}[a_k,a_l]\left(\widetilde{E}_1-\widetilde{E}_2 \right)_{kl}, \label{eq:QRE t0001}
\end{align}
where $\widetilde{E}_k=U^{-1}E_kU$, for $k = 1,2$,  
and $U$ diagonalizes $A$.
Since $\|E_1-E_2\|=\Theta(\epsilon)$ and the eigenbasis of
$A$ is independent of $\tau$, the transformed perturbation
satisfies
\begin{align}
\left(\widetilde E_1-\widetilde E_2\right)_{kl}
=
O(\epsilon) \quad \forall k,l.
\end{align}
The diagonal divided differences are
\begin{align}
f_{\rm RE}^{[1]}[a_\pm,a_\pm]
=
f_{\rm RE}'(a_\pm)
=
\frac{1}{\tau(1+\tau)}
=
\Theta(\tau^{-1}),
\label{eq:RE diagonal singularity}
\end{align}
while the off-diagonal divided differences are
\begin{align}
    f_{\rm RE}^{[1]}[a_\pm,a_\mp] = -\frac{4 \arccoth(1+2\tau)}{1+2\tau} = \Theta(|\ln \tau |). \label{eq:RE off diagonal singularity}
\end{align}
The diagonal divided differences therefore provide the dominant singular behavior. Hence, 
\begin{align}
\left(
D^{(1)}_{f_{\rm RE}(A)}[E_1-E_2]
\right)_{kl}
=
O\left(
\frac{\epsilon}{\tau}
\right).
\end{align}
Using the second-order QRE covariance-matrix contribution,
\begin{align}
D_{\rm RE,CM}^{(2)}
&=
-\frac{1}{4}
\trace\left[
(E_1-E_2)
D^{(1)}_{f_{\rm RE}(A)}
[E_1-E_2]
\right],
\end{align}
and noting that the trace is a finite sum over matrix entries, we obtain the leading singular scaling
\begin{align}
D_{\rm RE,CM}^{(2)}
=
O\left(
\frac{\epsilon^2}{\tau}
\right).
\label{eq:RE singular scaling}
\end{align}

\subsection{Quantum Chernoff quantity}

We next consider the second-order covariance-matrix contribution to
the quantum Chernoff quantity $\xi_{{\rm CE,CM},s}^{(2)}$ in \eqref{eq:QCE perturb CM}. In the eigenbasis of $A$, we have
\begin{align}
&\left(
D^{(1)}_{g_{\rm CE}^{(s)}(A)}[E_1-E_2]
\right)_{kl} \nonumber \\
&\qquad\qquad=
\left(g_{\rm CE}^{(s)}\right)^{[1]}[a_k,a_l]
\left(\widetilde E_1-\widetilde E_2\right)_{kl}, \label{eq: app QCE 0001}
\end{align}
where $\widetilde{E}_k=U^{-1}E_kU$, for $k = 1,2$,  
and $U$ diagonalizes $A$.
Since $\|E_1-E_2\|=\Theta(\epsilon)$ and the eigenbasis of
$A$ is independent of $\tau$, we have
\begin{align}
\left(\widetilde E_1-\widetilde E_2\right)_{kl}
=
O(\epsilon) \quad \forall k,l .
\end{align}
For fixed $s\in[0,1]$, the diagonal divided differences of
$g_{\rm CE}^{(s)}$ are
\begin{align}
\left(g_{\rm CE}^{(s)}\right)^{[1]}[a_\pm,a_\pm]
&=
g_{\rm CE}^{(s)\prime}(a_\pm)
\nonumber\\
&=
-\frac{
2s[\tau(1+\tau)]^{s-1}
}{
\left[(1+\tau)^s-\tau^s\right]^2
} \\
&=
\Theta(\tau^{s-1}).
\label{eq:CE diagonal s singularity}
\end{align}
Similarly,
\begin{align}
\left(g_{\rm CE}^{(1-s)}\right)^{[1]}[a_\pm,a_\pm]
&=
g_{\rm CE}^{(1-s)\prime}(a_\pm)
=
\Theta(\tau^{-s}).
\label{eq:CE diagonal 1-s singularity}
\end{align}
The corresponding off-diagonal divided differences remain nonsingular
\begin{align}
    \left(g_{\rm CE}^{(1-s)}\right)^{[1]}[a_\pm,a_\mp] &= \frac{2\left(\tau^s + (1+\tau)^s\right)}{(1+2\tau)(\tau^s-(1+\tau)^s)} \nonumber \\
    & = O(1). \label{eq:CE gs non singular 2}
\end{align}
Similarly, $ \left(g_{\rm CE}^{( s)}\right)^{[1]}[a_\pm,a_\mp]  = O(1)$. 
Furthermore, the zeroth-order matrix
$(1/R_{\rm CE}^{(s)})(A)$ remains nonsingular in the limit
$\tau\rightarrow0$:
\begin{align}
\left(
1/R_{\rm CE}^{(s)}
\right)
(A)
=
O(1). \label{eq:CE RA non singular}
\end{align}
Substituting these scalings into \eqref{eq:QCE perturb CM}, and noting that the trace is a finite
sum for fixed matrix dimension, the dominant singular contribution arises from the diagonal divided differences, giving
$\left(
\epsilon\tau^{s-1}
\right)
\left(
\epsilon\tau^{-s}
\right)
=
\frac{\epsilon^2}{\tau}$.
Hence,
\begin{align}
\xi_{{\rm CE,CM},s}^{(2)}
=
O\left(
\frac{\epsilon^2}{\tau}
\right).
\label{eq:CE singular scaling}
\end{align}
In particular, the same scaling holds at the perturbatively optimal
value $s=1/2$.

\subsection{Quantum fidelity}
Finally, consider the second-order covariance-matrix contribution to
the perturbative QF in \eqref{eq:QF trace representation},
\begin{align}
F_{\rm CM}^{(2)}
=
-\sum_{k,l}
\frac{
\left(
\widetilde{E}_1-\widetilde{E}_2
\right)_{kl}
\left(
\widetilde{E}_1-\widetilde{E}_2
\right)_{lk}
}{
4+16a_ka_l
},
\label{eq:QF singular starting}
\end{align}
where $a_k,a_l\in\{a_+,a_-\}$, $\widetilde{E}_k=U^{-1}E_kU$, for $k = 1,2$,  
and $U$ diagonalizes $A$.
Since $\|E_1-E_2\|=\Theta(\epsilon)$ and the eigenbasis of
$A$ is independent of $\tau$, the transformed perturbation
satisfies
\begin{align}
\left(
\widetilde E_1-\widetilde E_2
\right)_{kl}
=
O(\epsilon) \quad \forall k,l.
\end{align}
Therefore, each numerator in
 \eqref{eq:QF singular starting} is $O(\epsilon^2)$, and the
singular scaling is determined by the most singular denominator
$4+16a_ka_l$. For eigenvalues of opposite
sign,
\begin{align}
4+16a_+a_-
=
4+16\left(
\frac{1}{2}+\tau
\right)^2
=
O(1),
\end{align}
which is non-singular. The singular terms instead arise when the two eigenvalues have the
same sign:
\begin{align}
4+16a_\pm^2
&=
4-16\left(
\frac{1}{2}+\tau
\right)^2
\nonumber\\
&=
-16\tau(1+\tau)
=
\Theta(\tau).
\end{align}
Thus, the most singular terms in the finite sum
 \eqref{eq:QF singular starting} scale as
$\epsilon^2/\tau$. Therefore,
\begin{align}
F_{\rm CM}^{(2)}
=
O\left(
\frac{\epsilon^2}{\tau}
\right).
\label{eq:QF singular scaling}
\end{align}
\textit{Summary}---Thus, the second-order covariance-matrix contributions to the VNE, QRE,
quantum Chernoff quantity, and QF all exhibit a $\tau^{-1}$ singular
scaling as the reference state's symplectic eigenvalues approach $1/2$,
while the first-order VNE contribution scales only logarithmically
with $\tau$.

\section{Proof of Criterion~\ref{thm:leading pure mode perturbation}}
\label{app:proof mixed pure theorem}
In this appendix, we consider reference states whose covariance matrices contain symplectic eigenvalues both equal to and greater than $1/2$, and prove Criterion~\ref{thm:leading pure mode perturbation}.
Since Criterion~\ref{thm:leading pure mode perturbation} concerns the covariance-matrix contributions, we set
all mean vectors to zero throughout this appendix.

Let $n_{\rm g}$ and $n_{\rm e}$ denote the numbers of symplectic eigenvalues
strictly greater than and equal to $1/2$, respectively. As in
Appendix~\ref{app:singularity}, we introduce an auxiliary parameter
$\tau>0$ to characterize the singular behavior of the perturbative
coefficients associated with the eigenspaces that approach the
symplectic eigenvalue $1/2$. In this section, 
define 
\begin{align}
A_\tau = \Omega \sigma_0
=
U\Gamma_\tau U^{-1},
\end{align}
where
\begin{align}
\Gamma_\tau
&=
\operatorname{diag}
\left(
i\Lambda_{\rm g},
ir_\tau I_{n_{\rm e}},
-i\Lambda_{\rm g},
-ir_\tau I_{n_{\rm e}}
\right),
\label{eq:app theorem gamma}
\\
r_\tau
&=
\frac12+\tau , \\
\Lambda_{\rm g}
&=
\operatorname{diag}
(\nu_1,\ldots,\nu_{n_{\rm g}}),
\qquad
\nu_j>\frac12 .
\end{align}
Without loss of generality, we group the symplectic eigenvalues into those greater than $1/2$ and those equal to $1/2$.
We express the perturbation matrix in the eigenbasis of $A_\tau$ and
partition it according to the four eigenspaces appearing in \eqref{eq:app theorem gamma}:
\begin{align}
\widetilde E
=
U^{-1}EU
=
\begin{pmatrix}
\widetilde E_{11} & \widetilde E_{12}
& \widetilde E_{13} & \widetilde E_{14}\\
\widetilde E_{21} & \widetilde E_{22}
& \widetilde E_{23} & \widetilde E_{24}\\
\widetilde E_{31} & \widetilde E_{32}
& \widetilde E_{33} & \widetilde E_{34}\\
\widetilde E_{41} & \widetilde E_{42}
& \widetilde E_{43} & \widetilde E_{44}
\end{pmatrix},
\label{eq:app theorem E blocks}
\end{align}
and similarly, for the two-state quantities,
\begin{align}
\Delta\widetilde E
&=
U^{-1}(E_1-E_2)U
\nonumber\\
&=
\begin{pmatrix}
\Delta\widetilde E_{11} & \Delta\widetilde E_{12}
& \Delta\widetilde E_{13} & \Delta\widetilde E_{14}\\
\Delta\widetilde E_{21} & \Delta\widetilde E_{22}
& \Delta\widetilde E_{23} & \Delta\widetilde E_{24}\\
\Delta\widetilde E_{31} & \Delta\widetilde E_{32}
& \Delta\widetilde E_{33} & \Delta\widetilde E_{34}\\
\Delta\widetilde E_{41} & \Delta\widetilde E_{42}
& \Delta\widetilde E_{43} & \Delta\widetilde E_{44}
\end{pmatrix}.
\label{eq:app theorem Delta E blocks}
\end{align}
Since the matrix dimension is fixed, $U$ is independent of
$\epsilon$, and
\begin{align}
\|\widetilde E\|
=
\Theta(\epsilon),
\qquad
\|\Delta\widetilde E\|
=
\Theta(\epsilon),
\end{align}
 every block in \eqref{eq:app theorem E blocks} and \eqref{eq:app theorem Delta E blocks} is $O(\epsilon)$.

Before proving Criterion~\ref{thm:leading pure mode perturbation}, we show that the diagonal blocks associated with modes with symplectic eigenvalue $r_\tau$,
$\widetilde E_{22}$, $\widetilde E_{44}$,
$\Delta\widetilde E_{22}$, and
$\Delta\widetilde E_{44}$, are anti-Hermitian.
Since the covariance matrix of the reference state $\sigma_0$ is real
and symmetric, we have
\begin{align}
\left(
\sigma_0^{1/2}\Omega\sigma_0^{1/2}
\right)^\dagger
=
-\sigma_0^{1/2}\Omega\sigma_0^{1/2},
\end{align}
so $\sigma_0^{1/2}\Omega\sigma_0^{1/2}$ is anti-Hermitian.
Thus, it can be unitarily diagonalized as
\begin{align}
\sigma_0^{1/2}\Omega\sigma_0^{1/2}
=
V\Gamma V^\dagger,
\qquad
V^{-1}=V^\dagger , \label{eq:tm5 102}
\end{align}
where $\Gamma$ is diagonal and $V$ is the eigenbasis.
Since
\begin{align}
A_\tau
=
\Omega\sigma_0
=
\sigma_0^{-1/2}
\left(
\sigma_0^{1/2}\Omega\sigma_0^{1/2}
\right)
\sigma_0^{1/2}, \label{eq:tm5 103}
\end{align}
substituting \eqref{eq:tm5 102} into \eqref{eq:tm5 103}, we have
\begin{align}
A_\tau
=
\sigma_0^{-1/2}
V\Gamma V^\dagger
\sigma_0^{1/2}.
\end{align}
Thus, for diagonalizable matrix $A_\tau=U\Gamma_{\tau} U^{-1}$,
we take
\begin{align}
U
=
\sigma_0^{-1/2}V, \quad \Gamma = \Gamma_\tau, 
\quad
U^{-1}
=
V^\dagger\sigma_0^{1/2}. \label{eq:tm5 104}
\end{align}
Moreover, using \eqref{eq:tm5 104}, we have
\begin{align}
    U^{\dagger} = V^{\dagger} \sigma_0^{- \frac{1}{2}}  = V^{\dagger} \sigma_0^{ \frac{1}{2}} \sigma_0^{-1}    = U^{-1} \sigma_0^{-1}. \label{eq:app U inverse relation}
\end{align}
We now define the selection matrix that extracts the submatrix $\tilde{E}_{22}$ as
\begin{align}
e_+
=
\begin{bmatrix}
0_{n_{\rm g}\times n_{\rm e}}\\
I_{n_{\rm e}}\\
0_{n_{\rm g}\times n_{\rm e}}\\
0_{n_{\rm e}\times n_{\rm e}}
\end{bmatrix},
\end{align}
and
\begin{align}
U_+
=
Ue_+.
\end{align}
The columns of $U_+$ are the eigenvectors associated with the eigenvalue $ir_\tau$ since
\begin{align}
    A_\tau U_+ = A_\tau Ue_+ = U \Gamma_\tau e_+ = U ir_{\tau} e_+ = ir_{\tau} U_+.
\label{eq:app U plus eigen}
\end{align}
For $E=\Omega\Delta\sigma$, with
$\Delta\sigma=\Delta\sigma^\dagger$, we have
\begin{align}
\widetilde{E}_{22}  &= e_+^{\dagger} U^{-1} E U e_+ \\
& = (e_+^{\dagger} U^{\dagger} \sigma_0 \Omega) \Delta \sigma U e_+ \label{eq:app tm5 104.5} \\
& = -(A_\tau U_+)^{\dagger} \Delta \sigma U_+ \label{eq:app tm5 105} \\
& = ir_{\tau} U_+^{\dagger } \Delta \sigma U_+,\label{eq:app E22 Hermitian structure}
\end{align}
where \eqref{eq:app tm5 104.5} follows from \eqref{eq:app U inverse relation},  \eqref{eq:app tm5 105} follows from taking the adjoint of \eqref{eq:app U plus eigen} and using
$A_\tau^\dagger=-\sigma_0\Omega$, and \eqref{eq:app E22 Hermitian structure} follows from  \eqref{eq:app U plus eigen}.
Since $\Delta\sigma$ is Hermitian,
\begin{align}
\left(
U_+^\dagger\Delta\sigma U_+
\right)^\dagger
=
U_+^\dagger\Delta\sigma U_+,
\end{align}
thus, \eqref{eq:app E22 Hermitian structure} implies
\begin{align}
\widetilde E_{22}^\dagger
=
-\widetilde E_{22}.
\end{align}
Hence, $\widetilde E_{22}$ is anti-Hermitian.

The same argument applied to the eigenvalue $-ir_\tau$ shows that
$\widetilde E_{44}$ is anti-Hermitian. Replacing $\Delta\sigma$ by
$\Delta\sigma_1-\Delta\sigma_2$, which is also Hermitian, gives the
same result for $\Delta\widetilde E_{22}$ and
$\Delta\widetilde E_{44}$.
Consequently, for any of these blocks $B$,
\begin{align}
-\operatorname{tr}(B^2)
=
\operatorname{tr}(B^\dagger B)
=
\|B\|_F^2.
\label{eq:app anti Hermitian norm}
\end{align}
Since the matrix dimension is fixed, the Frobenius and operator norms have the same order of scaling. Therefore,
\begin{align}
\|B\|=\Theta(\epsilon)
\quad\Longleftrightarrow\quad
-\operatorname{tr}(B^2)=\Theta(\epsilon^2). \label{eq:app norm relation}
\end{align}
In the following subsections,
we show that the perturbative coefficients multiplying
$\widetilde E_{22}$ and $\widetilde E_{44}$ for the VNE, and
$\Delta\widetilde E_{22}$ and $\Delta\widetilde E_{44}$ for the
two-state quantities, exhibit the strongest singular behavior as
$\tau\rightarrow0$.

\subsection{Von Neumann entropy}

We first consider the VNE. Its first- and second-order
covariance-matrix contributions are
\begin{align}
S_{\rm VN}^{(1)}
&=
\frac12
\operatorname{tr}
\left[
E f'_{\rm VN}(A_\tau)
\right],
\label{eq:app F VNE first}
\\
S_{\rm VN}^{(2)}
&=
\frac14
\operatorname{tr}
\left[
E D^{(1)}_{f'_{\rm VN}(A_\tau)}[E]
\right].
\label{eq:app F VNE second}
\end{align}
We first consider the singular terms in the first-order contribution. In the eigenbasis of
$A_\tau$,
\begin{align}
S_{\rm VN}^{(1)}
=
\frac12
\operatorname{tr}
\left[
\widetilde E f'_{\rm VN}(\Gamma_\tau)
\right].
\end{align}
Since $f'_{\rm VN}(\Gamma_\tau)$ is diagonal, only the diagonal
blocks of $\widetilde E$ contribute:
\begin{align}
S_{\rm VN}^{(1)}
&=
\frac12
\operatorname{tr}
\left[
\widetilde E_{11}f'_{\rm VN}(i\Lambda_{\rm g})
\right]
+
\frac12
f'_{\rm VN}(ir_\tau)
\operatorname{tr}\widetilde E_{22}
\nonumber\\
&\quad+
\frac12
\operatorname{tr}
\left[
\widetilde E_{33}f'_{\rm VN}(-i\Lambda_{\rm g})
\right]
+
\frac12
f'_{\rm VN}(-ir_\tau)
\operatorname{tr}\widetilde E_{44}.
\label{eq:app F VNE first blocks}
\end{align}
For the symplectic eigenvalues $\nu_j>1/2$, we have
\begin{align}
f'_{\rm VN}(\pm i\nu_j)
=
O(1),
\end{align}
so the first and the third terms in \eqref{eq:app F VNE first blocks} are nonsingular. 
For the eigenvalues $\pm ir_\tau$, \eqref{eq:VN first derivative singular} in Appendix~\ref{app:singularity} gives
\begin{align}
f'_{\rm VN}(\pm ir_\tau) =
\Theta(|\log\tau|) .
\label{eq:app F VNE pure derivative}
\end{align}
with
\begin{align}
f'_{\rm VN}(-ir_\tau)
=
-f'_{\rm VN}(ir_\tau).
\end{align}
Therefore, the second and fourth terms in \eqref{eq:app F VNE first blocks} combine as
\begin{align}
\frac12
f'_{\rm VN}(ir_\tau)
\left(
\operatorname{tr}\widetilde E_{22}
-
\operatorname{tr}\widetilde E_{44}
\right).
\end{align}
Thus, by \eqref{eq:criterion 1} in Criterion~\ref{thm:leading pure mode perturbation},
\begin{align}
\left|
\operatorname{tr}\widetilde E_{22}
-
\operatorname{tr}\widetilde E_{44}
\right|
=
\Theta(\epsilon),
\label{eq:app F VNE trace condition}
\end{align}
then
\begin{align}
S_{\rm VN}^{(1)}
=
\Theta\!\left(
\epsilon|\log\tau|
\right).
\label{eq:app F VNE first scaling}
\end{align}
We next consider the singular terms in the second-order contribution in \eqref{eq:app F VNE second}. In the eigenbasis of
$A_\tau$,
\begin{align}
\left(
D^{(1)}_{f'_{\rm VN}(A_\tau)}[E]
\right)_{kl}
=
(f'_{\rm VN})^{[1]}[a_k,a_l]
\widetilde E_{kl},
\end{align}
so that
\begin{align}
S_{\rm VN}^{(2)}
=
\frac14
\sum_{k,l}
(f'_{\rm VN})^{[1]}[a_k,a_l]
\widetilde E_{lk}\widetilde E_{kl}.
\label{eq:app F VNE second sum}
\end{align}
Here, $a_k$ and $a_l$ are eigenvalues of $\Gamma_\tau$ in \eqref{eq:app theorem gamma}.
Equation \eqref{eq:VN diagonal DD singular} shows that
\begin{align}
(f'_{\rm VN})^{[1]}[\pm ir_\tau,\pm ir_\tau] =
\Theta(\tau^{-1}).
\label{eq:app F VNE same pure}
\end{align}
These coefficients multiply the entries within $\widetilde E_{22}$ and $\widetilde E_{44}$ blocks.
Furthermore, \eqref{eq: VNE opposite} shows
\begin{align}
(f'_{\rm VN})^{[1]}[\pm ir_\tau,\mp ir_\tau]
=
\Theta(|\log\tau|).
\label{eq:app F VNE pure mixed}
\end{align}
If one eigenvalue is of the form $\pm ir_\tau$ and the other is of
the form $\pm i\nu_j$ with $\nu_j>1/2$, the corresponding divided
difference scales as $f'_{\rm VN}(\pm ir_\tau)$, and therefore
\begin{align}
(f'_{\rm VN})^{[1]}[ ir_\tau,\pm i\nu_j], \  (f'_{\rm VN})^{[1]}[- ir_\tau,\pm i\nu_j] 
=
\Theta(|\log\tau|).
\end{align}
Finally, if both eigenvalues are of the form $\pm i\nu_j$ with $\nu_j>1/2$, the coefficients are nonsingular. 
Thus, the most singular coefficients in \eqref{eq:app F VNE second sum}  are those multiplying the $\widetilde E_{22}$ and $\widetilde E_{44}$ blocks, i.e., $k,l\in [n_{\rm g}+1,n_{\rm g}+n_{\rm e}]$ and $k,l\in [2n_{\rm g}+n_{\rm e}+1,2n_{\rm g}+2n_{\rm e}]$. 

Thus, 
\begin{align}
   & \frac{1}{4} \sum_{k,l\in [n_{\rm g}+1,n_{\rm g}+n_{\rm e}]}  (f'_{\rm VN})^{[1]}[a_k,a_l]
\widetilde E_{lk}\widetilde E_{kl} \nonumber \\
 = & \frac{1}{4} (f'_{\rm VN})^{[1]}[ir_\tau,ir_\tau] \sum_{k,l\in [n_{\rm g}+1,n_{\rm g}+n_{\rm e}]} 
\widetilde E_{lk}\widetilde E_{kl} \\
=&  -\frac{1}{4} (f'_{\rm VN})^{[1]}[ir_\tau,ir_\tau] \trace\left( \widetilde E_{22}  \widetilde E_{22}^{\dagger}\right) \label{eq:tm5 105} \\
=& -\frac{1}{4} (f'_{\rm VN})^{[1]}[ir_\tau,ir_\tau] \| \widetilde E_{22} \|_{F}^2 \label{eq:tm5 106}
\end{align}
where \eqref{eq:tm5 105} follows from the fact that $ \widetilde E_{22} $ is anti-Hermitian.
Similarly,
\begin{align}
       & \frac{1}{4} \sum_{k,l\in [2n_{\rm g}+n_{\rm e}+1,2n_{\rm g}+2n_{\rm e}] }  (f'_{\rm VN})^{[1]}[a_k,a_l]
\widetilde E_{lk}\widetilde E_{kl} \nonumber \\
=& -\frac{1}{4} (f'_{\rm VN})^{[1]}[-ir_\tau,-ir_\tau] \| \widetilde E_{44} \|_{F}^2. \label{eq:tm5 107}
\end{align}
From \eqref{eq:app F VNE trace condition}, at least one of
$|\operatorname{tr}\widetilde E_{22}|$ and
$|\operatorname{tr}\widetilde E_{44}|$ is $\Theta(\epsilon)$.
Since, for a fixed-dimensional matrix $B$,
\begin{align}
|\operatorname{tr}B|
\leq
\sqrt{n_{\rm e}}\,\|B\|_F,
\end{align}
and both blocks are $\Theta(\epsilon)$, at least one of
$\|\widetilde E_{22}\|_F$ and
$\|\widetilde E_{44}\|_F$ is $\Theta(\epsilon)$.
Furthermore, it follows from \eqref{eq:app F VNE same pure} that either 
\eqref{eq:tm5 106} or \eqref{eq:tm5 107} is $ \Theta(\epsilon^2 \tau^{-1} )$. Since $(f'_{\rm VN})^{[1]}[-ir_\tau,-ir_\tau] = (f'_{\rm VN})^{[1]}[ir_\tau,ir_\tau] $, \eqref{eq:tm5 106} and  \eqref{eq:tm5 107} have the
same sign and cannot cancel each other.
All remaining terms in \eqref{eq:app F VNE second sum} are 
asymptotically smaller than $\Theta(\epsilon^2\tau^{-1})$ as
$\tau\rightarrow0$.
Thus,
\begin{align}
    S_{\rm VN}^{(2)} = \Theta(\epsilon^2 \tau^{-1} ). \label{eq:app F VNE second scaling}
\end{align}
Therefore, under the condition \eqref{eq:app F VNE trace condition}, the leading VNE variation is determined by the modes with symplectic eigenvalue $r_\tau$. Tracing out the modes with symplectic eigenvalue $\nu_j>1/2$ therefore does not change the coefficients of the leading perturbative scaling of the VNE variation.

Using the scaling $\tau=\Theta(\epsilon)$ from
Sec.~\ref{subsec: pure scaling}, the leading first-order contribution \eqref{eq:app F VNE first scaling} yields 
\begin{align}
S_{\rm VN}^{(1)}
=
\Theta(\epsilon|\log\epsilon|),
\end{align}
while the second-order contribution \eqref{eq:app F VNE second scaling} yields
\begin{align}
S_{\rm VN}^{(2)}
=
\Theta(\epsilon).
\end{align}

\subsection{Quantum relative entropy}

We next consider the QRE. As shown in
Sec.~\ref{sec: perturbative formula}, the leading
covariance-matrix contribution is
\begin{align}
D_{\rm RE,CM}^{(2)}
&=
-\frac14
\operatorname{tr}
\left[
\Delta E
D^{(1)}_{f_{\rm RE}(A_\tau)}[\Delta E]
\right] \nonumber\\
& = -\frac14
\operatorname{tr}
\left[
\Delta \tilde{E}
D^{(1)}_{f_{\rm RE}(\Gamma_\tau)}[\Delta \tilde{E}]
\right]
\label{eq:app F QRE second}
\end{align}
where $\Delta E = E_1-E_2$.
In the eigenbasis of $A_\tau$,
\begin{align}
\left(D^{(1)}_{f_{\rm RE}(\Gamma_\tau)}[\Delta \tilde{E}]
\right)_{kl}
=
f_{\rm RE}^{[1]}[a_k,a_l]
\Delta\widetilde E_{kl},
\end{align}
so that
\begin{align}
D_{\rm RE,CM}^{(2)}
=
-\frac14
\sum_{k,l}
f_{\rm RE}^{[1]}[a_k,a_l]
\Delta\widetilde E_{lk}
\Delta\widetilde E_{kl}.
\label{eq:app F QRE second sum}
\end{align}
Here, $a_k$ and $a_l$ are eigenvalues of $\Gamma_\tau$ in \eqref{eq:app theorem gamma}.

We now identify the most singular perturbative coefficients in $ f_{\rm RE}^{[1]}[a_k,a_l] $.
In \eqref{eq:RE diagonal singularity}, we show
\begin{align}
f_{\rm RE}^{[1]}[\pm ir_\tau,\pm ir_\tau]
=
\Theta(\tau^{-1}).
\label{eq:app F QRE same pure}
\end{align}
These coefficients multiply the entries within
$\Delta\widetilde E_{22}$ and
$\Delta\widetilde E_{44}$ blocks.
In \eqref{eq:RE off diagonal singularity}, we show
\begin{align}
f_{\rm RE}^{[1]}[\pm ir_\tau,\mp ir_\tau]
=
\Theta(|\log\tau|).
\label{eq:app F QRE pure mixed}
\end{align}
If one argument is $\pm ir_\tau$ and the other is $\pm i\nu_j$ with $\nu_j>1/2$, the corresponding divided
difference scales as
\begin{align}
f_{\rm RE}^{[1]}[ ir_\tau,\pm i\nu_j], f_{\rm RE}^{[1]}[- ir_\tau,\pm i\nu_j] 
=
\Theta(|\log\tau|).
\end{align}
Finally, if both arguments are of the form $\pm i\nu_j$ with $\nu_j>1/2$, the coefficient is nonsingular.
Therefore, 
the most singular QRE coefficients are those multiplying the entries within 
$\Delta\widetilde E_{22}$ and
$\Delta\widetilde E_{44}$ blocks, i.e., $k,l\in [n_{\rm g}+1,n_{\rm g}+n_{\rm e}]$ and $k,l\in [2n_{\rm g}+n_{\rm e}+1,2n_{\rm g}+2n_{\rm e}]$.

The contributions
from the $\Delta\widetilde E_{22} $ and  $\Delta\widetilde E_{44}$ blocks to \eqref{eq:app F QRE second sum} are 
\begin{align}
    &-\frac14
\sum_{k,l\in [n_{\rm g}+1,n_{\rm g}+n_{\rm e}]}
f_{\rm RE}^{[1]}[a_k,a_l]
\Delta\widetilde E_{lk}
\Delta\widetilde E_{kl} \nonumber \\
=& -\frac14  f_{\rm RE}^{[1]}[i r_{\tau},i r_{\tau}]  \sum_{k,l\in [n_{\rm g}+1,n_{\rm g}+n_{\rm e}]}  \Delta\widetilde E_{lk}
\Delta\widetilde E_{kl}  \\
= & - \frac14  f_{\rm RE}^{[1]}[i r_{\tau},i r_{\tau}]  \trace\left( (\Delta\widetilde E_{22})^2  \right), \label{eq:app thm 5 QRE 101}
\end{align}
and 
\begin{align}
    &-\frac14
\sum_{k,l\in [2n_{\rm g}+n_{\rm e}+1,2n_{\rm g}+2n_{\rm e}]}
f_{\rm RE}^{[1]}[a_k,a_l]
\Delta\widetilde E_{lk}
\Delta\widetilde E_{kl} \nonumber \\
=& - \frac14  f_{\rm RE}^{[1]}[-i r_{\tau},-i r_{\tau}]  \sum_{k,l\in [2n_{\rm g}+n_{\rm e}+1,2n_{\rm g}+2n_{\rm e}]}  \Delta\widetilde E_{lk}
\Delta\widetilde E_{kl}  \\
= & - \frac14  f_{\rm RE}^{[1]}[-i r_{\tau},-i r_{\tau}]  \trace\left( (\Delta\widetilde E_{44})^2  \right). \label{eq:app thm 5 QRE 102}
\end{align}
Therefore, under the condition \eqref{eq:criterion 2}
\begin{align}
\max\left\{
\|\Delta\widetilde E_{22}\|,
\|\Delta\widetilde E_{44}\|
\right\}
=
\Theta(\epsilon),
\label{eq:app F QRE condition}
\end{align}
\eqref{eq:app norm relation} implies that at least one of
$-\operatorname{tr}
\left[
(\Delta\widetilde E_{22})^2
\right]$ or $-
\operatorname{tr}
\left[
(\Delta\widetilde E_{44})^2
\right]$
is $\Theta(\epsilon^2)$.
Together with  \eqref{eq:app F QRE same pure}, at least one of \eqref{eq:app thm 5 QRE 101} or \eqref{eq:app thm 5 QRE 102} is $\Theta(\epsilon^2\tau^{-1})$. Furthermore, \eqref{eq:RE diagonal singularity} gives   $f_{\rm RE}^{[1]}[-i r_{\tau},-i r_{\tau}]  = f_{\rm RE}^{[1]}[i r_{\tau},i r_{\tau}]   $. Thus, 
the two contributions \eqref{eq:app thm 5 QRE 101} and \eqref{eq:app thm 5 QRE 102} have the same sign and cannot cancel.
All remaining terms in \eqref{eq:app F QRE second sum} are at most
$O(\epsilon^2|\log\tau|)$ or $O(\epsilon^2)$ and are therefore
asymptotically subleading. Therefore,
\begin{align}
D_{\rm RE,CM}^{(2)}
=
\Theta(\epsilon^2\tau^{-1}).
\end{align}
Therefore, under the condition \eqref{eq:app F QRE condition}, the leading QRE is determined by the modes with symplectic eigenvalue $r_\tau$. Tracing out the modes with symplectic eigenvalue $\nu_j>1/2$ does not change the coefficients of the leading perturbative scaling of QRE.

Using the scaling $\tau=\Theta(\epsilon)$ from
Sec.~\ref{subsec: pure scaling}, the 
covariance-matrix contribution satisfies 
\begin{align}
D_{\rm RE,CM}^{(2)}
=
\Theta(\epsilon).
\end{align}

\subsection{Quantum Chernoff quantity}

We next consider the quantum Chernoff quantity for a fixed
$s\in[0,1]$. As shown in Sec.~\ref{sec: perturbative formula}, the
leading covariance-matrix contribution is
\begin{align}
\xi_{{\rm CE},{\rm CM},s}^{(2)}
&=
\frac14
\operatorname{tr}
\left[
\left(\frac{1}{R_{\rm CE}^{(s)}}\right)(A_\tau)
D^{(1)}_{g_{\rm CE}^{(s)}(A_\tau)}[\Delta E]
\right.
\nonumber\\
&\left.
\qquad\times
\left(\frac{1}{R_{\rm CE}^{(s)}}\right)(A_\tau)
D^{(1)}_{g_{\rm CE}^{(1-s)}(A_\tau)}[\Delta E]
\right],
\label{eq:app F QCE second}
\end{align}
where $\Delta E=E_1-E_2$.
In the eigenbasis of $A_\tau$, let
\begin{align}
H_s
=
\left(1/R_{\rm CE}^{(s)}\right)(\Gamma_\tau).
\end{align}
Since $H_s$ is diagonal, \eqref{eq:app F QCE second} becomes
\begin{align}
\xi_{{\rm CE},{\rm CM},s}^{(2)}
& =
\frac14
\sum_{k,l}
(H_s)_{kk}
(H_s)_{ll}
\left(g_{\rm CE}^{(s)}\right)^{[1]}[a_k,a_l] \nonumber\\
&\qquad \times
\left(g_{\rm CE}^{(1-s)}\right)^{[1]}[a_l,a_k]
\Delta\widetilde E_{kl}
\Delta\widetilde E_{lk}.
\label{eq:app F QCE second sum}
\end{align}
Here, $a_k$ and $a_l$ are eigenvalues of $\Gamma_\tau$ in \eqref{eq:app theorem gamma}.
We now identify the most singular perturbative coefficients in
\eqref{eq:app F QCE second sum}. 
Equation \eqref{eq:CE RA non singular} shows that the factors 
\begin{align}
     \left(1/R_{\rm CE}^{(s)}\right)(\pm ir_\tau)
\end{align}
 remain finite and nonzero as $\tau\rightarrow0$. Thus, $(H_s)_{kk}$ and $(H_s)_{ll}$ are nonsingular.
Equations \eqref{eq:CE diagonal s singularity} and  \eqref{eq:CE diagonal 1-s singularity} give 
\begin{align}
\left(g_{\rm CE}^{(s)}\right)'(\pm ir_\tau)
&=
\Theta(\tau^{s-1}),
\\
\left(g_{\rm CE}^{(1-s)}\right)'(\pm ir_\tau)
&=
\Theta(\tau^{-s}).
\end{align}
Therefore,
\begin{align}
\left(g_{\rm CE}^{(s)}\right)'(\pm ir_\tau)
\left(g_{\rm CE}^{(1-s)}\right)'(\pm ir_\tau)
=
\Theta(\tau^{-1}).
\label{eq:app F QCE same pure}
\end{align}
 Hence, the
coefficients multiplying the entries within the
$\Delta\widetilde E_{22}$ and $\Delta\widetilde E_{44}$ blocks scale
as $\Theta(\tau^{-1})$.
Furthermore, \eqref{eq:CE gs non singular 2} gives that for
\begin{align}
a_k=\pm ir_\tau,
\qquad
a_l=\mp ir_\tau,
\end{align}
the corresponding divided differences remain finite as
$\tau\rightarrow0$.
Similarly, if one eigenvalue is $\pm ir_\tau$ and the
other is $\pm i\nu_j$ with $\nu_j>1/2$, the corresponding 
coefficients are nonsingular.
Finally, if both eigenvalues are of the form $\pm i\nu_j$ with
$\nu_j>1/2$, the coefficients are nonsingular.

Therefore, the most singular coefficients in \eqref{eq:app F QCE second sum} are those multiplying the
$\Delta\widetilde E_{22}$ and $\Delta\widetilde E_{44}$ blocks, i.e., 
$k,l\in[n_{\rm g}+1,n_{\rm g}+n_{\rm e}]$
and 
$k,l\in
[2n_{\rm g}+n_{\rm e}+1,2n_{\rm g}+2n_{\rm e}]$,
respectively.
Define
\begin{align}
C_s^{+}(\tau)
&\!=\!
\left[
\left( 1/{R_{\rm CE}^{(s)}}\right)(ir_\tau)
\right]^2
\left(g_{\rm CE}^{(s)}\right)'(ir_\tau)
\left(g_{\rm CE}^{(1-s)}\right)'(ir_\tau),
\\
C_s^{-}(\tau)
&=
\left[
\left(1/{R_{\rm CE}^{(s)}}\right)(-ir_\tau)
\right]^2
\left(g_{\rm CE}^{(s)}\right)'(-ir_\tau) \nonumber\\
&\qquad\qquad\qquad\times
\left(g_{\rm CE}^{(1-s)}\right)'(-ir_\tau).
\end{align}
Using
$\left(1/R_{\rm CE}^{(s)}\right)(ir_\tau)
=
-\left(1/R_{\rm CE}^{(s)}\right)(-ir_\tau)$, together with \eqref{eq:CE diagonal s singularity} and
\eqref{eq:CE diagonal 1-s singularity}, we obtain
\begin{align}
C_s^+(\tau)
=
C_s^-(\tau)
=
\Theta(\tau^{-1}).
\label{eq:app tm5 100}
\end{align}
The contribution from the $\Delta\widetilde E_{22}$ block is
\begin{align}
&\frac14
\sum_{k,l\in[n_{\rm g}+1,n_{\rm g}+n_{\rm e}]}
(H_s)_{kk}(H_s)_{ll}
\left(g_{\rm CE}^{(s)}\right)^{[1]}[a_k,a_l] \nonumber\\
& \qquad\qquad \qquad \times 
\left(g_{\rm CE}^{(1-s)}\right)^{[1]}[a_l,a_k]
\Delta\widetilde E_{kl}\Delta\widetilde E_{lk}
\nonumber\\
=&
\frac14
C_s^+(\tau)
\operatorname{tr}
\left[
(\Delta\widetilde E_{22})^2 \right] \nonumber\\
=&
-\frac14
C_s^+(\tau)
\|\Delta\widetilde E_{22}\|_F^2,
\label{eq:app thm 5 QCE 101}
\end{align}
where the last equality follows from the anti-Hermiticity of
$\Delta\widetilde E_{22}$. Similarly, the contribution from the
$\Delta\widetilde E_{44}$ block is
\begin{align}
-\frac14
C_s^-(\tau)
\|\Delta\widetilde E_{44}\|_F^2.
\label{eq:app thm 5 QCE 102}
\end{align}
Under the condition \eqref{eq:criterion 2},
\begin{align}
\max\left\{
\|\Delta\widetilde E_{22}\|,
\|\Delta\widetilde E_{44}\|
\right\}
=
\Theta(\epsilon),
\label{eq:app F QCE condition}
\end{align}
\eqref{eq:app norm relation} implies that at least one of $\|\Delta\widetilde E_{22}\|_F^2 $ or $\|\Delta\widetilde E_{44}\|_F^2$
is $\Theta(\epsilon^2)$. Therefore, it follows from \eqref{eq:app tm5 100} that at least one of \eqref{eq:app thm 5 QCE 101} and
\eqref{eq:app thm 5 QCE 102} is
$\Theta(\epsilon^2\tau^{-1})$.
Since $C_s^+(\tau)=C_s^-(\tau)$ as in \eqref{eq:app tm5 100}, the contributions in \eqref{eq:app thm 5 QCE 101} and
\eqref{eq:app thm 5 QCE 102} cannot cancel. All remaining terms in
\eqref{eq:app F QCE second sum} are at most $O(\epsilon^2)$ and
are therefore asymptotically subleading. Hence,
\begin{align}
\xi_{{\rm CE},{\rm CM},s}^{(2)}
=
\Theta(\epsilon^2\tau^{-1}).
\end{align}
Therefore, under the condition \eqref{eq:app F QCE condition}, the leading quantum Chernoff
quantity is determined by the modes with symplectic eigenvalue $r_\tau$.
Tracing out the modes whose symplectic eigenvalues satisfy
$\nu_j>1/2$ does not change the coefficients of the leading perturbative scaling. 

Using the scaling $\tau=\Theta(\epsilon)$ from
Sec.~\ref{subsec: pure scaling}, the covariance-matrix contribution
satisfies
\begin{align}
\xi_{{\rm CE},{\rm CM},s}^{(2)}
=
\Theta(\epsilon).
\end{align}

\subsection{Quantum fidelity}

We finally consider the QF. As shown in
Sec.~\ref{sec: perturbative formula}, the leading covariance-matrix
contribution is
\begin{align}
F_{\rm CM}^{(2)}
=
-\sum_{k,l}
\frac{
\Delta\widetilde E_{kl}
\Delta\widetilde E_{lk}
}{
4+16a_ka_l
}.
\label{eq:app F QF second}
\end{align}
Here, $a_k$ and $a_l$ are eigenvalues of $\Gamma_\tau$ in
\eqref{eq:app theorem gamma}.

We now identify the most singular perturbative coefficients in
\eqref{eq:app F QF second}. For $a_k$ and $a_l$ of the same sign
$\pm ir_\tau$, the denominator of \eqref{eq:app F QF second} satisfies
\begin{align}
4+16(\pm ir_\tau)^2
&=
4-16r_\tau^2
\nonumber\\
&=
-16\tau(1+\tau)
=
\Theta(\tau).
\label{eq:app F QF same}
\end{align}
Therefore, the corresponding coefficients scale as
$\Theta(\tau^{-1})$.
For eigenvalues $\pm ir_\tau$ and $\mp ir_\tau$, the denominator is
\begin{align}
4+16(ir_\tau)(-ir_\tau)
=
4+16r_\tau^2,
\end{align}
which remains nonzero as $\tau\rightarrow0$.
Furthermore, if at least one eigenvalue is of the form $\pm i\nu_j$ with
$\nu_j>1/2$, the corresponding coefficient remains nonsingular as
$\tau\rightarrow0$.

Therefore, the most singular coefficients in
\eqref{eq:app F QF second} are those multiplying the entries within
$\Delta\widetilde E_{22}$ and $\Delta\widetilde E_{44}$, i.e.,
$k,l\in[n_{\rm g}+1,n_{\rm g}+n_{\rm e}]$ and
$k,l\in[2n_{\rm g}+n_{\rm e}+1,2n_{\rm g}+2n_{\rm e}]$,
respectively.
The contribution from the $\Delta\widetilde E_{22}$ block to
\eqref{eq:app F QF second} is
\begin{align}
\sum_{k,l\in[n_{\rm g}+1,n_{\rm g}+n_{\rm e}]} -
\frac{
\Delta\widetilde E_{kl}
\Delta\widetilde E_{lk}
}{
4+16a_ka_l
}
=&
\frac{\trace \left[
(\Delta\widetilde E_{22})^2
\right]}{16\tau(1+\tau)}
\nonumber\\
=&
-\frac{
\|\Delta\widetilde E_{22}\|_F^2}{16\tau(1+\tau)},
\label{eq:app thm 5 QF 101}
\end{align}
where the last equality follows from the anti-Hermiticity of
$\Delta\widetilde E_{22}$. Similarly, the contribution from the
$\Delta\widetilde E_{44}$ block is
\begin{align}
-\frac{\|\Delta\widetilde E_{44}\|_F^2}{16\tau(1+\tau)}.
\label{eq:app thm 5 QF 102}
\end{align}
Under the condition \eqref{eq:criterion 2},
\begin{align}
\max\left\{
\|\Delta\widetilde E_{22}\|,
\|\Delta\widetilde E_{44}\|
\right\}
=
\Theta(\epsilon),
\label{eq:app F QF condition}
\end{align}
\eqref{eq:app norm relation} implies
\begin{align}
\|\Delta\widetilde E_{22}\|_F^2
+
\|\Delta\widetilde E_{44}\|_F^2
=
\Theta(\epsilon^2).
\end{align}
Combining \eqref{eq:app thm 5 QF 101} and
\eqref{eq:app thm 5 QF 102},  we have
\begin{align}
-\frac{
\|\Delta\widetilde E_{22}\|_F^2
+
\|\Delta\widetilde E_{44}\|_F^2}{16\tau(1+\tau)} = \Theta(\epsilon^2\tau^{-1}), 
\end{align}
and cannot cancel.

All remaining terms in \eqref{eq:app F QF second} are at most
$O(\epsilon^2)$ and are therefore asymptotically subleading. Hence,
the leading covariance-matrix contribution to the quantum infidelity satisfies
\begin{align}
1-F
=
\Theta(\epsilon^2\tau^{-1}).
\end{align}
Therefore, under the condition
\eqref{eq:app F QF condition}, the leading infidelity is determined by
the modes with symplectic eigenvalue $r_\tau$. Tracing out the modes whose
symplectic eigenvalues satisfy $\nu_j>1/2$ does not change its leading
asymptotic scaling.

Using the scaling $\tau=\Theta(\epsilon)$ from
Sec.~\ref{subsec: pure scaling}, we obtain
\begin{align}
1-F
=
\Theta(\epsilon).
\end{align}

\subsection{Summary}

The analysis in this section shows that the strongest singular perturbative
coefficients are associated with the modes whose symplectic eigenvalues are $r_\tau$. For the VNE, under the condition
\begin{align}
\left|
\operatorname{tr}\widetilde E_{22}
-
\operatorname{tr}\widetilde E_{44}
\right|
=
\Theta(\epsilon), \label{eq:tm5 condtion 1}
\end{align}
the leading variation scales as
\begin{align}
S_{\rm VN}^{(1)}
=
\Theta(\epsilon|\log\epsilon|).
\end{align}
For the QRE, quantum Chernoff quantity, and QF, under the condition
\begin{align}
\max\left\{
\|\Delta\widetilde E_{22}\|,
\|\Delta\widetilde E_{44}\|
\right\}
=
\Theta(\epsilon), \label{eq:tm5 condtion 2}
\end{align}
the leading covariance-matrix contributions satisfy
\begin{align}
D_{\rm RE,CM}^{(2)}
&=
\Theta(\epsilon),\\
\xi_{{\rm CE},{\rm CM},s}^{(2)}
&=
\Theta(\epsilon),\\
1-F
&=
\Theta(\epsilon).
\end{align}
Here, $\tau=\Theta(\epsilon)$ has been used.

In each case, contributions involving eigenspaces associated with symplectic eigenvalues $\nu_j>1/2$ are asymptotically subleading. Therefore, under the corresponding conditions, \eqref{eq:tm5 condtion 1} and \eqref{eq:tm5 condtion 2}, tracing out the modes whose symplectic eigenvalues satisfy $\nu_j>1/2$ does not change the coefficients of the leading perturbative scaling. We emphasize that this conclusion is restricted to perturbative contributions through second order. Higher-order contributions are not considered in this analysis.

\section{Derivation of quantum illumination results} \label{appendix:Ali}
Throughout this appendix, we consider the asymptotic regime used in \cite{AliCox2024}, in which
$\kappa^{(i)},\bar{n}\rightarrow0$ and
$\bar{n}_{\rm B}\rightarrow\infty$ such that 
\begin{align}
\kappa^{(i)}
&\in
o\left(
\frac{\bar{n}}{\bar{n}_{\rm B}}
\right),
\quad
\bar{n}
\in
O\left(
e^{-\bar{n}_{\rm B}}
\right),
\quad
1\in o(\bar{n}_{\rm B}).
\label{eq:Ali asymptotic regime}
\end{align}
\subsection{Coherent-state input}
We first consider the coherent-state input and show that, in the asymptotic regime of \eqref{eq:Ali asymptotic regime}, the leading perturbative quantum Chernoff quantity arises from the mean-vector contribution. The mean-vector difference between the two hypotheses is given by
\cite[Eq.~(6)]{AliCox2024}
\begin{align}
\delta_{C,d}
&=
\mu_C^{(1)}-\mu_C^{(2)}
\nonumber\\
&=
\sqrt{2\bar{n}}
\begin{bmatrix}
\sqrt{\kappa^{(1)}}\,C^{(1)}
-
\sqrt{\kappa^{(2)}}\,C^{(2)}
\\
0_{n\times 1}
\end{bmatrix}.
\label{eq:Ali coherent mean-vector}
\end{align}
The covariance matrix under hypotheses $k = 1,2$ is   \cite[Eq.~(5)]{AliCox2024}
\begin{align}
    V_C^{(k)} &=  \begin{bmatrix}
   N_C  &  0_{n\times n } \\
        0_{n\times n } & N_C
    \end{bmatrix} \\
    N_C&=   \left( \Bar{n}_{\rm B}+\frac{1}{2} \right)I_n - \kappa^{(k)} \Bar{n}_{\rm B} C^{(k)}C^{(k)T}.  
\end{align}
We choose the reference covariance matrix as
\begin{align}
\sigma_{C,0}
=
\begin{bmatrix}
\left(\bar{n}_{\rm B}+\frac{1}{2}\right)I_n
&
0_{n\times n}
\\
0_{n\times n}
&
\left(\bar{n}_{\rm B}+\frac{1}{2}\right)I_n
\end{bmatrix},
\end{align}
such that
\begin{align}
A_C
=
\Omega\sigma_{C,0}.
\end{align}
For this reference matrix,
\begin{align}
\sqrt{-A_C^2}
=
\left(
\bar{n}_{\rm B}+\frac{1}{2}
\right)I_{2n}.
\end{align}
Thus,
\begin{align}
\Omega g_{\rm CE}^{(1/2)}(A_C)
&=
\sigma_{C,0}
\left(\sqrt{-A_C^2}\right)^{-1}
F^{(1/2)}
\left(\sqrt{-A_C^2}\right)
\nonumber\\
&=
F^{(1/2)}
\left(
\bar{n}_{\rm B}+\frac{1}{2}
\right)
I_{2n}
\nonumber\\
&=
\left(
\sqrt{\bar{n}_{\rm B}+1}
+
\sqrt{\bar{n}_{\rm B}}
\right)^2
I_{2n}.
\end{align}
Therefore, using
\begin{align}
R_{\rm CE}^{(1/2)}(A_C)
=
2g_{\rm CE}^{(1/2)}(A_C),
\end{align}
we obtain
\begin{align}
\!\!(1/ R_{\rm CE}^{(1/2)})(A_C)  \Omega^{-1}
\!=\!
\frac{1}{2}
\left(
\sqrt{\bar{n}_{\rm B}+1}
+
\sqrt{\bar{n}_{\rm B}}
\right)^{-2}
I_{2n}.
\end{align}
The mean-vector contribution to the perturbative QCE is thus
\begin{align}
\xi_{C,d}^{(2)}
&=
\delta_{C,d}^T
(1/ R_{\rm CE}^{(1/2)})(A_C)  \Omega^{-1} 
\delta_{C,d}
\nonumber\\
&=
\bar{n}
\left(
\sqrt{\bar{n}_{\rm B}+1}
+
\sqrt{\bar{n}_{\rm B}}
\right)^{-2}
\nonumber\\
&\qquad\times
\left\|
\sqrt{\kappa^{(1)}}\,C^{(1)}
-
\sqrt{\kappa^{(2)}}\,C^{(2)}
\right\|_F^2.
\label{eq:Ali coherent exact}
\end{align}
In the high-noise regime $\bar{n}_{\rm B}\gg1$,
\begin{align}
\left(
\sqrt{\bar{n}_{\rm B}+1}
+
\sqrt{\bar{n}_{\rm B}}
\right)^{-2}
\approx
\frac{1}{4\bar{n}_{\rm B}}.
\end{align}
Therefore, \eqref{eq:Ali coherent exact} yields
\begin{align}
\xi_{C,d}^{(2)}
\simeq
\frac{\bar{n}}{4\bar{n}_{\rm B}}
\left\|
\sqrt{\kappa^{(1)}}\,C^{(1)}
-
\sqrt{\kappa^{(2)}}\,C^{(2)}
\right\|_F^2.
\label{eq:Ali coherent mean-vector asymptotic}
\end{align}
We next consider the covariance-matrix contribution. The covariance
perturbation satisfies
\begin{align}
\| V_C^{(1)} -  V_C^{(2)}\|_F
=
O\left(
\kappa^{(k)}\bar{n}_{\rm B}
\right), \ k = 1,2.
\end{align}
Under the asymptotic condition in \eqref{eq:Ali asymptotic regime} that 
$\kappa^{(k)}\bar{n}_{\rm B}
=
o(\bar{n})$ as
$ \bar{n}\rightarrow 0 $, we have
\begin{align}
\|E_1-E_2\|_F
=
o(\bar{n}).
\end{align}
Since $\sigma_{C,0}$ has symplectic eigenvalues all greater than $1/2$, the covariance-matrix contribution is quadratic in
$E_1-E_2$. It follows that
\begin{align}
\xi_{C,{\rm CM}}^{(2)}
=
o(\bar{n}^2).
\end{align}
Hence, the covariance-matrix contribution is asymptotically  negligible compared with the mean-vector
contribution, which is $O(\bar{n})$ .
Therefore, the leading perturbative quantum Chernoff quantity for the
coherent-state input is
\begin{align}
\xi_C^{(2)}
\simeq
\frac{\bar{n}}{4\bar{n}_{\rm B}}
\left\|
\sqrt{\kappa^{(1)}}\,C^{(1)}
-
\sqrt{\kappa^{(2)}}\,C^{(2)}
\right\|_F^2,
\label{eq:Ali Coh}
\end{align}
which is quadratic in $\|\delta_d\|$, and reproduces the coherent-state result obtained in \cite{AliCox2024}.

\subsection{TMSV input}
We next consider the TMSV input. In contrast to the coherent-state case, the mean vectors are zero, so the perturbative quantum Chernoff quantity is determined entirely by the covariance-matrix contribution. We therefore begin by choosing the reference matrix corresponding to the zeroth-order covariance matrix in \cite[Eq.~(7)]{AliCox2024}
\begin{align}
A_Q = \Omega \sigma_{Q,0}
&=
\Omega\,
\diag\left(
\underbrace{\bar{n}_{\rm B}+\frac{1}{2},
\dots,
\bar{n}_{\rm B}+\frac{1}{2}}_{n},
\bar{n}+\frac{1}{2},
\right.
\nonumber\\
&\qquad\left.
\underbrace{\bar{n}_{\rm B}+\frac{1}{2},
\dots,
\bar{n}_{\rm B}+\frac{1}{2}}_{n},
\bar{n}+\frac{1}{2}
\right),
\end{align}
and for $k=1,2$, write the perturbation matrices
\begin{align}
E_k
=
\Omega
\begin{bmatrix}
E_+^{(k)} & 0_{(n+1)\times(n+1)}
\\
0_{(n+1)\times(n+1)} & E_-^{(k)}
\end{bmatrix},
\end{align}
where
\begin{align}
E_\pm^{(k)} 
&=
\begin{bmatrix}
-\kappa^{(k)}
(\bar{n}_{\rm B}-\bar{n})
C^{(k)}C^{(k)T}
&
\pm
D_Q C^{(k)}
\\
\pm
D_Q  C^{(k)T}
&
0
\end{bmatrix}\\
D_Q &= \sqrt{\kappa^{(k)}\bar{n}(\bar{n}+1)}.
\end{align}
Defining the difference of the perturbation matrices as
\begin{align}
\Delta E
=
E_1-E_2,
\end{align}
we have
\begin{align}
\Delta E
=
\Omega
\begin{bmatrix}
\Delta E_+ & 0
\\
0 & \Delta E_-
\end{bmatrix}
=
\begin{bmatrix}
0 & \Delta E_-
\\
-\Delta E_+ & 0
\end{bmatrix},
\label{eq:Ali TMSV Delta E}
\end{align}
where
\begin{align}
\Delta E_\pm
&=
\begin{bmatrix}
\Delta K
&
\pm \Delta C
\\
\pm \Delta C^T
&
0
\end{bmatrix},
\label{eq:Ali TMSV Delta Epm}
\end{align}
with
\begin{align}
\Delta K
&=
(\bar{n}-\bar{n}_{\rm B})
\left[
\kappa^{(1)}C^{(1)}C^{(1)T}
-
\kappa^{(2)}C^{(2)}C^{(2)T}
\right], \label{eq:QF dk}
\\
\Delta C
&=
\sqrt{\bar{n}(\bar{n}+1)}
\left[
\sqrt{\kappa^{(1)}}\,C^{(1)}
-
\sqrt{\kappa^{(2)}}\,C^{(2)}
\right]. \label{eq:QF dc}
\end{align}
Since the symplectic eigenvalues of $\sigma_{Q,0}$ are all greater than $1/2$, we use the perturbative formula for QCE in Theorem \ref{thm:QCE-perturbation}.
At $s=\frac{1}{2}$, using $A_Q^{-1}\sqrt{-A_Q^2} = A_Q^{-1} \sigma_{Q,0}
=
\Omega^{-1}$, we have
\begin{align}
\left(1/R_{\rm CE}^{(1/2)}\right) (A_Q)
&=
\frac{1}{2}\left(1/g_{\rm CE}^{(1/2)}\right) (A_Q)
\nonumber\\
&=
\frac{1}{2}
\Omega
\left(1/F^{(1/2)}\right) 
\left(\sqrt{-A_Q^2}\right). \label{eq:QF R 1/2}
\end{align}
Since $ \bar{n}_{\rm B} \gg 1$ and $\bar{n} \ll 1 $, we have
\begin{align}
\left(1/F^{(1/2)}\right)
\left(
\bar{n}_{\rm B}+\frac{1}{2}
\right)
&=
\left(
\sqrt{\bar{n}_{\rm B}+1}
+
\sqrt{\bar{n}_{\rm B}}
\right)^{-2}
\nonumber\\
&\approx
\frac{1}{4\bar{n}_{\rm B}},
\\
\left(1/F^{(1/2)}\right) 
\left(
\bar{n}+\frac{1}{2}
\right)
&=
\left(
\sqrt{\bar{n}+1}
+
\sqrt{\bar{n}}
\right)^{-2}
\nonumber\\
&\approx
1,
\end{align}
thus, \eqref{eq:QF R 1/2} yields
\begin{align}
\left(1/R_{\rm CE}^{(1/2)}\right) (A_Q)
&=
\frac{1}{2}
\Omega
\begin{bmatrix}
Q & 0
\\
0 & Q
\end{bmatrix},
\label{eq:Ali R inverse}
\end{align}
where
\begin{align}
Q
=
\diag\left(
\underbrace{
\frac{1}{4\bar{n}_{\rm B}},
\dots,
\frac{1}{4\bar{n}_{\rm B}}
}_{n},
1
\right).
\end{align}
We next evaluate the Fr\'echet derivative
$D^{(1)}_{g_{\rm CE}^{(1/2)}(A_Q)}[\Delta E]$.
Let
\begin{align}
A_Q
=
W\Lambda W^\dagger,
\end{align}
where $W = \frac{1}{\sqrt{2}} \begin{bmatrix}
    I_{n+1} &I_{n+1}\\
    iI_{n+1} &-iI_{n+1}\\
\end{bmatrix}$, and 
\begin{align}
\Lambda
&=
i \diag\left(
 \underbrace{
 \left(\bar{n}_{\rm B}+\frac{1}{2}\right),
\dots,
 \left(\bar{n}_{\rm B}+\frac{1}{2}\right)
}_{n},
 \left(\bar{n}+\frac{1}{2}\right),\right.
\nonumber\\
&\left.\underbrace{
- \left(\bar{n}_{\rm B}+\frac{1}{2}\right),
\dots,
- \left(\bar{n}_{\rm B}+\frac{1}{2}\right)
}_{n},
- \left(\bar{n}+\frac{1}{2}\right)
\right).
\label{eq:Ali Lambda}
\end{align}
In this eigenbasis, using \eqref{eq:First Frechet derivative Hadmard product}, we have
\begin{align}
D^{(1)}_{g_{\rm CE}^{(1/2)}(A_Q)}[\Delta E]
=
W
\left[
\left(g_{\rm CE}^{(1/2)}\right)^{[1]}[\Lambda]
\circ
\left(W^\dagger\Delta E W\right)
\right]
W^\dagger.
\label{eq:Ali Dg eig}
\end{align}
The divided-difference matrix can be written in block form as
\begin{align}
\left(g_{\rm CE}^{(1/2)}\right)^{[1]}[\Lambda]
=
\begin{bmatrix}
g_{+} & g_{-}
\\
g_{-} & g_{+}
\end{bmatrix},
\end{align}
where
\begin{align}
g_{+}
&=
\begin{bmatrix}
g_1\,\mathbbm{1}_{n\times n}
&
*
\\
*
&
*
\end{bmatrix},
\label{eq:Ali gpp}
\\
g_{-}
&=
\begin{bmatrix}
*
&
g_2\,\mathbbm{1}_{n\times 1}
\\
g_2\,\mathbbm{1}_{1\times n}
&
*
\end{bmatrix},
\label{eq:Ali gpm}
\end{align}
and the coefficients are obtained by the definition of the first divided difference in \eqref{eq:first_dd}
\begin{align}
g_1
&=
-
\frac{
\left(
\sqrt{\bar{n}_{\rm B}+1}
+
\sqrt{\bar{n}_{\rm B}}
\right)^2
}{
\sqrt{
\bar{n}_{\rm B}
(\bar{n}_{\rm B}+1)
}
},\label{eq:Ali g1}\\
g_2
&=
\frac{
2
\left(
\sqrt{\bar{n}}+\sqrt{\bar{n}+1}
\right)
\left(
\sqrt{\bar{n}_{\rm B}}+\sqrt{\bar{n}_{\rm B}+1}
\right)
}{
1+\bar{n}+\bar{n}_{\rm B}
}
\nonumber\\
&\qquad\times
\left[
\sqrt{\bar{n}\bar{n}_{\rm B}}
-
\sqrt{
(\bar{n}_{\rm B}+1)(\bar{n}+1)
}
\right].
\label{eq:Ali g4}
\end{align}
Here, $*$ denotes the trivial submatrices.
The perturbation in the eigenbasis of $A_Q$ is
\begin{align}
W^\dagger\Delta E W
=
\frac{1}{2}
\begin{bmatrix}
i P
&
iT
\\
-iT
&
-iP
\end{bmatrix}
\label{eq:Ali W DeltaE W}
\end{align}
where
\begin{align}
    P &= \Delta E_+ + \Delta E_- \\
    T &= \Delta E_+ - \Delta E_-.
\end{align}
Substituting \eqref{eq:Ali W DeltaE W} into
\eqref{eq:Ali Dg eig}, we obtain
\begin{align}
D^{(1)}_{g_{\rm CE}^{(1/2)}(A_Q)}[\Delta E]
&=
\frac{i}{2}
W
\begin{bmatrix}
P\circ g_{+}
&
T\circ g_{-}
\\
-T\circ g_{-}
&
-P\circ g_{+}
\end{bmatrix}
W^\dagger.
\label{eq:Ali Dg final}
\end{align}
Substituting \eqref{eq:Ali R inverse} and
\eqref{eq:Ali Dg final} into the second-order perturbative quantum
Chernoff exponent in \eqref{eq:QCE perturb CM}, we obtain
\begin{align}
 \xi_{Q,{\rm CM}}^{(2)} 
&=
\frac{1}{4}
\trace\left[
\left(1/R_{\rm CE}^{(1/2)}\right)(A_Q)
D^{(1)}_{g_{\rm CE}^{(1/2)}(A_Q)}[\Delta E]
\right.
\nonumber\\
&\qquad\qquad\left.
\times
\left(1/R_{\rm CE}^{(1/2)}\right)(A_Q)
D^{(1)}_{g_{\rm CE}^{(1/2)}(A_Q)}[\Delta E]
\right]
\nonumber\\
&=
\frac{1}{32}
\!\trace\!\left[
\left(Q
\left(
P\circ g_{+}
\right)\right)^2
+
\left(Q
\left(
T\circ g_{-}
\right)
\right)^2
\right].
\label{eq:Ali TMSV trace}
\end{align}
Using \eqref{eq:Ali TMSV Delta Epm}, we have
\begin{align}
P
=
2
\begin{bmatrix}
\Delta K & 0
\\
0 & 0
\end{bmatrix}, \ 
T
&=
2
\begin{bmatrix}
0 & \Delta C
\\
\Delta C^T & 0
\end{bmatrix},
\end{align}
Thus
\begin{align}
P\circ g_{+}
&=
2g_1
\begin{bmatrix}
\Delta K & 0
\\
0 & 0
\end{bmatrix},  \label{eq:QF P gp}
\\
T\circ g_{-}
&=
2g_2
\begin{bmatrix}
0 & \Delta C
\\
\Delta C^T & 0
\end{bmatrix}. \label{eq:QF T gm}
\end{align}
Substituting \eqref{eq:QF P gp} and \eqref{eq:QF T gm} into \eqref{eq:Ali TMSV trace}, we have
\begin{align}
& \xi_{Q,{\rm CM}}^{(2)} \\
&=
\frac{1}{8}
\trace\Bigg[
g_1^2
Q
\begin{bmatrix}
\Delta K & 0
\\
0 & 0
\end{bmatrix}
Q
\begin{bmatrix}
\Delta K & 0
\\
0 & 0
\end{bmatrix}
\nonumber\\
&\qquad
+
g_2^2
Q
\begin{bmatrix}
0 & \Delta C
\\
\Delta C^T & 0
\end{bmatrix}
Q
\begin{bmatrix}
0 & \Delta C
\\
\Delta C^T & 0
\end{bmatrix}
\Bigg]\\
&=
\frac{g_1^2}{8}
\trace\left[
\left(
\frac{\Delta K}{4\bar{n}_{\rm B}}
\right)^2
\right]
+
\frac{g_2^2}{16\bar{n}_{\rm B}}
\|\Delta C\|^2 \label{eq:Ali second} \\
&=
\frac{g_1^2}{8}
\trace\left[
\left( \frac{\bar{n}_{\rm B}-\bar{n}}{4\bar{n}_{\rm B}} 
\left[
\kappa^{(1)}C^{(1)}C^{(1)T} \right.\right.\right.\nonumber\\
&\left.\left.\left. \qquad\qquad\qquad\qquad\qquad
-
\kappa^{(2)}C^{(2)}C^{(2)T}
\right]
\right)^2
\right]
\nonumber\\
&\quad
+
\frac{
g_2^2\bar{n}(\bar{n}+1)
}{
16\bar{n}_{\rm B}
}
\left\|
\sqrt{\kappa^{(1)}}\,C^{(1)}
-
\sqrt{\kappa^{(2)}}\,C^{(2)}
\right\|_F^2
\label{eq:Ali second explicit}
\end{align}
where \eqref{eq:Ali second explicit} follows from the definition of $\Delta K$ in \eqref{eq:QF dk} and $\Delta C$ in \eqref{eq:QF dc}.
Under \eqref{eq:Ali asymptotic regime},
$\kappa^{(k)}\bar{n}_{\rm B}=o(\bar{n})$, so that the first term in
\eqref{eq:Ali second explicit} is $o(\bar{n}^2)$ and is asymptotically negligible. Moreover,
since $\bar{n}(\bar{n}+1)\approx \bar{n}$, the second term in \eqref{eq:Ali second explicit} yields
\begin{align}
\xi_{Q,{\rm CM}}^{(2)}
\approx
\frac{
\bar{n}g_2^2
}{
16\bar{n}_{\rm B}
}
\left\|
\sqrt{\kappa^{(1)}}\,C^{(1)}
-
\sqrt{\kappa^{(2)}}\,C^{(2)}
\right\|_F^2
\label{eq:Ali last}
\end{align}
Finally,
\begin{align}
\lim_{\bar{n}_{\rm B}\rightarrow\infty}
g_2^2
=
16,
\end{align}
which yields
\begin{align}
\xi_Q^{(2)}  = \xi_{Q,{\rm CM}}^{(2)}
\approx
\frac{\bar{n}}{\bar{n}_{\rm B}}
\left\|
\sqrt{\kappa^{(1)}}\,C^{(1)}
-
\sqrt{\kappa^{(2)}}\,C^{(2)}
\right\|_F^2.
\label{eq:Ali TMSV}
\end{align}
Thus, we recover the TMSV result obtained in \cite{AliCox2024}.
Combining \eqref{eq:Ali Coh} and \eqref{eq:Ali TMSV}, we obtain
\begin{align}
\xi_Q^{(2)}
=
4\xi_C^{(2)},
\end{align}
showing the factor-of-four enhancement in the quantum Chernoff
exponent achieved by the TMSV input in this asymptotic regime.

Finally, we show that the perturbative QCE retains the
quadratic scaling with the covariance-matrix perturbation.
From \eqref{eq:Ali TMSV Delta E} and \eqref{eq:Ali TMSV Delta Epm}, we have
\begin{equation}
    \|\Delta E\|_F^2
    = 2\|\Delta K\|_F^2 + 4\|\Delta C\|_F^2.
\end{equation}
Since $g_1$ and $g_2$ in \eqref{eq:Ali g1} and \eqref{eq:Ali g4} remain bounded
in the asymptotic regime \eqref{eq:limit regime}, \eqref{eq:Ali second} gives
\begin{align}
    \xi_Q^{(2)}
    &= \frac{g_1^2}{128\bar n_B^2}\|\Delta K\|_F^2
    + \frac{g_2^2}{64\bar n_B}\|\Delta C\|_F^2
    \nonumber\\
    &= O\left(\|\Delta K\|_F^2
    + \|\Delta C\|_F^2\right)
    \nonumber\\
    &= O\left(\|\Delta E\|_F^2\right) \nonumber \\
    &= O\left(\|\Delta E\|^2\right),
\end{align}
where the last equality follows from the equivalence of
the Frobenius and operator norms in fixed finite dimensions.
Thus, the quantum-illumination problem retains the
usual quadratic perturbative scaling for a reference
state whose symplectic eigenvalues are all greater than $1/2$.

\end{document}